\documentclass[a4paper,11pt,reqno]{amsart}
\usepackage[utf8]{inputenc}
\usepackage[T1]{fontenc} 

\usepackage{amsthm,amsmath,amsfonts,amssymb,amsxtra,dsfont,bm,mathrsfs,amstext,amsopn,mathrsfs,mathtools,esint}

\usepackage[colorlinks, linkcolor={blue}, citecolor={red}]{hyperref}

\newtheorem{theorem}{Theorem}[section]\newtheorem{lemma}[theorem]{Lemma}\newtheorem{remark}[theorem]{Remark}\newtheorem{corollary}[theorem]{Corollary}\newtheorem{proposition}[theorem]{Proposition}

    \newcommand\cD{\mathcal{D}}          \newcommand\cS{\mathcal{S}}  \newcommand\cI{\mathcal{I}}\newcommand\cN{\mathcal{N}}\newcommand\cM{\mathcal{M}}\newcommand\cK{\mathcal{K}}\newcommand\cJ{\mathcal{J}}

\DeclareMathOperator{\tr}{Tr}\DeclareMathOperator{\Ker}{Ker}\DeclareMathOperator{\re}{Re}\DeclareMathOperator{\dist}{dist}\let\div\relax\DeclareMathOperator{\div}{div}
\def\d{{\rm d}}

\let\C\relax\newcommand{\C}{\mathbb{C}}\newcommand{\Z}{\mathbb{Z}}\newcommand{\R}{\mathbb{R}}\newcommand{\N}{\mathbb{N}}\newcommand{\Q}{\mathbb{Q}}

\renewcommand{\ge}{\geqslant}\renewcommand{\le}{\leqslant}

\newcommand\cC{\mathcal{C}}
\newcommand{\ld}{L^2_{\textup{per}}} 
\newcommand{\nor}[2]{ \left\| #1 \right\|_{#2} } 

\newcommand{\pa}[1]{\left( #1 \right)} 
\newcommand{\acs}[1]{\left\{ #1 \right\}} 
\newcommand{\ab}[1]{\left|#1\right|} 
\newcommand{\ps}[1]{\left< #1 \right>} 

\newcommand\vp{\varphi} 
\def\eps{\varepsilon}\newcommand{\ep}{\varepsilon} 
\let\p\relax\newcommand{\p}{\psi} 
\newcommand{\na}{\nabla} 

\newcommand{\f}[2]{\frac{#1}{#2}}

\newcommand{\mymax}[1]{\underset{\substack{#1}}{\text{\normalfont{max}}}\quad}

\newcommand{\nak}{\nabla_{x,k}}

\def\L{\mathbb L}

\def\bH{{\bold H}}
\def\bS{{\bold S}}

\newcommand{\1}{\mathds 1}
\newcommand{\sC}{\mathscr{C}}

\newcommand{\au}{\gamma}
\newcommand{\cu}{\chi_1}
\newcommand{\cd}{\chi_2}

\newcommand{\cdep}{\pa{\chi_2}_\ep}
\newcommand{\Au}{A_2 }

\newcommand{\Auep}{\pa{A_2 }_\ep}

\newcommand{\Mz}{\widetilde\cM^{\ep,(2)}_{z,k}}
\newcommand{\Md}{\cM^{\ep,(2)}_{k}}
\newcommand{\Mzu}{\cM^{\ep,(1)}}
\newcommand{\Mze}{\widetilde\cM^{\ep,(2)}_{E^0_{\ell,k},k}}
\newcommand{\pE}{\p^\ep_{\ell,k}}
\newcommand{\pz}{\p^0_{\ell,k}}

\newcommand{\per}{{\rm per}}

\newcommand{\dps}{\displaystyle }

\renewcommand{\paragraph}[1]{\medskip {\bf #1}}

\title[Homogenization of periodic Schrödinger operators]{Second-order homogenization \\ of periodic Schrödinger operators \\ with highly oscillating potentials}

\author{Éric Cancès}
\address[Éric Cancès]{CERMICS, \'Ecole des Ponts and Inria Paris, 6 and 8 av. Pascal, 77455 Marne-la-Vallée, France} 
\email{eric.cances@enpc.fr}

\author{Louis Garrigue}
\address[Louis Garrigue]{CERMICS, \'Ecole des Ponts and Inria Paris, 6 and 8 av. Pascal, 77455 Marne-la-Vallée, France} 
\email{louis.garrigue@enpc.fr}

\author{David Gontier}
\address[David Gontier]{CEREMADE, University of Paris-Dauphine, PSL University, 75016 Paris, France \& ENS/PSL University, Département de Mathématiques et Applications, F-75005, Paris, France.
}
\email{gontier@ceremade.dauphine.fr}

\date{\today}

\begin{document} 
\maketitle

\begin{abstract}
We consider the homogenization at second-order in $\eps$ of $\L$-periodic Schr\"odinger operators with rapidly oscillating potentials of the form $H^\eps =-\Delta + \eps^{-1} v(x,\eps^{-1}x ) + W(x)$ on $L^2(\R^d)$, where $\L$ is a Bravais lattice of $\R^d$, $v$ is $(\L \times \L)$-periodic, $W$ is $\L$-periodic, and $\eps \in \N^{-1}$. We treat both the linear equation with fixed right-hand side and the eigenvalue problem, as well as the case of physical observables such as the integrated density of states. We illustrate numerically that these corrections to the homogenized solution can significantly improve the first-order ones, even when $\ep$ is not small.
\end{abstract}


\section{Introduction}

The homogenization of Schr\"odinger operators of the form
\begin{equation}\label{eq:main_op}
H^\eps = -\Delta + \eps^{-1} v\pa{x,\eps^{-1} x} + W(x) \quad \mbox{on } L^2(\cD),
 \end{equation}
 where $\cD$ is a domain of $\R^d$, $W \in L^\infty(\cD;\R)$, $\ep > 0$, and the potential $v \in \cC^0(\cD \times \R^d;\R)$ is periodic with zero mean in its second variable, has an already long history. To our knowledge, such operators have been first studied in the reference textbook \cite[Section 12]{BenLioPap78} in the case when $\cD$ is a bounded domain and $H^\varepsilon$ has form domain $H^1_0(\Omega)$ (Dirichlet boundary conditions).
 The authors identified a homogenized Schrödinger operator of the form
 \[
    H^0 := - \Delta + V_0 + W \quad \text{on} \quad L^2(\cD),
 \]
and a corrector $\chi_1 : \R^d \to \R$, and showed that, for $\mu < \min(\sigma(H^0))$, 
 \[
 \lim_{\eps \to 0} \, {\nor{u^\ep_{\mu,f} - \pa{1 + \eps \cu \pa{\ep^{-1} \cdot}}  u^0_{\mu,f}}{H^1(\cD)} = 0},
 \]
 where the functions $u_{\mu,f}^\eps$ and  $u_{\mu,f}^0$ are respectively the solutions to the equations $(H^\eps-\mu)u_{\mu,f}^\eps=f$ and $(H^0-\mu)u_{\mu,f}^0=f$ with Dirichlet boundary conditions.
 They also proved error bounds on the ground-state eigenvalue: ${\ab{E_1^\eps - E_1^0} \le c \ep}$, where $E_\ell^\eps$ and $E_\ell^0$ denote the lowest $\ell^{\rm th}$ eigenvalue of $H^\eps$ and~$H^0$, respectively, counting multiplicities. 
 Then in \cite{Zhang21}, using techniques of \cite{KenLinShe13}, Zhang proved that there is a constant $c \in \R_+$ independent of $\eps$ such that
 \[
 \nor{u^\ep_{\mu,f} - \pa{1 + \eps \cu\pa{\ep^{-1} \cdot}}  u^0_{\mu,f}}{H^1(\cD)} \le c \ep, 
 \quad  \quad
 \ab{E_{\ell}^\eps - E_{\ell}^0} \le c \pa{E^0_{\ell}-\mu}^{\f 32} \ep.
 \]
 
 The present article is a contribution to the study of Schr\"odinger operators of the form \eqref{eq:main_op}, whose originality is threefold. First, we consider a fully periodic setting, that is $\cD=\R^d$, $v$ is $\L \times \L$-periodic, $W$ is $\L$-periodic, and $\eps \in \N^{-1}$, where $\L$ is a given Bravais lattice of $\R^d$.  Our results can be extended to some point to rational values of $\eps$ (see Remark~\ref{rem:rational} and Section~\ref{sec:numerical_res}); on the other hand, the irrational case corresponding to incommensurate potentials remains out of the scope of this study. As will be clarified in a future work, Schr\"odinger operators similar to these ones are encountered in the study of the electronic structure of emerging moiré materials~\cite{Andrei_2021}. Possible applications to cold-atom lattices or optical lattices could also be considered. Second, we provide $H^1$-approximations of eigenmodes of order $\eps^2$, and approximations of eigenvalues up to order $\eps^4$. Third, we use a slightly different strategy of proof based on resolvent estimates.
 
 \medskip
 
Let us briefly explain our strategy in the simple case of first-order approximations. When $\eps \in \N^{-1}$, both operators $H^\eps$ and $H^0$ commute with $\L$-translations, and one can consider their Bloch transforms~\cite[Chapter XIII]{ReeSim4}. For $k$ in the first Brillouin zone $\Omega^*$, we denote their respective Bloch fibers by $H^\eps_k$ and $H^0_k$. First, we identify a first order corrector $\cM^{\eps, (1)}$ so that, for all $\eps \in \N^{-1}$, and all $z \in \C$,
 \[
    \left\| \left( H^\eps_k - z \right) \cM^{\eps, (1)} - \left( H^0_k - z \right) \right\|_{H^2_{\per} \to H^{-1}_\per} \le C (1 + | z |) \eps,
 \]
 for a constant $C \ge 0$ independent of $z \in \C$, $k \in \Omega^*$ and $\eps \in \N^{-1}$ (see Lemma~\ref{th:key_th}). Note that the operators are seen as bounded operators from the $H^2_\per$ to the $H^{-1}_\per$ Sobolev spaces, instead of the {\em usual} $H^2_\per$ and $L^2_\per$ ones (in some sense, we loose one derivative because of the highly-oscillatory potential). Next, we transform the previous inequality into a resolvent estimate, of the form
 \[
    \nor{\pa{H^\ep_k -z}^{-1} - \Mzu \pa{H^0_k -z}^{-1}}{L^2_\per \rightarrow H^1_\per} \le C(k,z) \ep,
 \]
 see Theorem~\ref{thm:res}. From this estimate and usual functional calculus, we derive {\em a priori} estimates on the Bloch functions of $H^\eps$ (Section~\ref{sec:homog_eig}) as well as on several quantities of interest (QoI) such as the kinetic and potential energies of the fermionic ground states, and the integrated density of states (IDOS) of the operator $H^\eps$ (Section~\ref{subsec:phys_quantities}).

\medskip

 We believe that our technique could be used to address other cases. Let us mention for instance the analysis of the scattering and interface properties of Schr\"odinger operators with highly-oscillatory potentials~\cite{DucWei11}, or the homogenization of operators in divergence form $- \div A(x,\f x\ep) \na$, which have been thoroughly studied in~\cite{MosVog97,Zhuge20,Kesavan79,SanVog93} for first order corrections, and in~ \cite{AllAma99,ConOriVan08,BelFehFis17} for second order ones (see also \cite{AllHab13,HawAhm21} for other models).

\medskip

The article is organized as follows. The main theoretical results are presented in Section~\ref{sec:main_result}, and are illustrated numerically in Section~\ref{sec:numerical_res}. We show in particular that our second-order corrector greatly improve the first-order one, even in a regime where $\ep$ is not so small.
Interestingly, the Ansatz for second-order homogenization given by the formal two-scale expansion (Section~\ref{sec:formal_expansions})  turns out to be unsuitable for some applications. We propose a variant of it, better suited to e.g. compute spectral projectors and related physical properties, or deal with degenerate eigenvalues. 
All the proofs are postponed until Section~\ref{sec:proofs}.

\medskip


\subsection*{Acknowledgement}
This project has received funding from the European Research Council (ERC) under the European Union's Horizon 2020 research and innovation programme (grant agreement EMC2 No 810367) and from the Simons foundation (Targeted Grant Moir\'e Materials Magic).

\section{Main results}\label{sec:main_result}

We introduce some notation in Section~\ref{sec:notation}, and define first- and second-order corrector operators in Section~\ref{sec:def_correctors}. We then state in Section~\ref{sec:resolvent} estimates on the error between the resolvent $(H^\eps_k-z)^{-1}$ of the Bloch fiber $H^\eps_k$ of $H^\eps$ and computable first and second-order approximations of it given by homogenization theory. These estimates provide error bounds on the solution to the linear equation $(H^\eps_k-z)u^\eps=f$. They also pave the way to the derivation of error estimates on eigenvalues and eigenfunctions (Section~\ref{sec:homog_eig}) via Cauchy's residue formula. Error bounds on physical quantities of interest are given in Section~\ref{subsec:phys_quantities}.

\subsection{Notation and mathematical framework}
\label{sec:notation}

We consider a Bravais lattice $\L= \Z e_1 + \cdots + \Z e_d$ of $\R^d$ with unit cell 
$\Omega = (0,1) e_1 + \cdots + (0,1) e_d$, and denote by $\L^*$ the dual lattice of $\L$ and $\Omega^*$ the first Brillouin zone. 
We also denote by
\begin{align*}
	& \qquad\qquad\qquad L^p_{\rm per}(\Omega) := \{ f \in L^p_{\rm loc}(\R^d) \; | \; f \mbox{ $\L$-per}\}, \\
	&W^{s,p}_\per(\Omega) :=  \{ f \in W^{s,p}_{\rm loc}(\R^d) \; | \; f \mbox{ $\L$-per}\}, \qquad H^s_{\rm per}(\Omega):=W^{s,2}_{\rm per}(\Omega),
\end{align*}
the $\L$-periodic Lebesgue and Sobolev spaces for all $1 \le p \le \infty$ and $s \in \R$. We endow $H_\per^s(\Omega)$ with its natural inner product
\[
\langle f, g \rangle_{H^s_\per} := \sum_{n \in \L^*} \left( 1 + | n |^2 \right)^s \overline{c_n(f)} c_n(g),
\]
where, for $n \in \L^*$, $c_n(f) :=  | \Omega |^{-1/2} \int_\Omega f(x) e^{ - i n \cdot x} \d x$ is the $n$-th  Fourier coefficient of the $\L$-periodic distribution $f$.

\medskip

We also introduce the following Banach spaces of $\L \times \L$-periodic functions: for all $m,n \in \N$,
\begin{align*}
	\cC^0_\per(\Omega;\cC^0_\per(\Omega))&:= \left\{ f \in \cC^0(\R^d \times \R^d) \; \big| \; f \mbox{ $\L \times \L$-periodic} \right\}, 
\\
\cC^m_\per(\Omega;\cC^n_\per(\Omega))&:= \bigg\{ f \in \cC^0_\per(\Omega;\cC^0_\per(\Omega)) \; \big| \; \\
\forall \alpha,&\beta \in \N^d \; \mbox{s.t.} \; |\alpha| \le m,  \; |\beta| \le n, \; \frac{\partial^{|\alpha|+|\beta|} g}{\partial x^\alpha \partial^\beta y} \in \cC^0(\R^d \times \R^d)  \bigg\}, 
\\
	\cC^m_\per(\Omega;\cC^n_{\per,0}(\Omega))&:= \left\{ f \in \cC^m_\per(\Omega;\cC^n_\per(\Omega)) \; \big| \; \fint_\Omega f(\cdot ,y) \, dy = 0 \right\}, 
\end{align*}
where $\fint_\Omega:=\frac{1}{|\Omega|}\int_\Omega$. We denote the two variables of $f \in \cC^0_{\rm per}(\Omega;\cC^0_\per(\Omega))$ by $x$ and $y$ respectively, and call $x$ the macroscopic variable and $y$ the microscopic one. We endow $\cC^m_\per(\Omega;\cC^n_\per(\Omega))$ with the norm
 \begin{align*}
\|g\|_{\cC^m_\per(\cC^n_\per)} := \mymax{\alpha \in \N^d \\|\alpha| \le m} \mymax{\beta \in \N^d \\|\alpha| \le n}  \mymax{(x,y) \in \R^d \times \R^d}  \left| \frac{\partial^{|\alpha|+|\beta|} g}{\partial x^\alpha \partial y^\beta} (x,y) \right|. 
 \end{align*}

\medskip

We are interested in studying the Schr\"odinger operator
\begin{equation} \label{eq:cH}
\boxed{ H^\ep= -\Delta + \eps^{-1} v \left( x, \eps^{-1} x  \right) + W(x) 
\quad \mbox{on } L^2(\R^d),}
\end{equation}
where $\eps > 0$ and 
\begin{align}\label{hypo:on_v_W}
	 v \in \cC^2_\per(\Omega,\cC^0_{\per,0}(\Omega)), \qquad W \in W^{1,\infty}_\per\pa{\Omega}.
 \end{align}
These smoothness conditions may probably be weakened, but we do not explore this question. Note that some smoothness properties are required for the function $v$ so that $x \mapsto \eps^{-1} v(x, \eps^{-1} x)$ is a well-defined function (this may not hold for instance if $v$ is only defined {\em almost everywhere}).

\medskip

We will repeatedly use the next Lemma, which is a straightforward consequence of the continuous embeddings $\cC^0_\per(\Omega) \hookrightarrow L^p_\per(\Omega)$ and $W^{2,p}(\Omega) \hookrightarrow \cC^1_\per(\Omega)$ for all $1 \le p < \infty$, together with elliptic regularity.

\begin{lemma} \label{lem:regularity_assumption}
    Let $g \in \cC^m_\per(\Omega ; \cC^0_\per(\Omega))$, and set $M_g(x) := \fint_\Omega g(x,y) \, dy$. The function $g - M_g$ is in $\cC^m_\per(\Omega ; \cC^0_{\per,0}(\Omega))$, and there is a unique solution $h_g$ in $ \cC^m_\per(\Omega ; \cC^1_\per(\Omega))$ to the family of elliptic problems
    \[
        (- \Delta_y h_g)(x,y) = (g - M_g)(x,y).
    \]
    In addition, there is $C \in \R_+$ independent of $g$ so that
    \begin{equation}\label{eq:elliptic_reg}
        \|h_g\|_{\cC^m_\per(\cC^1_\per)} \le C_m \| g \|_{\cC^m_\per(\cC^0_\per)}.
    \end{equation}
\end{lemma}

\subsection{Definition of the correctors} \label{sec:def_correctors}
From now on, we use the notation
\begin{align*}
\boxed{f_\ep(x) := f\pa{x,\eps^{-1} x}.}
\end{align*}
The study of the operator $H^\eps$ in~\eqref{eq:main_op} is numerically difficult, due to the large and fast oscillating potential $\eps^{-1} v_\eps$. In order to define the homogenized operator, we introduce the function $\cu \in \cC^2_\per(\Omega;\cC^1_{\per,0}(\Omega))$ solution to
\begin{equation} \label{eq:def_chi}
    \Delta_y \cu(x,y) = v(x,y).
\end{equation}
The function $\cu$ is involved in the definition of the first-order corrector of the homogenized problem and will play a key role in the sequel. We also define the homogenized (macroscopic) potential $V \in \cC^2_{\rm per}(\Omega)$ by
\begin{equation} \label{eq:def_V}
    V(x) := \fint_\Omega v(x,y) \cu(x,y) \, \d y = - \fint_\Omega |\nabla_y\cu(x,y)|^2 \, \d y.
\end{equation}
The above formulae for $V_0$ and $\cu$ result from a formal two-scale expansion presented in Section~\ref{subsec:TSE_eigenvalue}. As will be shown later, and as expected from the results in \cite{BenLioPap78, Zhang21}, the (first-order) homogenized Schr\"odinger operator is given by
\begin{equation} \label{eq:cH0}
    \boxed{ H^0 := -\Delta +V_0(x)  \quad \text{on } L^2(\R^d), \quad \mbox{with} \quad  V_0(x):=V(x)+W(x)}
\end{equation}

Since $V_0$ and $W$ are $\L$-periodic, $H^0$ is decomposed by the Bloch transform (see e.g.~\cite[Chapter XIII]{ReeSim4}): we denote by $\nabla_k:=\nabla+ik$,  and
\begin{equation}\label{eq:Bloch_dec_of_H0}
    H_k^0 := (-i\nabla_k)^2+ V_0 = - \nabla_k^2 + V_0 = -\Delta - 2ik \cdot \nabla + |k|^2 + V_0,
\end{equation} 
the Bloch fiber of $H^0$ at quasi-momentum $k \in \R^d$. The operator $H^0_k$ is self-adjoint on $L^2_{\rm per}(\Omega)$ with domain $H^2_{\rm per}(\Omega)$. It is bounded from below and has compact resolvent. 

\medskip

In general, the operator $H^\eps$ is not $\L$-periodic. However, it is $\L$-periodic if $\eps \in \N^{-1}$, in which case we can also define its Bloch decomposition:
\begin{equation}\label{eq:Bloch_dec_of_H}
    H^\eps_k: = - \nabla_k^2 + \eps^{-1} v_\eps +W, \quad \mbox{ for } \eps \in \N^{-1}.
\end{equation} 

\begin{remark} [Supercells]
\label{rem:rational}
When $\eps = p/q \in \Q$ is rational, the operators $H^0$ and $H^\eps$ are both $p\L$-periodic. The results below allow one to study the limit $\eps \to 0$ for a family of rational numbers with uniformly bounded numerators. The constants appearing in our estimates depend {\em a priori} on $p$, the numerical simulations in Section~\ref{sec:numerical_res} suggest that they do not however.
\end{remark}

The operator $H^\eps_k$ is difficult to handle numerically for small values of $\ep$ by brute force approaches because of the highly-oscillatory, large magnitude, potential $\ep^{-1} v_\ep$. Our goal is to approach the eigenmodes of $H^\eps_k$ by post-treatment of the ones of the homogenized $H^0_k$ (or of a perturbation of $H^0_k$ {\em at the macroscopic scale}, see below). The corresponding expressions are first inferred formally by a two-scale expansion, and then justified by rigorous error bounds. 


We already defined the function $\chi_1$ in~\eqref{eq:def_chi}. 
We now introduce the first-order corrector as the multiplication operator $\cM^{\eps,(1)}$ on $L^2_{\rm per}(\Omega)$ defined by
\begin{equation}\label{eq:M1}
    \boxed{ \Mzu := 1 + \eps (\chi_1)_\eps.}
\end{equation}
In order to define our second order correctors, we first define the macroscopic potential $V_1 \in \cC^1_{\rm per}(\Omega)$ by
\begin{equation} \label{eq:def_V1}
     V_1(x) := 2  \fint_\Omega \left( \nabla_x\cu(x,y) \cdot \nabla_y\cu(x,y) - \cu(x,y)|\nabla_y\cu(x,y)|^2 \right) \, \d y.
\end{equation}
We also define the two-scale vector field $A_2  \in \left( \cC^2_\per(\Omega;\cC^1_{\per,0}(\Omega)) \right)^d$ solution to 
\begin{equation} \label{eq:def_A}
    \Delta_y A_2 (x,y) := - 2 \left( \nabla_y \chi_1 \right) (x,y), \quad   \mbox{for a.e. } x \in \R^d, 
\end{equation}
and the second-order corrector $\chi_2 \in  \cC^1_\per(\Omega;\cC^1_{\per,0}(\Omega))$ as the solution to
\begin{equation} \label{eq:def_chi2}
    \Delta_y \chi_2(x,y) :=  \left[ \chi_1(x,y) v(x,y) - V(x) \right]  - 2 [(\nabla_x \cdot \nabla_y) \chi_1] (x,y).
\end{equation}
Whereas there is only one natural way to define a homogenized operator and an associated corrector for first-order homogenization, there are (at least) two natural ways to do that at second order. The first way is to use the pair $(H^0_k,\Mz)$ suggested by the formal two-scale expansion (see Section~\ref{subsec:TSE_eigenvalue}); the second one is to use the pair $(H^0_k+\eps V_1,\Md)$ where $\Md$ is defined as follows:
\begin{equation} \label{eq:def_M2}
    \boxed{ \Md := 1 + \eps (\chi_1)_\eps
    + \ep^2 \left( \pa{\chi_2}_\eps + \Auep \cdot \na_k \right).}
\end{equation}
The operators $\Md$ and $\Mz$ are related by the formula
\begin{equation} \label{eq:def_M2z}
	\boxed{\Mz = \Md -\ep\pa{1+\ep (\chi_1)_\eps} \pa{H^0_k - z}^{-1}_\perp V_1,}
\end{equation}
where $(H^0_k - z)^{-1}_\perp$ denotes the pseudo-inverse of $(H^0_k - z)$ (i.e. $(H^0_k - z)^{-1}_\perp$ coincides with the inverse of $(H^0_k - z)$ on $(\Ker (H^0_k -z))^\perp$ and is identically equal to $0$ on $\Ker (H^0_k -z)$). The respective merits of these two approaches are discussed in Remark~\ref{rem:M2vstM2} below.

\subsection{Approximations of the resolvent and solutions to the linear equation}
\label{sec:resolvent}

The following theorem is key for deriving all the results in this article.

\begin{theorem}[Approximations of the resolvent]\label{thm:res} Let $v$ and $W$ satisfying \eqref{hypo:on_v_W}, $\omega^*$ a compact subset of $\overline{\Omega^*}$, and $\sC$ a compact subset of $\C$ such that
$$
\dist \left( \sC, \overline{\bigcup_{k \in \omega^*}\sigma(H^0_k)}\right) > 0.
$$
Then, there exist $C \in \R_+$ and $\eps_0 > 0$ such that, for all $\eps \in (0, \eps_0] \cap \N^{-1}$, $k \in \omega^*$ and $z \in \sC$, we have  $z \notin \sigma(H^\eps_k) \cup \sigma(H^0_k+\eps V_1)$ and
    \begin{align}
     \left\| \pa{H^\ep_k -z}^{-1} - \pa{H^0_k -z}^{-1} \right\|_{L^2_\per \to L^2_\per}  &\le C \ep,  \label{eq:res_L2} \\
        \nor{\pa{H^\ep_k -z}^{-1} - \Mzu \pa{H^0_k -z}^{-1}}{L^2_\per \rightarrow H^1_\per} &\le C \ep,\label{eq:resest1} \\
        \nor{\pa{H^\ep_k -z}^{-1} - \cM^{\eps, (2)}_{k} \pa{(H^0_k+\eps V_1) -z}^{-1}}{H^1_\per \rightarrow H^1_\per} &\le C \ep^2\label{eq:resest2}, \\
        \nor{\pa{H^\ep_k -z}^{-1} - \Mz \pa{H^0_k -z}^{-1}}{H^1_\per \rightarrow H^1_\per} &\le C \ep^2\label{eq:resest3}.
    \end{align}
    
\end{theorem}
The inequality~\eqref{eq:res_L2} implies in particular that the spectrum of $H_k^\eps$ converges to the one of $H_k^0$. In particular, the operator $H_k^\eps$ is bounded from below with a bound independent of $\eps$.

\medskip

The following corollary is a straightforward reformulation of Theorem~\ref{thm:res}. We state it as an independent result because of its importance for practical applications, and skip its proof for the sake of brevity. For $f$ in $L^2_\per(\Omega)$, we define by $u^\eps_{k, z, f}$, $u^0_{k, z, f}$ and $u^{0, \eps}_{k, z, f}$ the solutions, when they exist, to
\[
    (H_k^\eps - z) u = f, \quad (H_k^0 - z) u = f, \quad \text{and} \quad (H_k^0 + \eps V_1 - z) u = f.
\]

\begin{corollary}[First and second-order homogenization of the linear equation]\label{cor:main_LE}
	Let $v$ and $W$ satisfying \eqref{hypo:on_v_W}, $\omega^*$ a compact subset of $\overline{\Omega^*}$, and $\sC$ a compact subset of $\C$ such that
$$
    \dist \left( \sC, \overline{\bigcup_{k \in \omega^*}\sigma(H^0_k)}\right) > 0.
$$
Then, there exists $ C \in \R_+$ and $\eps_0 > 0$ such that for all $k \in \omega^*$, $z \in \sC$, and $\eps \in (0, \eps_0] \cap \N^{-1}$, $z$ is in the resolvent sets of both $H^\eps_k$ and $H^0_k+\eps V_1$, and
    \begin{itemize}
        \item for all $f \in L^2_\per(\Omega)$, 
        \[
		\nor{u^\eps_{k,z,f} - u^0_{k,z,f}}{L^2_\per} +
		\nor{u^\eps_{k,z,f}- \cM^{\eps, (1)} u^0_{k,z,f} }{H^1_\per} \le C \eps \nor{f}{L^2_\per};
        \]
        \item if in addition $f \in H^1_\per(\Omega)$, then
		\begin{align}
		 \nor{u^\eps_{k,z,f} - \Mz u^0_{k,z,f}}{H^1_\per} &\le C \eps^2 \nor{f}{H^{1}_\per}, \label{eq:approx_lin_eq_2} \\
		\nor{u^\eps_{k,z,f} - \cM^{\eps, (2)}_{k} u^{0,\eps}_{k,z,f} }{H^1_\per} &\le C \eps^2 \nor{f}{H^{1}_\per}.  \label{eq:approx_lin_eq}
	 \end{align}
    \end{itemize}
\end{corollary}


\begin{remark} \label{rem:M2vstM2} There are two ways to construct second-order corrections, given by \eqref{eq:approx_lin_eq_2} and \eqref{eq:approx_lin_eq}.  In view of the definition of $\Mz$ in~\eqref{eq:def_M2z}, to compute the approximation $\Mz u^0_{k,z,f}$ in \eqref{eq:approx_lin_eq_2}, one needs to solve two macro-scale equations, independent of $\eps > 0$, namely
\[
    (H^0_k-z) u^0_{k,z,f} =f \quad \text{and} \quad (H^0_k-z)v^1_{k,z,f}=V_1u^0_{k,z,f}.
\]
(note that $(H^0_k-z)_\perp^{-1}=(H^0_k-z)$ since $z \notin \sigma(H^0_k)$). Instead, for the computation of $\cM^{\eps, (2)}_{k} u^{1,\eps}_{k,z,f}$, one only needs to solve one macro-scale equation, which depends on $\eps > 0$, namely
\[
    (H^0_k+\eps V_1-z)u^{0,\eps}_{k,z,f}=f.
\]
%
%
%
Using \eqref{eq:approx_lin_eq_2} is computationally more efficient if we are interested in computing the solution for many values of $\eps$ since only two macroscopic elliptic PDEs independent of $\eps$ have to be solved. On the other hand, if one is interested in the solution for a single value of $\eps$, it is better to use \eqref{eq:approx_lin_eq} since only one macroscopic PDE (depending on $\eps)$ has to be solved.
\end{remark}

\subsection{Homogenization of the eigenvalue problem}
\label{sec:homog_eig}
We now consider the approximation of the eigenmodes of $H^\eps_k$. The eigenmodes of $H^0_k$, $H^\eps_k$ and $H^0_k + \eps V_1$ are respectively denoted by 
\[
(E_{\ell, k}^0, \psi_{\ell, k}^0), \quad (E_{\ell, k}^\eps, \psi_{\ell, k}^\eps), 
\quad \text{and} \quad
(E_{\ell, k}^{0, \eps}, \psi_{\ell, k}^{0, \eps}),
\]
where $E_{1, k}^0 \le E_{2, k}^0 \le \cdots$ are the eigenvalues of $H^0_k$ counted with multiplicities and ranked in non-decreasing order, and where $(\psi_{\ell, k}^0)_{\ell \in \N^*}$ is a corresponding $L^2_\per(\Omega)$ orthonormal basis of eigenvectors (and similarly for the other eigenmodes, recall that all these operators have a compact resolvent). We assume that the functions $\psi_{\ell, k}^\eps$ and $\psi_{\ell, k}^{0, \eps}$ are oriented in such a way that
\[
\left\langle  \psi_{\ell, k}^\eps,\psi_{\ell, k}^{0} \right\rangle_{L^2_{\rm per}}  \in \R_+
\quad \text{and} \quad
\big\langle  \psi_{\ell, k}^{0, \eps}, \psi_{\ell, k}^{\eps} \big\rangle_{L^2_{\rm per}}  \in \R_+.
\]

First, we prove the convergence of the spectrum.

\begin{theorem} \label{thm:unif_CV_eigenvalues} Let $v$ and $W$ satisfying \eqref{hypo:on_v_W}. We have 
\begin{equation}\label{eq:asympt_eigenvalues}
\min_{k \in \Omega^*} \min_{\eps \in \N^{-1}} E^\eps_{\ell,k} \mathop{\longrightarrow}_{\ell \to \infty} +\infty 
\end{equation}
and for each $\ell \in \N^*$, there exists a constant $C_\ell \in \R_+$ such that
\begin{equation}\label{eq:unif_eigenvalues} 
\forall \eps \in \N^{-1}, \quad \max_{k \in \Omega^*} \left| E^\eps_{\ell,k} - E^0_{\ell,k} \right| \le C_\ell \eps.
\end{equation}
\end{theorem}


We then introduce the first- and second-order approximate eigenmodes
\begin{align}
	 \p^{\ep,(1)}_{\ell,k} 
        & := \frac{\Mzu \p^0_{\ell,k}}{\nor{\Mzu \p^0_{\ell,k}}{L^2_{\rm per}}}, 
        \qquad  E^{\ep,(1)}_{\ell,k}:=\langle  \p^{\ep,(1)}_{\ell,k} | H^\eps_k |  \p^{\ep,(1)}_{\ell,k} \rangle, \label{eq:first_order_app_eig} \\
	 \p^{\ep,(2)}_{\ell,k} 
        & := \frac{\cM_k^{\eps, (2)} \p^{0, \eps}_{\ell,k}}{\nor{\cM_k^{\eps, (2)} \p^{0, \eps}_{\ell,k}}{L^2_{\rm per}}},
        \qquad E^{\ep,(2)}_{\ell,k}:=\langle \p^{\ep,(2)}_{\ell,k} | H^\eps_k |  \p^{\ep,(2)}_{\ell,k} \rangle, \label{eq:second_order_app_eig} \\	
	\widetilde \p^{\ep,(2)}_{\ell,k} 
        &:= \frac{\Mze \pz}{\nor{\Mze \pz}{L^2_{\rm per}}},
        \qquad \widetilde E^{\ep,(2)}_{\ell,k}:=\langle \widetilde \p^{\ep,(2)}_{\ell,k} | H^\eps_k | \widetilde \p^{\ep,(2)}_{\ell,k} \rangle. \label{eq:second_order_app_eig2}
\end{align}

Let us first consider the case of simple eigenvalues.

\begin{theorem}[Homogenization of non-degenerate eigenmodes]\label{thm:eigenmode_nondeg} 
Let $v$ and $W$ satisfying \eqref{hypo:on_v_W}, $\omega^*$ a compact subset of $\overline{\Omega^*}$, and $\ell \in \N^*$ such that for all $k \in \omega^*$, 
$E^0_{\ell, k}$ is a non-degenerate eigenvalue of $H^0_k$. Then, there exists $C \in \R_+$ and $\eps_0 > 0$ such that for all $k \in \omega^*$ and $\eps \in (0, \eps_0] \cap \N^{-1}$, $E^{\eps}_{\ell, k}$ is a non-degenerate eigenvalue of $H^\eps_k$ while $E^{0, \eps}_{\ell, k}$ is a non-degenerate eigenvalue of $(H^0_k+\eps V_1)$, and 
	\begin{align}
	& \| \p^\ep_{\ell,k} -  \p^{\ep,(1)}_{\ell,k} \|_{H^1_{\rm per}} \le C \eps, \quad 
           |E^\ep_{\ell,k} - E^{\eps,(1)}_{\ell,k} |+   |E^\ep_{\ell,k} -  E^{0,\eps}_{\ell,k} | \le C \ep^2,  \label{eq:bounds_order_1} \\
   &  \| \p^\ep_{\ell,k} -  \p^{\ep,(2)}_{\ell,k} \|_{H^1_{\rm per}} \le C \eps^2, \quad      
           |E^\ep_{\ell,k} -  E^{\eps,(2)}_{\ell,k} |  \le C \ep^4,  \label{eq:bounds_order_1.5} \\
   &\| \p^\ep_{\ell,k} - \widetilde \p^{\ep,(2)}_{\ell,k} \|_{H^1_{\rm per}} \le C \eps^2, \quad      
  |E^\ep_{\ell,k} - \widetilde E^{\eps,(2)}_{\ell,k} | \le C \ep^4. \label{eq:bounds_order_2}
	\end{align}
    Moreover, we have the expansion
\begin{align}
E^{\eps}_{\ell,k} &= E^0_{\ell,k} + \ep \int_{\Omega} V_1 |\psi^0_{\ell,k}|^2 + \eps^2 \eta_{\ell,k}^{(2)} + \eps ^3 \eta_{\ell,k}^{(3)}  + O(\ep^4), \label{eq:second_order_app}
\end{align}
uniformly in $k \in \omega^*$, where $\eta_{\ell,k}^{(2)}$ and $\eta_{\ell,k}^{(3)}$ can be computed explicitly from $\psi^0_{\ell,k}$, $v$ and $W$ by solving non-oscillatory elliptic problems in the unit cell $\Omega$, and computing non-oscillatory integrals on~$\Omega$. 
\end{theorem}


\begin{remark} \label{rem:nonnormalized}
We will prove later that, uniformly on $\omega^*$, the three norms 
\[  
    \| \cM^{\eps, (1)} \psi^0_{\ell,k} \|_{L^2_\per}, \quad 
    \| \cM_k^{\eps, (2)} \psi_{\ell, k}^{0,\eps} \|_{L^2_\per}  \quad \text{and} \quad
    \| \Mze \pz \|_{L^2_\per}
\]
behave like $1 + O(\eps^2)$.
 As a consequence, we also have 
 \begin{align*}
     & \| \p^\ep_{\ell,k} -  \Mzu \p^0_{\ell,k} \|_{H^1_{\rm per}} \le C \eps, \\
      & \| \p^\ep_{\ell,k} -  \cM_k^{\eps, (2)} \p^{0,\ep}_{\ell,k} \|_{H^1_{\rm per}} + 
       \| \p^\ep_{\ell,k} -  \widetilde{\cM}_{E^0_{\ell, k},k}^{\eps, (2)} \p^0_{\ell,k} \|_{H^1_{\rm per}} \le C \eps^2.
 \end{align*}
 \end{remark}


Let us now turn to the case of degenerate eigenvalues of $H^0_k$.

\begin{theorem}[Homogenization of degenerate eigenmodes]\label{thm:eigenmode_deg}
	Assume that \eqref{hypo:on_v_W} holds true and that, for some $k \in \Omega^*$, $E^0_{\ell,k}= \cdots = E^0_{\ell+m-1,k}$ is a degenerate eigenvalue of $H^0_k$ of multiplicity ${m \ge 1}$. 
    Then, there exists $C \in \R_+$, $\delta > 0$ and $\ep_0 > 0$ such that for all $\ep \in (0,\eps_0] \cap \N^{-1}$, $H^\eps_k$ and $(H^0_k+\eps V_1)$ have exactly $m$ eigenvalues $E^\eps_{\ell,k} \le \cdots \le E^\ep_{\ell+m-1,k}$ and $E^{0,\eps}_{\ell,k} \le \cdots \le E^{0,\eps}_{\ell+m-1,k}$ (counted with their multiplicities) in the interval $ I_\delta := [E^0_{\ell,k}-\delta,E^0_{\ell,k}+\delta]$, and, setting
    \[
        P^0 :=   \1_{I_\delta}(H^0_k) = \1_{\{E^0_{\ell,k}\}}(H^0_k), \quad
        P^\eps :=  \1_{I_\delta}(H^\eps_k), \quad
        P^{0,\eps} :=  \1_{I_\delta}(H^0_k + \eps V_1),
    \]
    we have the convergences
\begin{align}
	 \nor{ P^\eps - P^0}{L^2_{\rm per} \to L^2_{\rm per}} &\le C \, \eps, \label{eq:deg_evec00} \\
\nor{ P^\eps - \cM^{\eps,(1)} P^0}{L^2_{\rm per} \to H^1_{\rm per}} &\le C \, \eps, \label{eq:deg_evec0} \\
\nor{ P^\eps - \Md P^{0, \eps} }{H^1_{\rm per} \to H^1_{\rm per}} &\le C \, \eps^2, \label{eq:deg_evec} \\
\nor{ \pa{P^\eps - \Mze P^0} P^0  }{L^2_{\rm per} \to H^1_{\rm per}} &\le C \, \eps^2. \label{eq:deg_evec2}
\end{align}
and
\begin{equation}  \label{eq:deg_eval}
 \sum_{j=0}^{m-1} \left| E^\eps_{\ell+j,k} - E^{0,\eps}_{\ell+j,k} \right| \le C \eps^2,
  \quad
 \sum_{j=0}^{m-1} \left| E^\eps_{\ell+j,k} - E^{0}_{\ell+j,k} - \eps \eta_j \right| \le C \eps^2,
\end{equation}
where $\eta_0 \le \eta_1 \le \cdots \le \eta_{m-1}$ are the eigenvalues of the $m \times m$ hermitian matrix $\left( \langle \psi_{\ell+i,k}^{0}, V_1 \psi_{\ell + j, k}^{0}\rangle \right)_{0 \le i, j \le m-1}$.
\end{theorem}

An interpretation of \eqref{eq:deg_eval} is that the degeneracy can be lifted at order $\eps$ by the macroscopic perturbation potential $V_1$. To obtain more accurate approximations of the eigenvalues of $H^\eps_k$, we can proceed as follows. The families of functions $(\psi^{0}_{\ell + j,k})_{0 \le j \le m-1}$ and $(\psi^{0,\eps}_{\ell+j,k})_{0 \le j \le m-1}$ form $L^2_\per$-orthonormal bases of the ranges of the orthogonal projectors $P^0$ and $P^{0, \eps}$. Applying respectively the operators $\Mze$ and $\Md$ to these functions, we deduce from \eqref{eq:deg_evec}-\eqref{eq:deg_evec2} that the so-obtained two families of functions form second-order $H^1_\per$-approximations of $L^2_\per$-orthogonal bases of $P^\eps$. Solving the variational approximation of the eigenvalue problem $H^\eps_k \psi = E \psi$ in one of these two families, we can obtain approximations of the $m$ eigenvalues of $H^\eps_k$ close to $E^0_{\ell,k}$ with $O(\eps^4)$ accuracy.


\subsection{Homogenization of physical quantities}\label{subsec:phys_quantities}

Let us now explain how to use the above estimates to compute approximations of a few physical quantities of interest (QoI). We first consider the kinetic and potential contributions 
\begin{equation}\label{eq:kinetic_potential}
T^\eps_{\ell,k}:= \int_\Omega |\nabla_k \psi^\eps_{\ell,k}|^2  \quad \mbox{and} \quad V^\eps_{\ell,k}:= \int_\Omega (\eps^{-1} v_\eps +W) |\psi^\eps_{\ell,k}|^2
\end{equation}
to the total energy $E^\eps_{\ell,k}$. We set 
\begin{align*}
    & T^0_{\ell,k}:= \int_\Omega |\nabla_k \psi^0_{\ell,k}|^2, \quad V_{\ell,k}:= \int_\Omega V |\psi^0_{\ell,k}|^2, \\
     & W_{\ell,k}:= \int_\Omega W |\psi^0_{\ell,k}|^2, \quad V^0_{\ell,k}=V_{\ell,k}+W_{\ell,k}.
\end{align*}
and introduce the macroscopic potentials $F$ and $G$ defined as
\begin{align}
F(x) &:= -\frac 12 \fint_\Omega v(x,y) \nabla_xA_2 (x,y) \, \d y = -\fint_\Omega \nabla_x\cu(x,y)  \cdot \nabla_y\cu(x,y) \, \d y , \label{eq:def_F} \\
G(x)&:=  \fint_\Omega v(x,y) |\cu(x,y)|^2 \, \d y =  -2 \fint_\Omega \cu(x,y) |\nabla_y\cu(x,y)|^2 \, \d y, \label{eq:def_G}
\end{align}
which are related to $V_1$ through the equality
\begin{equation}\label{eq:V1=G-2F}
V_1=G-2F.
\end{equation}

\begin{theorem}[Kinetic and potential energies]\label{thm:phys_quantities1} Assume the $v$ and $w$ satisfy \eqref{hypo:on_v_W}. Let $\ell \in \N^*$ and $k \in \Omega^*$ be such that $E^0_{\ell,k}$ is a non-degenerate eigenvalue of $H^0_k$. Then for $\eps > 0$ small enough, $E^\eps_{\ell,k}$ is a non-degenerate eigenvalue of $H^\eps_k$ and the kinetic and potential components of the energy $E^\eps_{\ell,k}$ defined in \eqref{eq:kinetic_potential} satisfy
\begin{align*}
T^\eps_{\ell,k}&=  T^0_{\ell,k}-V_{\ell,k} + \eps \left( \int_\Omega \pa{2F-2G} \ab{\p^0_{\ell,k}}^2 - 2 \int_\Omega (2V+W)   {\rm Re}(\overline{\psi^0_{\ell,k}}\phi^1_{\ell,k}) \right) \\
&\qquad\qquad\qquad\qquad\qquad\qquad\qquad\qquad\qquad\qquad\qquad\qquad\qquad+ O(\ep^2), \\
 V^\eps_{\ell,k}&= V^0_{\ell,k} + V_{\ell,k} + \eps \left( \int_\Omega \pa{3G - 4F} \ab{\p^0_{\ell,k}}^2 + 2 \int_\Omega (2V+W)   {\rm Re}(\overline{\psi^0_{\ell,k}}\phi^1_{\ell,k})\right) \\
		&\qquad\qquad\qquad\qquad\qquad\qquad\qquad\qquad\qquad\qquad\qquad\qquad\qquad +  O(\ep^2),
\end{align*}
where $\phi^1_{\ell,k}:=-(H^0_k-E^0_{\ell,k})_\perp^{-1}(V_1\psi^0_{\ell,k})$.
\end{theorem}

In particular, the homogenization process leads to a non-negative energy transfer from kinetic to potential energy of magnitude $\delta E_{\ell,k}=-V_{\ell,k} \ge 0$ (recall that $V \le 0$ , see \eqref{eq:def_V}):
\begin{equation} \label{eq:lim_kinetic_eps}
    T^0_{\ell,k} =  \left(\lim_{\eps \to 0} T^\eps_{\ell,k} \right) - \delta E_{\ell,k} \quad \mbox{and} \quad V^0_{\ell,k} =  \left(\lim_{\eps \to 0} V^\eps_{\ell,k}\right)+\delta E_{\ell,k}.
\end{equation}
This is due to the fact that $\psi^\eps_{\ell,k}$ oscillates at scale $\eps$ and therefore has a higher kinetic energy than $\psi^0_{\ell,k}$. Note also that in view of \eqref{eq:V1=G-2F},
\begin{align}
	E^\eps_{\ell,k}&=T^\eps_{\ell,k}+V^\eps_{\ell,k} = T^0_{\ell,k}+V^0_{\ell,k} + \eps  \int_\Omega \pa{G-2F} \ab{\p^0_{\ell,k}}^2  +O(\eps^2) \nonumber \\
& = E^0_{\ell,k} + \eps \int_\Omega V_1 |\psi^0_{\ell,k}|^2 + O(\eps^2),
\end{align}
in agreement with \eqref{eq:second_order_app}. We also have (see Section~\ref{sec:proof_phys_quantities1})
\begin{align}\label{eq:H2ep_approx}
\ep^2 \int_\Omega \ab{\Delta \p^\ep_{\ell,k}}^2 = \int_\Omega \d x \ab{\p^0_{\ell,k}(x)}^2 \fint_\Omega  v(x,y)^2\d y + O(\ep),
\end{align}
so that $\ep \|\p^\ep_{\ell,k} \|_{H^2_\per}$ is bounded uniformly in $\eps$.

\medskip

The second QoI is the integrated density of states $\cI^\ep: \R \to \R_+$ defined by
\begin{align*}
	\cN^\ep(E) := \fint_{\Omega^*} \tr \1_{(-\infty,E]}\pa{H^\ep_k} \d k = \sum_{\ell \in \N^*} \fint_{\Omega^*} \1\pa{E^\ep_{\ell,k} \le E} \d k,
\end{align*}
which we want to compare to the homogenized integrated density of states 
\begin{align*}
	\cN^0(E) := \fint_{\Omega^*} \tr \1_{(-\infty,E]}\pa{H^0_k} \d k = \sum_{\ell \in \N^*} \fint_{\Omega^*} \1\pa{E^0_{\ell,k} \le E} \d k.
\end{align*}
The arguments we use are similar to the ones in~\cite{Canc_s_2020}. Consider an $\L$-periodic Schr\"odinger operator $H$, with Bloch eigenvalues $E_{\ell,k}$. For each $\ell \in \N^*$ and $\mu \in \R$, we define the level set  
\begin{align*}
	S^\mu_\ell := \acs{k \in \Omega^* \, \bigr\vert \, E_{\ell,k} = \mu}.
\end{align*}
We say that $\mu \in \R$ is a non-degenerate point of the spectrum of $H$ if the following conditions are satisfied:
\begin{description}
\item[(C1)] $S^\mu:=\bigcup_{\ell \in \N^*} S_{\ell}^\mu \neq \emptyset$ \qquad ($\mu \in \sigma(H)$), 
\item[(C2)] for all $\ell \in \N$, $S_\ell^\mu \cap S_{\ell +1}^\mu = \emptyset$ (no surface crossing at energy level~$\mu$),
\item[(C3)] $\partial_k E_{\ell,k} \neq 0$ for all $k \in S_\ell^\mu$ (no van Hove singularities at energy level~$\mu$). 
\end{description}
Here and below $\partial_k E_{\ell,k}$ denotes the gradient of the function $\R^d \ni k \mapsto E_{\ell,k} \in \R$, which is well-defined by Kato's perturbation theory on $S_\ell^\mu$ in view of the non-degeneracy condition (C2).

\begin{theorem}[Integrated density of states]\label{thm:phys_quantities2} Assume $v$ and $W$ satisfy \eqref{hypo:on_v_W}. Let $E \in \R$. 
\begin{enumerate}
\item If $E \notin \sigma(H^0)$, there exists $\eps_0 > 0$ such that for all $\eps \in (0,\eps_0] \cap \N^{-1}$, $\cN^\ep(E)=\cN^0(E)$;
\item If $E$ is a non-degenerate point of the spectrum of $H^0$, then we have, using the notation of Theorem~\ref{thm:eigenmode_nondeg},
\begin{align*}
	\cN^\ep(E) &= \cN^0(E)+ O(\eps)= \sum_{\ell \in \N^*} \fint_{\Omega^*} \1\pa{E^{0, \eps}_{\ell,k}  \le E} \, \d k +  O(\ep^2) \\
	&= \sum_{\ell \in \N^*} \fint_{\Omega^*} \1\pa{E^{\eps,(2)}_{\ell,k} \le E} \, \d k +  O(\ep^4) \\
    &= \sum_{\ell \in \N^*} \fint_{\Omega^*} \1\pa{\widetilde E^{\eps,(2)}_{\ell,k} \le E} \, \d k +  O(\ep^4).
 \end{align*}
\end{enumerate}
\end{theorem}

%
%
%

\section{Numerical results}\label{sec:numerical_res}

\subsection{Comparison between first-order and second-order approximations}

We first illustrate Theorem~\ref{thm:eigenmode_nondeg} on the error between the eigenmodes of $H^\eps_k$ and their first- and second-order approximations obtained by homogenization theory. In order to be able to compute accurate reference solutions for small values of $\eps$, we limit ourselves to the one-dimension case ($d=1$). In our simulations, we consider the lattice $\L=2\pi\Z$, the unit cell $\Omega=[0,2\pi)$, and the potentials
\begin{equation*}
    v(x,y) = \cos (2 x) \cos (2y) - \cos(x) \, \sin(y) \quad \mbox{and} \quad W(x)=0.
\end{equation*}
We aim at approximating the ground-state ($\ell=1$) of $H^\eps_k$ at quasi-momentum $k=0$. On Figure \ref{fig:norms_diff}, we plot the quantities
\begin{align} 
    & \ab{E^\ep_{\ell,k} - E^{0, \eps}_{\ell,k}}, \qquad  
         \ab{E^\ep_{\ell,k} - E^{\eps, (1)}_{\ell,k}}, \qquad
         \ab{E^\ep_{\ell,k} - E^{\ep,(2)}_{\ell,k}}, \qquad
        \ab{E^\ep_{\ell,k} - \widetilde{E}^{\ep,(2)}_{\ell,k}},  \nonumber \\
    & \qquad \nor{\p^\ep_{\ell,k} - \p^{\eps, (1)}_{\ell,k}}{H^1_\per}, \quad 
    \nor{\p^\ep_{\ell,k} -  \p^{\ep,(2)}_{\ell,k}}{H^1_\per}, \quad  
    \nor{\p^\ep_{\ell,k}  - \widetilde{\p}^{\ep,(2)}_{\ell,k} }{H^1_\per}.
    \label{eq:norms_diff}
\end{align}
When $\eps = p/q$ with $p,q \in \N^*$, $p > 1$, the above quantities are computed in a supercell of size $p$, see Remark~\ref{rem:rational}. For consistency, the $H^1_\per$-norm is defined by
$$
\nor{\psi}{H^1_\per} = \left( \frac 1 p \int_{p\Omega} \pa{ |\psi|^2+|\nabla\psi|^2} \right)^{1/2},
$$
and the eigenfunctions are normalized in such a way that
$$
\frac 1p \int_{p\Omega} \left|\p^\ep_{\ell,k} \right|^2 = \frac 1p \int_{p\Omega} \left|\p^{\ep,(1)}_{\ell,k} \right|^2 =\frac 1p \int_{p\Omega} \left|\p^{\ep,(1)}_{\ell,k} \right|^2=\frac 1p \int_{p\Omega} \left|\widetilde\p^{\ep,(2)}_{\ell,k} \right|^2 =1.
$$
We observe that in the regime $\eps < 0.4$, our second-order approximations greatly improve the first-order ones. We also observe that the asymptotic regimes described in Theorem~\ref{thm:unif_CV_eigenvalues} seem to be attained around $\eps \approx 1/5$. 

\begin{figure}[h]
    \begin{center}
        \includegraphics[width=0.65\textwidth]{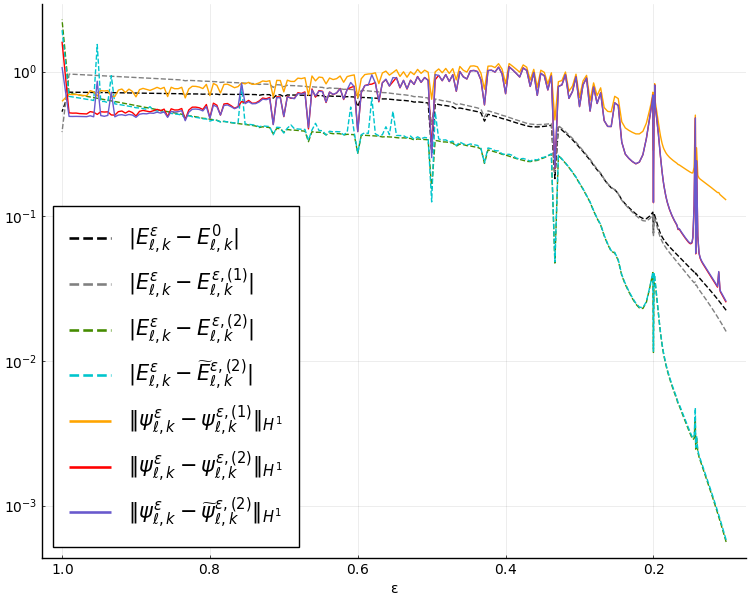} 
        \includegraphics[width=0.3\textwidth]{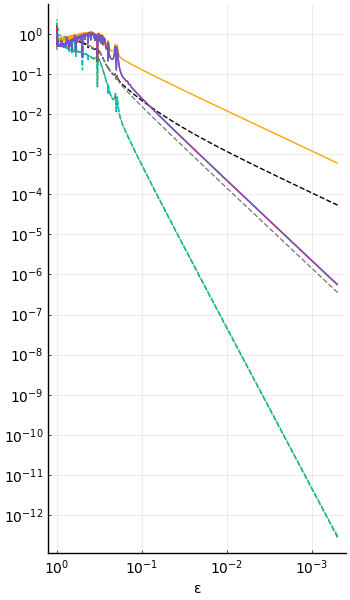} 
        \caption{Quantities \eqref{eq:norms_diff}  (for $\ell=1$ and $k=0$) against $\ep$, for $\ep \in [10^{-1},1] \cap \Q$ (left), and $\ep \in [10^{-3},1] \cap \Q$ (right). 
        }\label{fig:norms_diff}
    \end{center}
\end{figure}


Next, we numerically evaluate the optimal constants $C$ in the estimates \eqref{eq:bounds_order_1}-\eqref{eq:bounds_order_2}. In Figure~\ref{fig:constants}, we plot the quantities
\begin{align}
    & \eps^{-1} \ab{E^\ep_{\ell,k} \hspace{-0.05cm} -\hspace{-0.05cm} E^{0, \eps}_{\ell,k}}, \  
    \eps^{-2} \ab{E^\ep_{\ell,k}\hspace{-0.05cm} - \hspace{-0.05cm}E^{\eps, (1)}_{\ell,k}}, \
    \eps^{-4} \ab{E^\ep_{\ell,k}\hspace{-0.05cm} -\hspace{-0.05cm} E^{\eps, (2)}_{\ell,k}}, \
    \eps^{-4} \ab{E^\ep_{\ell,k}\hspace{-0.05cm} -\hspace{-0.05cm} \widetilde{E}^{\ep,(2)}_{\ell,k}},  \nonumber\\
    &   
    \eps^{-1} \nor{\p^\ep_{\ell,k} - \p^{\eps, (1)}_{\ell,k}}{H^1_\per}, \ 
    \eps^{-2}  \nor{\p^\ep_{\ell,k} -  \p^{\ep,(2)}_{\ell,k}}{H^1_\per}, \  
    \eps^{-2}  \nor{\p^\ep_{\ell,k}  - \widetilde{\p}^{\ep,(2)}_{\ell,k} }{H^1_\per}, \label{eq:constants}
\end{align}
still for $\ell=1$ and $k=0$.
As expected, they are bounded, and their maxima over $\ep \in (0,\ep_0]$ give the optimal prefactors in the bounds~\eqref{eq:bounds_order_1}-\eqref{eq:bounds_order_2} for the example under consideration. We numerically observe that these constants are of orders $1$ to $100$. 
\begin{figure}[!h]
    \begin{center}
        \includegraphics[trim={1.3cm 0 0 0},clip,width=0.99\textwidth]{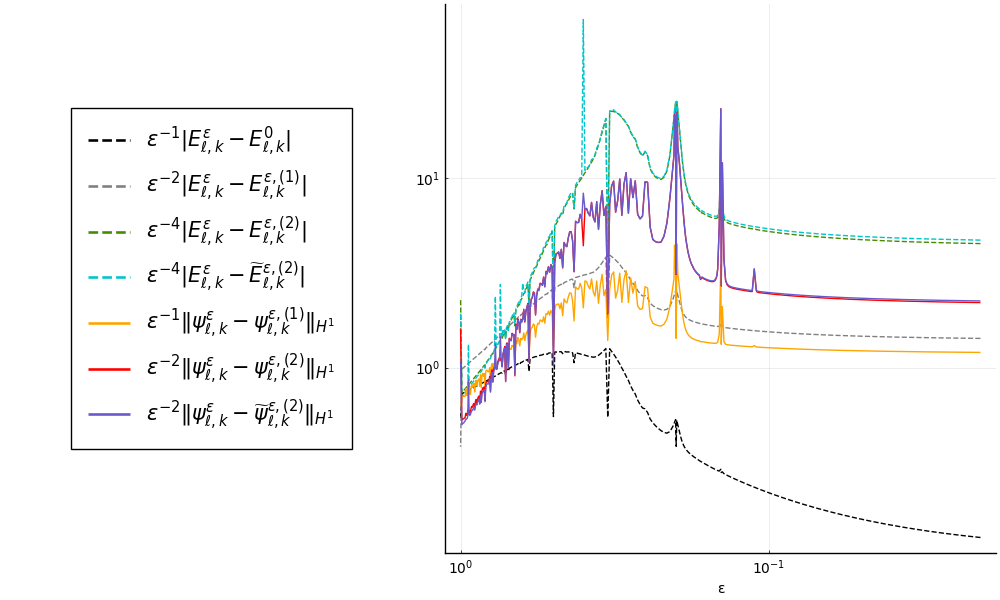} 
    \end{center}
    \caption{Quantities \eqref{eq:constants} (for $\ell=1$ and $k=0$) against $\ep \in \Q$.}\label{fig:constants}
\end{figure}

We remark on $E^{\ep,(2)}_{\ell,k}$ and $\widetilde{E}^{\ep,(2)}_{\ell,k}$ (resp. $\psi^{\ep,(2)}_{\ell,k}$ and $\widetilde{\psi}^{\ep,(2)}_{\ell,k}$) give very similar approximations of $E^\eps_{\ell,k}$ (resp. $\psi^\eps_{\ell,k}$ in $H^1_\per$-norm). Figure~\ref{fig:constants} also suggests that the prefactors are independent of the numerator of $\eps$ for $\eps=\frac pq$, with $p$ and $q$ coprime integers (at least for the quasi-momentum $k = 0$)

\medskip

We do not plot the errors obtained with the non-normalized approximations introduced in Remark~\ref{rem:nonnormalized}, and only mention that the normalization procedure only slightly modifies the numerical results.

\subsection{Localization of $\p^\ep_{\ell,k}$}
Next, we study the effect of the `level of commensurability' of $\eps = p/q \in \Q$ on the localization of the eigenfunctions. When $\eps = p/q$ with $p$ and $q$ coprime integers, both operators $H^0$ and $H^\eps$ are $p\L$-periodic. We display in Figure~\ref{fig:kin} the kinetic energy {\em per unit volume} of $\psi^\eps_{\ell,k}$, that is
\[
    \frac{\dps \int_{p \Omega} | \nabla \psi^\eps_{\ell, k} |^2}{\dps \int_{p \Omega} |\psi^\eps_{\ell, k} |^2}, 
\]
as well as the one of $\psi^{\eps,(2)}_{\ell,k}$, as a function of $\eps$. We see that this gives smooth curves, except for some exceptional values of $\eps$ (often corresponding to low values of $p$ and $q$). 
We also note that the kinetic energy per unit volume of $\psi^\eps_{\ell, k}$ is larger than the one of its approximation $\psi^{\eps,(2)}_{\ell,k}$ for large values of $\eps$ and that the difference between these two quantities goes to zero with $\eps$. 

\begin{figure}[h]
    \begin{center}
        \includegraphics[width=0.99\textwidth]{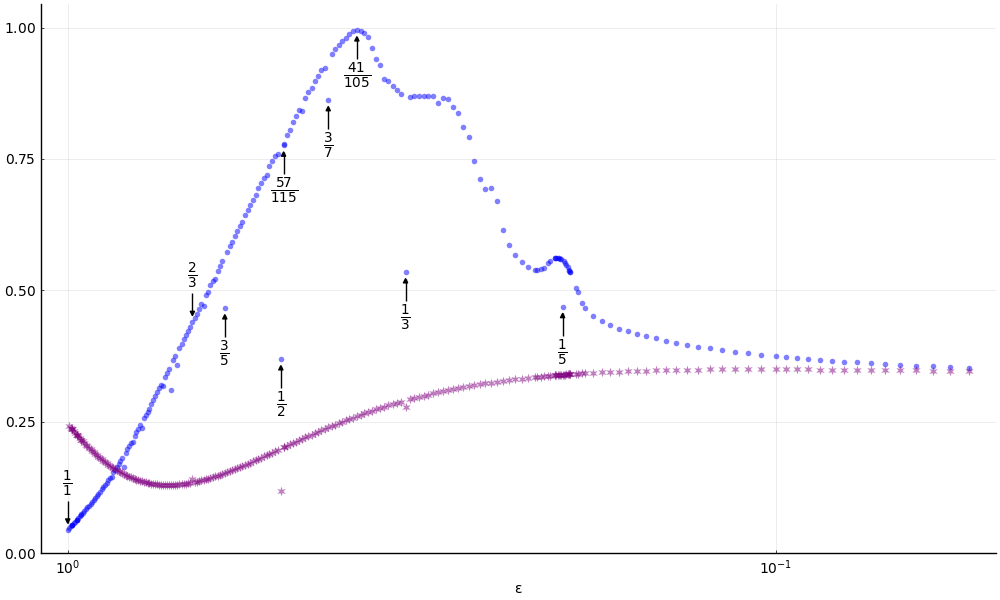} 
    \end{center}
    \caption{Kinetic energy per unit volume of $\psi^\eps_{\ell,k}$ (circles) and $\psi^{\eps,(2)}_{\ell,k}$ (stars) as functions of $\ep = p/q$, for many values of $p$ and $q$. The values of $\ep = p/q$ are indicated for some points.}\label{fig:kin}
\end{figure}

In Figure~\ref{fig:curves}, we plot $| \psi_{\ell, k}^\eps |^2$ for $\eps = 1/2$ and $\eps = 57/115 \approx 1/2$. In the first case, the function is delocalized and has a low kinetic energy per unit volume. In the second case however, the wave function is much more localized, and has a higher kinetic energy per unit volume. This  is reminiscent of Anderson localization~\cite{Anderson_1958}: the local disregistry field seems to play a role similar to disorder and localize the eigenstates. For $\ep$ small enough, more precisely for $\ep \lesssim 1/6$, this localization phenomenon does not happen, i.e. in this case $\p^\ep_{\ell,k}$ is always delocalized.

\begin{figure}[!h]
    \begin{center}
        \includegraphics[trim={0.1cm 0 0.1cm 0},clip,width=0.99\textwidth]{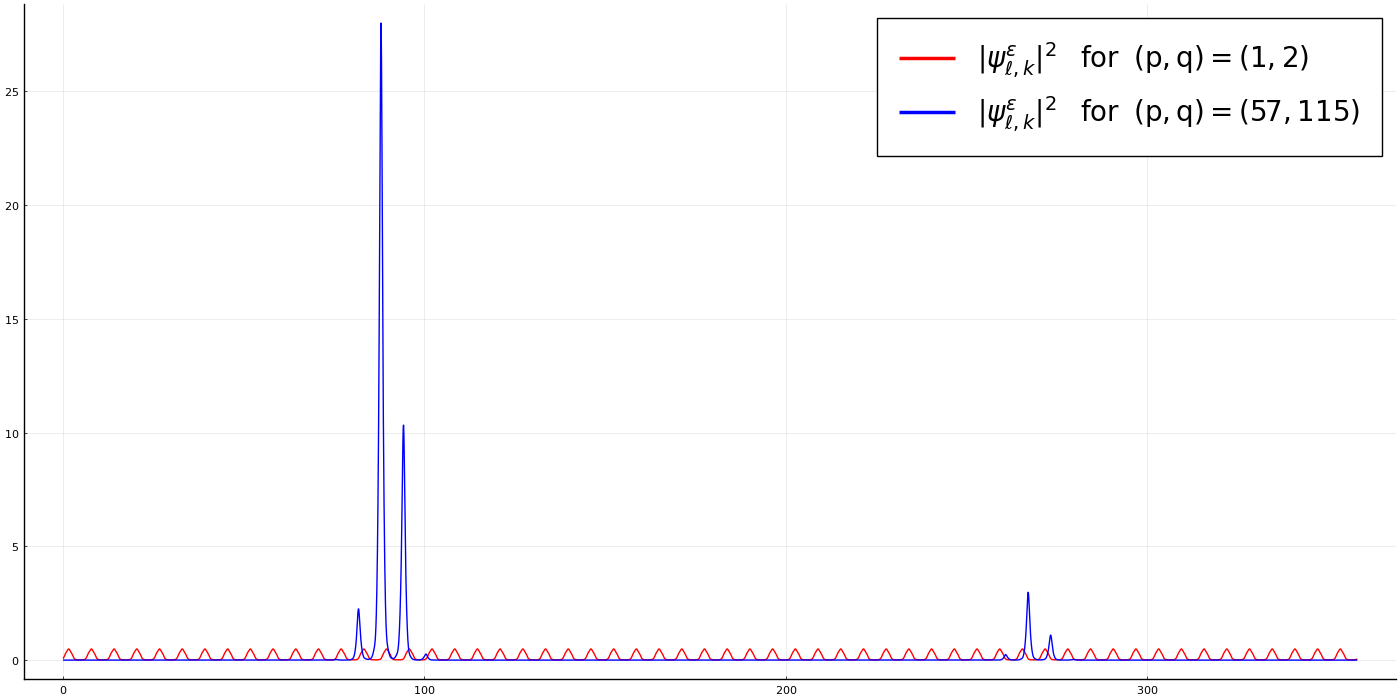} 
            \end{center}
    \caption{Plots of $|\p^\ep_{\ell,k}|^2$ (with $\ell=1$ and $k=0$) for $\eps = \frac 12$ and  $\eps=\frac{57}{115}=\frac 12 - \frac 1{330}$ on their smallest common periodic cell $[0,114\pi]$.}\label{fig:curves}
\end{figure}

Lastly, we plot in Figure~\ref{fig:curves2} the densities and potentials for $\ell=1$, $k=0$ and $\eps=\frac 12$. We observe that even for such a large value of $\eps$, the second-order approximation $|\psi^{\eps,(2)}_{\ell,k}|^2$ of the density $|\psi^{\eps}_{\ell,k}|^2$ has a qualitatively correct profile, while the zeroth- and first-order approximations $|\psi^{0}_{\ell,k}|^2$ and $|\psi^{\eps,(1)}_{\ell,k}|^2$ have not.

\begin{figure}[!h]
    \begin{center}
        \includegraphics[trim={0.1cm 0 0.1cm 0},clip,width=0.99\textwidth]{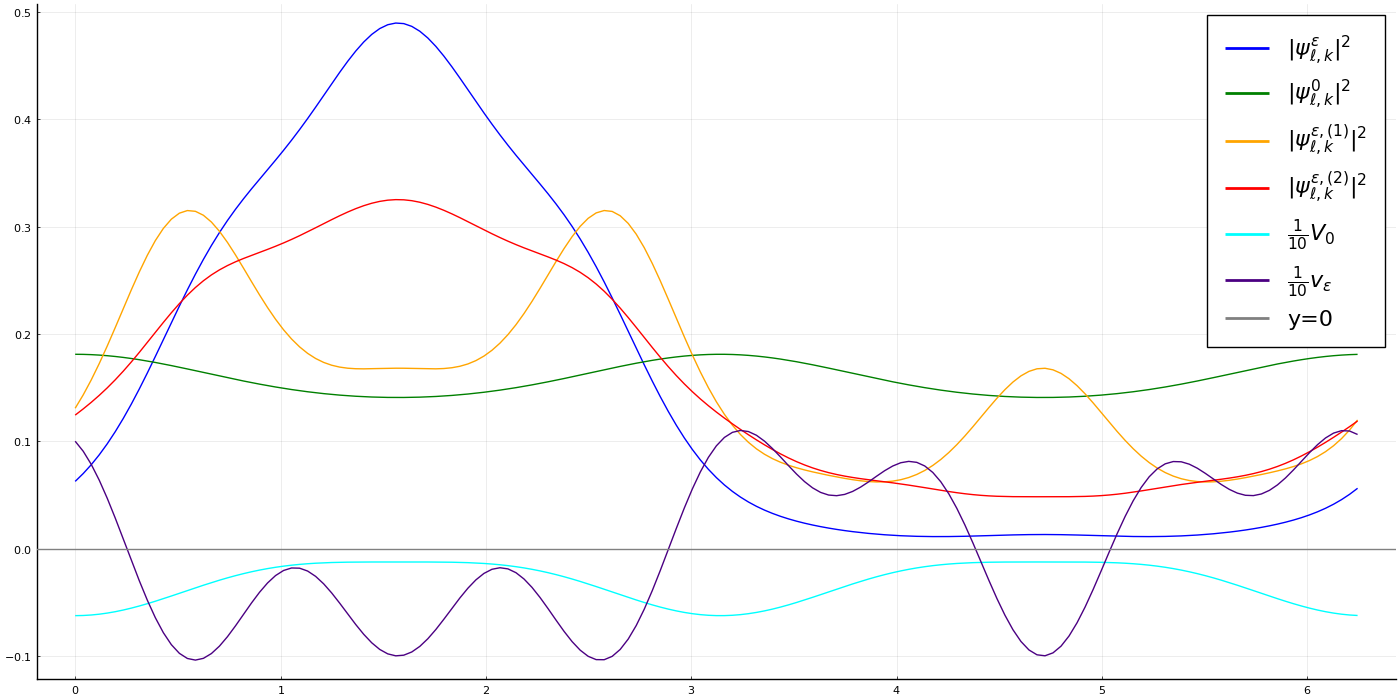} 
    \end{center}
    \caption{Plots of the densities and potentials (for $\ell=1$ and $k=0$) for $\eps=\frac 12$ on $[0,2\pi]$.}\label{fig:curves2}
\end{figure}



\section{Formal two-scale expansions}\label{sec:formal_expansions}

In this section, we perform the formal two-scale expansions for both the eigenvalue problem and the linear equation. 

\subsection{Two-scale expansion of the eigenvalue problem}\label{subsec:TSE_eigenvalue}

We start with the eigenvalue problem
\begin{equation} \label{eq:Schrodinger_eps}
H^\ep_k \p^\eps = E^\eps \p^\ep.
\end{equation}
Our goal is to approximate the eigenmode $(E^\eps_{\ell, k}, \psi^\eps_{\ell,k})$ with simpler quantities, computed from macroscopic objects. Without loss of generality, we may impose the intermediate normalization condition 
\begin{equation} \label{eq:intermediate_norm}
    \langle \p^0_{\ell,k},\p^\eps \rangle_{L^2_{\rm per}}=1.
\end{equation}
We consider the formal expansions
\begin{equation}\label{eq:formal_expansions}
\psi^\ep(x) = \sum_{i=0}^{+\infty} \ep^i \p_i\left( x, \frac x\eps \right), 
\quad \text{and} \quad 
E^\eps = \sum_{i=0}^{+\infty} \ep^i E_i.
\end{equation}
where $\p_i(x,y)$ is $\L$-periodic with respect to both $x$ and $y$, and
insert these expansions in \eqref{eq:Schrodinger_eps}-\eqref{eq:intermediate_norm}. In what follows, we write $f(x)$ if $f$ depends only on the variable $x$, and $f$ if it depends on the two variable $(x,y)$. Regrouping the terms in $\ep^i$ and considering $x$ and $y=\frac x\ep$ as independent variables, we obtain
\begin{multline*}
     \eps^{-2} \left( - \Delta_y \psi_0 \right) + \eps^{-1} \left[- \Delta_y \psi_1 + (v - 2 \na_{x,k} \cdot \na_y ) \psi_0 \right] \\
      +  \sum_{i \ge 0} \ep^i \bigg( - \Delta_y \p_{i+2}  + \pa{v -2 \na_{x,k} \cdot \na_y } \p_{i+1} \\
     + \pa{-\nak^2 + W(x)} \p_i - \pa{E*\p}_i \bigg) = 0,
 \end{multline*}
where we used the relations
\begin{align*}
\na_k \pa{fg} = f \na_k g + g \na f, \qquad \na_k^2 \pa{fg} = g \na_k^2 f + 2 \na_k f \cdot \na g + f \Delta g.
\end{align*}
At the formal level, we interpret the normalization condition~\eqref{eq:intermediate_norm} as
\[
    \fint_{\Omega} \langle \psi_{\ell, k}^0, \psi_0(\cdot, y) \rangle_{L^2_\per} \d y = 1, 
\quad \text{and} \quad
    \forall i \ge 1, \quad \fint_{\Omega} \langle \psi_{\ell, k}^0, \psi_i(\cdot, y) \rangle_{L^2_\per} \d y = 0.
\]

\paragraph{Terms in $\ep^{-2}$.} At leading order in $\eps$, the equation becomes $\Delta_y\p_{0}(x,y)=0$, from which we infer that $\psi_{0}$ is independent of $y$. We write $\psi_0(x,y)=\psi_0(x)$ to simplify the notation. The normalization condition simply reads $\langle \psi^0_{\ell,k},\psi_0 \rangle_{L^2_{\rm per}}=1$.

\paragraph{Terms in $\ep^{-1}$.} Since $\psi_0$ is independent of $y$, the next order simply reads
\begin{equation*}
\Delta_y \p_1(x,y) = v(x,y)\p_0(x).
\end{equation*}
Integrating in $y$ on $\Omega$, we obtain the compatibility condition 
\begin{align*}
\psi_0(x)\fint_\Omega v(x,y)\, \d y=0
\end{align*}
 for all $x$, which justifies the assumption $\fint_\Omega v(\cdot,y)\, \d y=0$. We will see later that $\psi_0$ is a normalized eigenfunction of a Schr\"odinger operator, and therefore does not vanish on a set of positive measure by the unique continuation principle. This also motivates the introduction of the corrector $\chi_1$ defined in \eqref{eq:def_chi}, and satisfies $\Delta_y \chi_1 = v$. We obtain
\begin{equation}\label{eq:psi1}
\p_1(x,y) = \chi_1(x,y) \p_0(x) +U_1(x),
\end{equation} 
where the function $U_1$ only depends on the macroscopic variable $x$. Since $\fint_\Omega \chi_1(x,y) \d y = 0$, the normalization condition reads 
\begin{equation} \label{eq:psi0_U1}
    \langle \psi^0_{\ell,k},U_1 \rangle_{L^2_{\rm per}}=0.
\end{equation}

\paragraph{Terms in $\ep^0$.} We have 
\begin{align*}
	- \Delta_y \p_2 + \pa{v -2 \na_{x,k} \cdot  \na_y} \p_1 +\pa{- \na_{x,k}^2 + W(x) -E_0} \p_0 = 0.
	\end{align*}
	Inserting the formula for $\psi_1$ in~\eqref{eq:psi1} and reordering the terms, we get
	\begin{align}
	0= &  \left[ \pa{-\na_{k}^2 +W - E_0} \p_0  \right] (x) + v \chi_1 \p_0(x) -2 \pa{\nabla_x \cdot  \na_y \chi_1} \psi_0(x) \nonumber \\
& - \Delta_y \p_2 -2 \na_y \chi_1 \cdot \nabla_k \p_0(x)  + vU_1(x). \label{eq:eps0}
\end{align}
Averaging with respect to $y$ on the periodic $\Omega$, the last four terms of the right-hand side vanish and we get
$$
\pa{-\na_{k}^2 +W(x) - E_0} \p_0(x) + \underbrace{\left( \fint_\Omega v(x,y) \chi_1(x,y) \, \d y \right)}_{=V(x)} \p_0(x) =0,
$$
which justifies the definition~\eqref{eq:def_V} of the homogenized potential $V$. The pair $(E_0,\psi_0)$ therefore is an eigenmode of the homogenized operator $H^0_k=-\na_{k}^2 +V+W$ such that $ \langle \psi^0_{\ell,k},\psi_0 \rangle_{L^2_{\rm per}}=1$. Assuming for simplicity that $E^0_{\ell,k}$ is a non-degenerate eigenvalue of $H^0_k$, we obtain that 
\begin{equation}\label{eq:E0_psi0}
E_0=E^0_{\ell,k} \quad \mbox{and} \quad \psi_0=\p^0_{\ell,k}.
\end{equation}

Going back to~\eqref{eq:eps0} and recalling that $V_0 = V + W$, we obtain
\begin{align*}
	\Delta_y \p_2 &=  \left[ v \chi_1 -V_0(x) -2 \pa{\nabla_x \cdot \na_y} \chi_1 \right] \p^0_{\ell,k}  -2 \na_y \chi_1 \cdot \na_{k} \p^0_{\ell,k}(x)  + vU_1(x).
\end{align*}
All quantities have null average in $y \in \Omega$. Recalling the definition of $\chi_2$ and $A_2 $ in~\eqref{eq:def_chi2}-\eqref{eq:def_A}, we integrate this expression to obtain
\begin{align}
	\p_2 = \chi_2 \p^0_{\ell,k}(x) + A_2  \cdot \na_{k} \p^0_{\ell,k}(x) + \chi_1 U_1(x) + U_2(x), \label{eq:psi2}
\end{align}
where the function $U_2$ only depends on the macroscopic variable $x$. In order to have an approximation of $\psi^\eps_{\ell, k}$ up to order $2$, we need to identify the functions $U_1$ and $U_2$ appearing in the expression of $\psi_2$. It turns out that only the expression of $U_1$ matters (the term $U_2$ will give a higher order contribution). To find $U_1$, we consider the next order. 

\paragraph{Terms in $\ep^{1}$.} Using \eqref{eq:E0_psi0}, the term in $\eps^1$ reads
\begin{align*}
	 \Delta_y \p_3 = \pa{v- 2 \na_{x,k} \cdot  \na_y } \p_2 +\pa{- \na_{x,k}^2 +W(x) - E_0 } \p_1 - E_1 \p^0(x).
 \end{align*}
Inserting the expressions found for $E_0$, $\psi_0$, $\psi_1$ and $\psi_2$, this is also
\begin{align}
     \Delta_y \p_3  = & \left( v- 2 \na_{x,k} \cdot  \na_y  \right) \left[ \chi_2 \p^0_{\ell,k}(x) + A_2  \cdot \na_{k} \p^0_{\ell,k}(x) + \chi_1 U_1(x) + U_2(x)  \right]   \nonumber \\
     & \quad + \left(- \na_{x,k}^2 +W(x) - E^0_{\ell, k} \right) \left[\chi_1 \p_0(x) +U_1(x) \right] - E_1 \psi^0_{\ell, k}(x). \label{eq:eps1}
\end{align}
We first average in $y$. Recall that $V_1$ was defined in~\eqref{eq:def_V1} by $V_1(x) := \fint_{\Omega} v(x,y) \chi_2(x,y) \d y$,
and notice that
\begin{align}
    \fint_{\Omega} v(x,y) A_2 (x,y) \d y \hspace{-0.05cm}= \fint_{\Omega} \chi_1 \left( \Delta_y A_2 \right)\hspace{-0.05cm}  = -2 \fint_\Omega \chi_1 \left( \nabla_y \chi_1 \right) \hspace{-0.05cm} = - \fint_\Omega \nabla_y | \chi_1 |^2 = 0.  \label{eq:integral_vA1}
\end{align}
We obtain
 \begin{align*}
 E_1 \p^0_{\ell,k}(x) &= V_1(x) \psi^0_{\ell, k}(x) + V(x) U_1(x) +   \pa{- \na_{k}^2 +W(x) - E^0_{\ell,k} } U_1(x)   \nonumber  \\
 & = \pa{H^0_k - E^0_{\ell,k}} U_1(x)  + V_1(x)  \p^0_{\ell,k}(x). 
  \end{align*}
Taking the inner product with $\psi^0_{\ell, k}$ yields
 \begin{align} \label{eq:def_E1}
	 E_1 = \int_{\Omega} V_1 \ab{\p^0_{\ell,k}}^2,
 \end{align}
 and hence
 \begin{align}\label{cro}
	(H^0_k-E^0_{\ell,k}) U_1 =   (E_1 - V_1)\p^0_{\ell,k} \in (\p^0_{\ell,k})^\perp \quad \mbox{(in $L^2_{\rm per}(\Omega)$)}.
 \end{align}
Still assuming that $E_k^0$ is non-degenerate, and recalling that $\langle \psi^0_{\ell, k}, U_1 \rangle_{L^2_\per} = 0$, see~\eqref{eq:psi0_U1}, we finally get 
\begin{equation}\label{eq:def_U1}
    U_1 = - \pa{H^0_k - E_k^0}_{\perp}^{-1} (V_1 \p^0_{\ell,k}).
\end{equation}
Note that $U_1$ coincides with the function $\phi^1_{\ell,k}$ introduced in Theorem~\ref{thm:phys_quantities1}.
One can continue the computations, and find explicit expressions for $U_2(x)$, $\psi_3(x, y)$, and so on (see Appendix~\ref{sec:detail_ansatz}). For our purpose however, we may stop here.

\paragraph{Conclusion.} Collecting the results in \eqref{eq:psi1}, \eqref{eq:E0_psi0}, \eqref{eq:psi2}, \eqref{eq:def_U1}, we obtain 
\begin{align*}
	\psi^\eps(x) &=  \psi^0_{\ell, k} (x) + \eps \left[ (\chi_1)_\eps  \psi^0_{\ell, k} (x) -  \pa{H^0_k - E_k^0}_{\perp}^{-1} (V_1 \p^0_{\ell,k}) \right] \\
    & \quad + \eps^2 \left[(\chi_2)_\eps \psi^0_{\ell, k} (x) + (A_2 )_\eps \cdot \nabla_k \psi^0_{\ell, k} -  \left( \chi_1 \right)_\eps  \pa{H^0_k - E_k^0}_{\perp}^{-1} (V_1 \p^0_{\ell,k}) \right] \\
    & \quad + \eps^2 U_2(x) + O(\eps^3) \qquad \mbox{(formal two-scale expansion)}.
\end{align*}
At this point, we may expect that the remainder is actually $O(\eps^3)$ for the $L^2_{\rm per}$-norm, but certainly not for the $H^1_\per$-norm because the following term in the expansion, namely $\eps^3 \psi_3 \left(x, \frac{x}{\eps}\right)$, is highly oscillatory. On the other hand, we may expect that the remainder is $O(\eps^2)$ for the $H^1_\per$-norm since the latter only involves first derivatives of quantities oscillating at scale $\eps$. As our goal is to obtain first-order approximations of eigenfunctions of $H^\eps_k$ for the $H^1_\per$-norm, we may get rid of the non-oscillatory term $\eps^2 U_2(x)$. We finally end up with 
\begin{align*}
    \psi^\eps = & \left(1+ \eps  (\chi_1)_\eps  + \eps^2 \left[ (\chi_2)_\eps + (A_2 )_\eps \cdot \nabla_k \right] \right) \psi^0_{\ell, k}  \\
    & \quad - \eps (1+\eps (\chi_1)_\eps)   \pa{H^0_k - E_k^0}_{\perp}^{-1} (V_1 \p^0_{\ell,k}) + O_{H^1_\per}(\eps^2) \qquad \mbox{(expected)}.
\end{align*}
This motivates the introduction of the operator $\Mz$ defined in~\eqref{eq:def_M2z}.

\subsection{Two-scale expansion of the linear equation}\label{subsec:expansion_lineq}
 
We now perform similar computations to solve the linear equation
  \begin{align*}
	 \pa{H^\ep_k -z} u^\eps = f.
 \end{align*}
In what follows, we assume that $z \in \C \setminus \sigma(H^0_k)$, and we denote by $u^0_{k,z, f} := (H^0_k - z)^{-1} f$ the solution of the homogenized equation. We make the Ansatz
\[
    u^\ep(x) = \sum_{i=0}^{+\infty} \ep^i u_i\left( x, \frac x\eps \right).
\]
We obtain the new set of equations
\begin{multline*}
	\sum_{i \ge -2} \ep^i \bigg( - \Delta_y u_{i+2}  +\delta_{i \ge -1} \pa{v -2 \na_{x,k} \cdot \na_y } u_{i+1} \\
	+ \delta_{i \ge 0} \pa{\pa{-\nak^2 + W(x)} u_i - f} \bigg) = 0.
\end{multline*}
Repeating the same arguments as in the previous section (this time, $(H_k^0 - z)$ is invertible), we get (recall that we note $f$ for $f(x,y)$)
 \begin{align*}
	 u_0 =&  u^0_{k,z,f}(x) \\
	 u_1 =& \chi_1 u^0_{k,z,f}(x) - [\pa{H^0_k -z}^{-1} (V_1 u^0_{k,z,f})](x) \\
	 u_2 =& \pa{\chi_2 + A_2  \cdot \na_k}u^0_{k,z,f}(x)   - \chi_1 [\pa{H^0_k - z}^{-1} (V_1 u^0_{k,z,f})](x)  + U_2(x),
 \end{align*}
which suggests that
\[
    u^\eps = \Mz u^0_{k,z,f} + O_{H^1_\per}(\eps^2) \qquad \mbox{(expected)}.
\]

\section{Proofs}
\label{sec:proofs}

We now provide the proofs of the previous results. This section is organized as follows: 
\begin{enumerate}
\item First, we establish in Section~\ref{sec:bounds} bounds on the multiplication operator $\eps^{-1}v_\eps$, from which we deduce that the operator $H^\eps$ is uniformly bounded from below, hence its Bloch fibers $H^\eps_k$ as well (uniformly in $k$); we then derive bounds on $(H^\eps_k-z)^{-1}$ seen as an operator from $H^{-1}_\per(\Omega)$ to $H^1_\per(\Omega)$.
\item In Section~\ref{sec:operator_expansions} we compute expansions of the operators $H^\eps_k\cM^{\eps,(1)}$,  $H^\eps_k\cM^{\eps,(2)}_k$  and 
$H^\eps_k\widetilde \cM^{\eps,(2)}_{z,k}$ leading to equalities such as 
\begin{equation} \label{eq:pre_bounds}
(H_k^\eps - z) \cM^{\eps, (1)} - (H^0_k - z) = \mbox{highly-oscillatory terms}.
\end{equation}
\item We show in Section~\ref{sec:proof_thm_res} that these equalities, together with the bounds on $(H^\eps_k-z)^{-1}$ established in Section~\ref{sec:bounds} and a technical lemma allowing us to control the highly-oscillatory terms, yields the proofs of Theorem~\ref{thm:res} (approximations of the resolvent) and Corollary~\ref{cor:main_LE} (homogenization of the linear equation).
\item In Section~\ref{sec:proof_eig} we prove Theorems~\ref{thm:eigenmode_nondeg} and~\ref{thm:eigenmode_deg} on the homogenization of the eigenvalue problem in the non-degenerate and degenerate cases respectively.
\item We finally prove Theorems~\ref{thm:phys_quantities1} and \ref{thm:phys_quantities2} on the approximations of physical quantities of interest (kinetic/potential energies and integrated density of states) in Section~\ref{sec:QoI}.
\end{enumerate}

\subsection{Preliminary bounds on the operators $\eps^{-1}v_\eps$, $H^\eps$, and $H^\eps_k$}
\label{sec:bounds}

The following lemma collects some useful uniform bounds on $\eps^{-1}v_\eps$, $H^\eps$, and $H^\eps_k$. It shows in particular that $H^\eps$ is uniformly bounded from below.

\begin{lemma}\label{basiclemma}
Let $v$ and $W$ satisfying \eqref{hypo:on_v_W}. 
\begin{enumerate}
	\item \label{lemma:1} There exists a constant $c \in \R_+$ such that for all $a > 0$, $\ep \in (0,1]$, and $u \in H^1(\R^d)$,
	\begin{equation}\label{eq:estim_1}
		\ab{\int_{\R^d} \ep^{-1} v_\eps |u|^2}  \le a \int_{\R^d} \ab{\na u}^2 + c \, (1+a^{-1}) \int_{\R^d} |u|^2 .
\end{equation}

\item \label{lemma:2}  For any $a \in [0,1)$ and $b > 0$, there exist $c_a, C_b \in \R_+$ such that for any $\ep \in (0,1]$, we have
\begin{align} \label{deus}
a (-\Delta) - c_a  & \le H^\ep \le \pa{1+b} (-\Delta) + C_b,
\end{align}
in the sense of quadratic forms.

\item \label{lemma:3} For all $\mu \ge 1-c_{\frac 12}$, there exists a constant $M_\mu \in \R_+$ such that for all 
$\eps \in (0,1]$, we have $H^\ep + \mu \ge 1$ and 
\begin{align}\label{tss}
  \nor{\pa{H^\ep + \mu}^{-1/2}}{L^2_\per \to H^1_\per} =\nor{\pa{H^\ep + \mu}^{-1/2} }{H^{-1}_\per \to L^2_\per} \le M_\mu.
\end{align}

\item \label{lemma:4} 
	There exists a constant $c \in \R_+$ (depending only on $V_0=V+W$) such that for all $k \in \R^d$ and $z \in \C$, we have
\begin{equation}\label{eq:zdep}
 \nor{\pa{1-\nabla_k^2} \pa{H^0_k - z}_\perp^{-1}}{L^2_{\rm per} \to L^2_{\rm per}} \le d_k(z), 
  \end{equation}
  where
 \begin{equation}\label{eq:def_dk}
 d_k(z):=1+  \pa{c+ \ab{z}}  \pa{\dist\pa{z,\sigma\pa{H^0_k}\setminus\{z\}}}^{-1}.
\end{equation}
 \end{enumerate}
\end{lemma}

\begin{proof} Since $v$ and $\Delta_x v$ are in  $\cC^0_\per(\Omega;\cC^0_{\per,0}(\Omega))$, we have $\cu$, $\na_x \cu$, $(\na_x \cdot \na_y) \cu$ and  $\Delta_x \cu$ in $\cC^0_\per(\Omega;\cC^0_{\per,0}(\Omega))$. Using the relation 
 \begin{align*}
 \eps^{-1} v_\ep = \ep \Delta (\chi_1)_\eps - \ep \pa{\Delta_x \cu}_\ep - 2  \pa{\na_x \cdot \na_y \cu}_\ep,
 \end{align*}
we obtain by integration by parts that for all $u \in \cC^\infty_{\rm c}(\R^d)$, 
\begin{align*}
& \int_{\R^d} \eps^{-1}v_\ep \ab{u}^2 \\
& \quad = -\ep \int_{\R^d}\hspace{-0.1cm} \hspace{-0.1cm} \na (\chi_1)_\eps \cdot 2 \re \overline{u} \na u  - \ep \int_{\R^d}\hspace{-0.1cm} \pa{\Delta_x \cu}_\ep \ab{u}^2 - 2 \int_{\R^d}\hspace{-0.1cm} \pa{(\na_x \cdot \na_y) \cu}_\ep \ab{u}^2 \\
& \quad = - \int_{\R^d} ( \pa{\na_y \cu}_\ep + \ep \pa{\na_x \cu}_\ep)  \cdot 2 \re \overline{u} \na u  \\
& \qquad\qquad\qquad\qquad - \ep \int_{\R^d} \pa{\Delta_x \cu}_\ep \ab{u}^2 - 2\int_{\R^d} \pa{(\nabla_x \cdot  \na_y) \cu}_\ep \ab{u}^2.
\end{align*}
Using the inequality 
\begin{align*}
\forall \alpha > 0, \quad \ab{2 \overline{u} \na u} \le \alpha \ab{\na u}^2 + \alpha^{-1}  \ab{u}^2,
\end{align*}
we obtain that for all $\eps \in (0,1]$,
\begin{align*}
& \ab{\int_{\R^d} ( \pa{\na_y \cu}_\ep + \ep \pa{\na_x \cu}_\ep)  \cdot 2 \re \overline{u} \na u }  \\
& \quad \le a \int_{\R^d} \ab{\na u}^2 + \f{1}a \pa{\nor{\na_y \cu}{\cC^0_\per(\cC^0_\per)} +  \nor{\na_x \cu}{\cC^0_\per(\cC^0_\per)}} \int_{\R^d} \ab{u}^2,
\end{align*}
hence \eqref{eq:estim_1} with 
\begin{multline*}
c=\max \Big( \nor{\na_y \cu}{\cC^0_\per(\cC^0_\per)} +  \nor{\na_x \cu}{\cC^0_\per(\cC^0_\per)},\\
\|\Delta_x \cu \|_{\cC^0_\per(\cC^0_\per)}+\| (\nabla_x \cdot  \na_y) \cu\|_{\cC^0_\per(\cC^0_\per)} \Big).
\end{multline*}
We conclude by density that this inequality holds true for all $u \in H^1(\R^d)$.

\medskip

To derive \eqref{deus}, we use \eqref{eq:estim_1} and the fact that $W$ is bounded.

\medskip

The operator norms of a bounded operator and its adjoint are equal, and the inequality \eqref{tss} is equivalent to $(1-\Delta) \le M_\mu^2 \pa{H^\ep + \mu}$ in the sense of quadratic forms, which holds by \eqref{deus}.

\medskip

For all $z \in \C$, we have
	\begin{align}\label{eq:dec}
	 \pa{1-\na_k^2}\pa{H^0_k - z}_\perp^{-1} &= \pa{H^0_k -z+z+1 -V_0} \pa{H^0_k-z}_\perp^{-1} \nonumber \\
	 & = P_{k,z}^{0,\perp} + \pa{z+1 -V_0} \pa{H^0_k-z}_\perp^{-1},
 \end{align}
 where $P_{k,z}^{0,\perp}=1-P_{k,z}^0$ is the orthogonal projector on $\Ker \pa{H^0_k - z}^\perp$ (in particular, $P_{k,z}^{0,\perp}=1$ if $z \notin \sigma(H^0_k)$). 
 Since $V_0$ is a bounded potential, we have
 \begin{multline*}
 \nor{\pa{1-\na_k^2}\pa{H^0_k - z}_\perp^{-1}}{L^2_{\rm per} \to L^2_{\rm per}} \\
 \le 1 + \pa{1+\|V_0\|_{L^\infty}+|z|}  \pa{\dist\pa{z,\sigma\pa{H^0_k}\setminus\{z\}}}^{-1}.
 \end{multline*}
 Hence the result.
\end{proof}

\subsection{Operator expansions} \label{sec:operator_expansions}
We will make repeated use of the following relations. First, for $f : \R^d \times \R^d \to \C$, we have
\begin{equation}\label{eq:nabla_eps}
\begin{cases}
    \nabla (f_\ep) = (\nabla_xf)_\ep + \ep^{-1} (\nabla_yf)_\eps, \\
    \Delta (f_\ep) = (\Delta_xf)_\ep + 2 \ep^{-1} (\nabla_x \cdot \nabla_y f)_\ep + \ep^{-2} (\Delta_yf)_\ep.
\end{cases}
\end{equation}
Second, if $f : \R^d \times \R^d \to \C$ and $A : \R^d \times \R^d \to \C^d$ have suitable regularities, we have in the sense of linear operators on $L^2_\per(\Omega)$,
\begin{align} 
	\nabla_k^2 \pa{f_\eps \bullet} = &  (\Delta_xf)_\ep+2\ep^{-1}((\nabla_x\cdot\nabla_y)f)_\eps + \ep^{-2} (\Delta_yf)_\eps \nonumber \\
	& + 2 \left[ (\nabla_xf)_\ep+\ep^{-1}(\nabla_yf)_\eps \right] \cdot \nabla_k + f_\eps \nabla_k^2,	 \label{eq:commut_1} \\
	\nabla_k^2(A_\eps \cdot \nabla_k \bullet)
 =& \left[ (\Delta_x A)_\ep+2\ep^{-1}\pa{(\nabla_x\cdot\nabla_y)A}_\eps + \ep^{-2}(\Delta_yA)_\eps \right] \cdot \nabla_k  \nonumber \\
&+2  \left[ (\d_xA)_\ep+\ep^{-1}(\d_yA)_\eps \right] : \na_k^{\otimes 2} + (A_\ep\cdot\nabla_k) \nabla_k^2,  \label{eq:commut_2}
\end{align}
where we used the following notation:
\begin{itemize}
\item if $\psi \in H^2_\per(\Omega)$, $\nabla_k^{ \otimes 2} \psi$ is the matrix-valued function of $L^2_\per(\Omega;\C^{d \times d})$ with entries $[\nabla_k^{\otimes 2} \psi]_{ij} = (\partial_{x_i} + i k_i) (\partial_{x_j} + i k_j) \psi$. Note that $[\nabla_k^2\psi](x)$ is the trace of the matrix $[\nabla_k^{\otimes 2} \psi](x) \in \C^{d \times d}$;
\item the derivatives $\d_xA$ and $\d_yA$ of the two-scale vector field $A$ are the two-scale matrix fields
$$
\d_xA(x,y) = \left[\frac{\partial A_i}{\partial x_j}(x,y)\right]_{1 \le i,j \le d} 
\quad \mbox{and} \quad 
\d_yA(x,y) = \left[\frac{\partial A_i}{\partial y_j}(x,y)\right]_{1 \le i,j \le d};
$$
\item the doubly contracted product $P:Q$ is defined as $P:Q=\tr(P^TQ)$.
\end{itemize}

\subsubsection{Expansion of $H^\ep_k \cM^{\eps, (1)}$}
 Recall that
\[
    H^\eps_k = - \nabla_k^2 + \eps^{-1} v_\eps + W(x)
    \quad \text{and} \quad
    \cM^{\eps, (1)} = 1 + \eps \left( \chi_1 \right)_\eps.
\]
Using \eqref{eq:commut_1}, we have the operator equality
\begin{align*}
    & H_k^\eps \cM^{\eps, (1)} \\
    & \quad= (-\nabla_k^2+\eps^{-1}v_\eps+W(x)) \pa{1 + \eps (\chi_1)_\eps}  \nonumber \\
    &\quad= -\nabla_k^2+\eps^{-1}v_\eps+W- \eps\nabla_k^2((\chi_1)_\eps \bullet)+v_\eps(\chi_1)_\eps+\eps (\chi_1)_\eps W  \\
    &\quad=  -\nabla_k^2+\eps^{-1}v_\eps+W
    - \eps\big[ (\Delta_x\chi_1)_\ep+2\ep^{-1}((\nabla_x\cdot\nabla_y)\chi_1)_\eps + \ep^{-2} (\Delta_y\chi_1)_\eps  \\
    & \quad \qquad+ 2 \left( (\nabla_x\chi_1)_\ep+\ep^{-1}(\nabla_y\chi_1)_\eps \right) \cdot \nabla_k + (\chi_1)_\eps \nabla_k^2 \big] +v_\eps(\chi_1)_\eps+\eps(\chi_1)_\eps W.
\end{align*}
Together with the identities
\begin{align*}
    &(\Delta_y \chi_1)_\eps = v_\eps, \qquad \qquad v_\eps(\chi_1)_\eps - 2((\nabla_x \cdot \nabla_y)\chi_1)_\eps = V + (\Delta_y \chi_2)_\eps, \\
   &\qquad \qquad -2(\nabla_y\chi_1)_\eps = (\Delta_yA_2 )_\eps, 
\end{align*}
this gives
\begin{align}
    H_k^\eps \cM^{\eps, (1)} & = H^0_k + \pa{\Delta_y \chi_2}_\eps   + \pa{\Delta_y A_2 }_\eps \cdot \na_k  \nonumber \\
    & \qquad + \eps \left[  (\chi_1)_\eps W - \pa{\Delta_x \cu}_\eps - 2\pa{ \na_x \cu}_\eps \cdot \na_k- (\chi_1)_\eps \na_k^2 \right] . \label{eq:commut_3}
\end{align}
Replacing $W$ by $W - z$, we obtain
\begin{equation}\label{eq:expansionH1}
   \boxed{  ( H_k^\eps - z) \cM^{\eps, (1)} - (H^0_k - z) = \cI^{\eps, (0)}_{k} + \eps \cI^{\eps, (1)}_{z, k}, }
\end{equation}
with
\begin{equation} \label{eq:def_J01}
    \begin{cases}
        \cI^{\eps, (0)}_{k} & := \pa{\Delta_y \chi_2}_\eps   + \pa{\Delta_y A_2 }_\eps \cdot \na_k ,  \\
        \cI^{\eps, (1)}_{z, k} & := (\chi_1)_\eps (W - z) - \pa{\Delta_x \cu}_\eps - 2\pa{ \na_x \cu}_\eps \cdot \na_k- (\chi_1)_\eps \na_k^2.
    \end{cases}
\end{equation}


\subsubsection{Expansion of $H^\ep_k \Md$}

We now perform a similar expansion for the operator $H^\ep_k \Md$. The computations are similar, but more tedious. Recall that (see~\eqref{eq:def_M2})
\[
    \Md := \cM^{\eps, (1)}
    + \ep^2 \left[ \pa{\chi_2}_\eps + \Auep \cdot \na_k \right].
\]
Using \eqref{eq:commut_1}, we have
\begin{align}
    & \ep^2 (H^\ep_k  - z) \cdep  \nonumber\\
    & \quad = -\ep^2 \nabla_k^2((\chi_2)_\eps \bullet)+\eps v_\eps (\chi_2)_\eps + \eps^2 (W -z) (\chi_2)_\eps \nonumber \\
    & \quad =  - \eps^2 (\Delta_x\chi_2)_\ep-2\ep( (\nabla_x\cdot\nabla_y) \chi_2)_\eps - (\Delta_y\chi_2)_\eps  \nonumber \\
    & \qquad - 2 \left[ \ep^2(\nabla_x\chi_2)_\ep+\ep(\nabla_y\chi_2)_\eps \right] \cdot \nabla_k  \nonumber \\
    & \qquad - \eps^2(\chi_2)_\eps \nabla_k^2 +\eps v_\eps (\chi_2)_\eps+ \eps^2 (W -z) (\chi_2)_\eps \nonumber \\
    & \quad =  - (\Delta_y\chi_2)_\eps + \eps \left[ (v \chi_2)_\eps  - 2 \pa{(\na_x \cdot \na_y) \cd}_\eps - 2 \pa{\na_y \cd}_\eps \cdot \na_k \right] \nonumber  \\
    &\qquad + \ep^2 \left[ (W - z) \cdep  - \pa{\Delta_x \cd}_\eps - 2  \pa{\na_x \cd}_\eps \cdot \na_k -\cdep \na_k^2  \right]. \label{eq:expansion_eps_2}
\end{align}
Note that the leading order term $-(\Delta_y \chi_2)_\eps$ is the first term of the operator $\cI^{\eps, (0)}_k$ defined in~\eqref{eq:def_J01}. Note also that the term $(v \chi_2)$ is not in $\cC^0_\per(\Omega;\cC^0_{\per,0}(\Omega))$, since its integral w.r.t. $y$ does not vanish  for all $x \in \Omega$. Actually, its average w.r.t. the $y$ variable is the function $V_1$ defined in~\eqref{eq:def_V1}. All the other two-scale functions in the RHS of~\eqref{eq:expansion_eps_2} are in $\cC^0_\per(\Omega;\cC^0_{\per,0}(\Omega))$.

Similarly, using~\eqref{eq:commut_2}, we have
\begin{align}
    & \ep^2 (H^\ep_k - z) ((A_2 )_\eps \cdot \na_k) \nonumber \\
    & \quad =  -\ep^2 \nabla_k^2 (((A_2 )_\eps \cdot \na_k) \bullet) +\eps v_\eps  ((A_2 )_\eps \cdot \na_k)  + \eps^2 (W - z)  (A_2 )_\eps \cdot \na_k  \nonumber   \\
    & \quad =  - \ep^2 (\Delta_xA_2 )_\eps \cdot \nabla_k - 2 \eps \pa{\nabla_x\cdot\nabla_yA_2 }_\eps \cdot \nabla_k - (\Delta_yA_2 )_\eps \cdot \nabla_k   \nonumber  \\
    &\qquad -2   \left[ \eps^2 (\d_xA_2 )_\ep+\eps (\d_yA_2 )_\eps \right] : \nabla_k^{\otimes 2} -\eps^2 ((A_2 )_\ep\cdot\nabla_k) \nabla_k^2  \nonumber \\
    &\qquad + \eps ( vA_2 )_\eps \cdot \na_k  + \eps^2 (W - z)  (A_2 )_\eps \cdot \na_k \nonumber   \\
    & \quad = -\pa{\Delta_y A_2 }_\eps \cdot \na_k + \eps \left( \left( v A_2  -2 (\na_x \cdot \na_y) \Au \right)_\eps \cdot \na_k - 2\pa{\d_y \Au}_\eps : \nabla_k^{\otimes 2}  \right)  \nonumber  \\
    & \quad + \ep^2 \left( \left( -\Delta_x A_2  + (W - z) A_2  \right)_\eps \cdot \na_k - 2 \pa{\d_x\Au}_\eps : \nabla_k^{\otimes 2} -  \pa{\Auep \cdot \na_k} \na_k^2 \right). \label{eq:last_estim}
\end{align}
Again, the leading order term $-(\Delta_y A_2 )_\eps \cdot \nabla_k$ is the second term of the operator $\cI^{\eps, (0)}_k$ defined in~\eqref{eq:def_J01}. All the two-scale functions in the RHS of \eqref{eq:last_estim} (and in particular the function $(v A_2 )$, see \eqref{eq:integral_vA1}) are in $\cC^0_\per(\Omega;\cC^0_{\per,0}(\Omega))$.

Adding $- \eps V_1$ to both sides of \eqref{eq:expansionH1}-\eqref{eq:def_J01}, replacing the terms $-(\Delta_y \chi_2)_\eps$ and $-(\Delta_y A_2 )_\eps \cdot \nabla_k$ by the expressions given by \eqref{eq:expansion_eps_2} and \eqref{eq:last_estim}, and rearranging the terms, we finally obtain 
\begin{align}\label{eq:expansion_md}
    \boxed{\pa{H^\ep_k  - z} \Md - ((H^0_k+\eps V_1)- z)  =  \eps \cJ^{\eps, (1)}_{k, z} + \eps^2 \cJ^{\eps, (2)}_{k, z},}
\end{align}
where
\begin{equation} \label{eq:def_Ieps}
    \begin{cases}
	    \cJ^{\eps, (1)}_{k,z} &:=  (\alpha^{(1)}_{z})_\eps +  (\beta^{(1)})_\eps \cdot \na_k -  (\chi_1)_\eps \na^2_k - 2(\d_yA_2)_\eps : \nabla_k^{\otimes 2}, \\
	    \cJ^{\eps, (2)}_{k,z} &:=  (\alpha^{(2)}_{z})_\eps +  (\beta^{(2)}_{z})_\eps \cdot \na_k  - (\chi_2)_\eps \na^2_k -2 (\d_xA_2)_\eps : \nabla_k^{\otimes 2} \\
				  &\qquad\qquad\qquad\qquad\qquad \qquad \qquad \qquad \qquad  -  \pa{\Auep \cdot \na_k} \na_k^2,
    \end{cases}
\end{equation}
with
\begin{align*}
&\alpha^{(1)}_{z} = \chi_1 (W - z) - \Delta_x \chi_1 + (v \chi_2 - V_1) - 2 (\nabla_x \cdot \nabla_y)\chi_2 , \\
&  \alpha^{(2)}_{z} =  \chi_2 (W - z) - \Delta_x \chi_2,  \\
& \beta^{(1)} =-2 \na_x \chi_1 - 2 \nabla_y\chi_2 + v A_2  -2 (\na_x \cdot \na_y) \Au, \\
&  \beta^{(2)}_{z} = -2 \na_x \chi_2 - \Delta_xA_2+A_2(W-z).
\end{align*}
All the two-scaled functions involved in \eqref{eq:def_Ieps} are in $\cC^0_\per(\Omega;\cC^0_{\per,0}(\Omega))$.

\subsubsection{Expansion of $H^\ep_k \Mz$}
A similar calculation, using in addition the equality $(H^\ep_k -z)(H^\ep_k -z)^{-1}_\perp = 1-P^0_{k,z}$, where $P^0_{k,z}$ is the orthogonal projection onto $\Ker (H^0_k-z)$, leads to 
\begin{align}\label{eq:expansion_mz}
	\boxed{\pa{H^\ep_k -z} \Mz - \pa{H^0_k -z}  -  \ep P_{k,z}^0 V_1 =  \ep \cK^{\ep,(1)}_{k,z} +  \ep^2 \cK^{\ep,(2)}_{k,z},}
\end{align}
with 
\begin{equation} \label{eq:def_Keps}
    \begin{cases}
	\cK^{\ep,(1)}_{k,z} &:= \cJ^{\eps, (1)}_{z,k} - \cI^{\ep,(0)}_{k} \pa{H^0_k-z}_{\perp}^{-1} V_1,  \\
	\cK^{\ep,(2)}_{k,z} &:=  \cJ^{\eps, (2)}_{z,k} - \cI^{\ep,(1)}_{k,z} \pa{H^0_k -z}_\perp^{-1} V_1,
   \end{cases}
\end{equation}
where the operators $\cI^{\ep,(j)}_{k}$ and $\cJ^{\ep,(j)}_{k}$ are defined in \eqref{eq:def_J01} and \eqref{eq:def_Ieps}.


\subsection{Proof of Theorem~\ref{thm:res} and Corollary~\ref{cor:main_LE} (resolvent and linear equation)}
\label{sec:proof_thm_res}

Let us first establish the following lemma about Sobolev norms of multiplication operators by oscillating functions.

\begin{lemma}\label{lem:geps} 
    There exists a constant $C \in \R_+$ such that for all $\eps \in \N^{-1}$,
    \begin{align} 
    &\forall g \in \cC^0_\per(\Omega;\cC^0_{\per}(\Omega)), \; \; \forall \psi \in L^2_{\rm per}(\Omega),  \quad \|g_\eps \psi\|_{L^2_{\rm per}} \le \|g\|_{\cC^0_\per(\cC^0_\per)} \|\psi\|_{L^2_{\rm per}},  \label{eq:osc_L2} \\
        &\forall g \in \cC^1_\per(\Omega;\cC^0_{\per,0}(\Omega)), \; \forall \psi \in H^1_{\rm per}(\Omega),  \; \|g_\eps \psi\|_{H^{-1}_{\rm per}} \le C \eps \|g\|_{\cC^1_\per(\cC^0_\per)} \|\psi\|_{H^1_{\rm per}}.  \label{eq:osc_H1} 
    \end{align}
\end{lemma}
    In other words, the multiplication operator by $g_\eps$, still denoted by $g_\eps$, satisfies
    \begin{equation}\label{eq:norm_geps}
        \| g_\eps \|_{L^2_\per \to L^2_\per}  \le  \| g \|_{\cC^0_\per(\cC^0_\per)}, \qquad \| g_\eps \|_{H^1_\per \to H^{-1}_\per}  \le C \eps \| g \|_{\cC^1_\per(\cC^0_\per)}.
        \end{equation}
\begin{proof}[Proof of Lemma~\ref{lem:geps}]
    The inequality \eqref{eq:osc_L2} is trivial. To establish~\eqref{eq:osc_H1}, we consider $g \in \cC^1_\per(\Omega;\cC^0_{\per,0}(\Omega))$ and denote by $h_g$ the unique function in $\cC^1_\per(\Omega;\cC^1_{\per,0}(\Omega))$ such that $\Delta_yh_g=g$. We have
    \begin{equation} \label{eq:g_to_f}
    \nabla \cdot \left[  (\nabla_y h_g)_\eps \right] = \left( (\nabla_x \cdot \nabla_y) h_g \right)_\eps + \eps^{-1} (\Delta_y h_g) 
    =  \left( (\nabla_x \cdot \nabla_y) h_g \right)_\eps + \eps^{-1} g_\eps.
    \end{equation}
   It follows that
    \begin{align*}
	    \langle g_\eps \psi,\phi\rangle_{H^{-1}_{\rm per},H^1_{\rm per}} &= \int_\Omega g_\eps \overline{\psi}\phi \\
	& = \eps \int_\Omega \left( \nabla \cdot \left[  (\nabla_y h_g)_\eps \right]  \right) \overline{\psi} \phi - \eps \int_{\Omega} \left( (\nabla_x \cdot \nabla_y) h_g \right)_\eps \overline{\psi} \phi \\
        & = - \eps \int_\Omega (\nabla_y h_g)_\eps \cdot \nabla (\overline{\psi} \phi) - \eps \int_{\Omega} \left( (\nabla_x \cdot \nabla_y) h_g \right)_\eps \overline{\psi} \phi,
 \end{align*}
 so that in view of \eqref{eq:elliptic_reg},
 \begin{align*}
	 & \ab{\langle g_\eps \psi,\phi\rangle_{H^{-1}_{\rm per},H^1_{\rm per}}} \\
     & \quad \le 2 \eps \max \left( \| \nabla_y h_g \|_{\cC^0_\per(\cC^0_\per)},  \| (\nabla_x \cdot \nabla_y) h_g \|_{\cC^0_\per(\cC^0_\per)} \right)  \|\psi\|_{H^1_{\rm per}}\|\phi\|_{H^1_{\rm per}} \\
	& \quad \le 2 \eps \nor{h_g}{\cC^1_\per(\cC^1_\per)}  \|\psi\|_{H^1_{\rm per}}\|\phi\|_{H^1_{\rm per}} \\
	&\quad \le C \eps \nor{g}{\cC^1_\per(\cC^0_\per)}  \|\psi\|_{H^1_{\rm per}}\|\phi\|_{H^1_{\rm per}}.
    \end{align*}
Using
    \begin{equation*} \label{eq:def_norm_Hm1}
	    \| F \|_{H^{-1}_\per} = \sup \left\{ \langle F,\phi\rangle_{H^{-1}_{\rm per},H^1_{\rm per}}, \quad \phi \in H^1_{\rm per}(\Omega), \ \| \phi \|_{H^1_\per} = 1 \right\}
    \end{equation*}
    we obtain \eqref{eq:osc_H1}.
\end{proof}

From the expansions established in the previous section and Lemma~\ref{lem:geps} above, we obtain the following.

\begin{lemma} \label{th:key_th}
Let $v$ and $W$ satisfying \eqref{hypo:on_v_W}. There exists a constant $C \in \R_+$ such that for all $k \in \Omega^*$, $\eps \in  (0,1] \cap \N^{-1}$,  and $z \in \C$, 
    \begin{align}
	    \nor{\pa{H_k^\eps - z}  \Mzu - \pa{H^0_k - z}}{H^2_{\rm per} \rightarrow H^{-1}_{\rm per}} &\le C \, \pa{1+ \ab{z}} \, \ep,  \label{eq:op_approx} \\
        \nor{\pa{H_k^\eps - z} \Md   - \pa{(H^0_k +\eps V_1)- z} }{H^3_{\rm per} \rightarrow H^{-1}_{\rm per}} &\le C\, \pa{1+ \ab{z}} \, \ep^2, \label{eq:op_approx2} \\
        \nor{\pa{H_k^\eps - z}  \Mz - \pa{H^0_k - z} -  \ep P_{k,z}^0 V_1}{H^3_{\rm per} \rightarrow H^{-1}_{\rm per}} &\le C_{k,z} \, \pa{1+ \ab{z}} \, \ep^2, \label{eq:op_approx3}
    \end{align}
    with 
    $$
    C_{k,z}= C(1+\dist (z,\sigma(H^0_k)\setminus\{z\})^{-1}).
    $$
\end{lemma}

\begin{proof} From \eqref{eq:expansionH1}, \eqref{eq:def_J01}, and \eqref{eq:norm_geps}, we get
\begin{multline*}
 \nor{\pa{H_k^\eps - z}  \Mzu - \pa{H^0_k - z}}{H^2_{\rm per} \rightarrow H^{-1}_{\rm per}} \\
 \le  \| \cI^{\eps,(0)}_k \|_{H^2_{\rm per} \rightarrow H^{-1}_{\rm per}} 
 +  \eps \| \cI^{\eps,(1)}_{z,k} \|_{H^2_{\rm per} \rightarrow H^{-1}_{\rm per}} ,
\end{multline*}
with 
\begin{align*}
 \| \cI^{\eps,(0)}_k \|_{H^2_{\rm per} \rightarrow H^{-1}_{\rm per}} &\le   \| (\Delta_y\chi_2)_\eps \|_{H^2_{\rm per} \rightarrow H^{-1}_{\rm per}} +  \| (\Delta_y A_2)_\eps \cdot \nabla_k \|_{H^2_{\rm per} \rightarrow H^{-1}_{\rm per}}  \\ &\le   \| (\Delta_y\chi_2)_\eps \|_{H^1_{\rm per} \rightarrow H^{-1}_{\rm per}} + C  \| (\Delta_y A_2)_\eps \|_{H^1_{\rm per} \rightarrow H^{-1}_{\rm per}}  \\ 
 &\le C \left(  \|\Delta_y\chi_2\|_{\cC^1_\per(\cC^0_\per)} + \|\Delta_y A_2\|_{\cC^1_\per(\cC^0_\per)} \right) \eps \\
 %
  \| \cI^{\eps,(1)}_{z,k} \|_{H^2_{\rm per} \rightarrow H^{-1}_{\rm per}} & \le \| (\chi_1(W-z)-\Delta_x\chi_1)_\eps \|_{H^2_{\rm per} \rightarrow H^{-1}_{\rm per}}  \\
& \quad   + C \|(\nabla_x\chi_1)_\eps\|_{H^1_{\rm per} \rightarrow H^{-1}_{\rm per}} + C \|(\chi_1)_\eps\|_{L^2_{\rm per} \rightarrow H^{-1}_{\rm per}} \\
  & \le C \left( \| \chi_1(W-z)-\Delta_x\chi_1\|_{\cC^0_\per(\cC^0_\per)} +   2 \|\chi_1\|_{\cC^1_\per(\cC^0_\per)} \right).
\end{align*}
Hence \eqref{eq:op_approx}. We obtain \eqref{eq:op_approx2} by proceeding in the same way with \eqref{eq:expansion_md}-\eqref{eq:def_Ieps}. Note that we need to work with the space $H^3_{\rm per}(\Omega)$ in order to lead with the last term in $\cJ^{\eps,(2)}_{z,k}$, in which the input function is differentiated three times before being multiplied by a highly-oscillating function. Likewise, we deduce \eqref{eq:op_approx3}  from \eqref{eq:expansion_mz}-\eqref{eq:def_Keps} using in addition the fact that
  \begin{multline*}
   \| \cI^{\ep,(j)}_{k} \pa{H^0_k-z}_{\perp}^{-1} V_1\|_{H^3_{\rm per} \rightarrow H^{-1}_{\rm per}} \\
   \le  \| \cI^{\ep,(j)}_{k} \|_{H^2_{\rm per} \rightarrow H^{-1}_{\rm per}} 
   \|\pa{H^0_k-z}_{\perp}^{-1} V_1\|_{H^3_{\rm per} \rightarrow H^{2}_{\rm per}},
    \end{multline*}
with
\begin{align*}
&\|\pa{H^0_k-z}_{\perp}^{-1} V_1\|_{H^3_{\rm per} \rightarrow H^{2}_{\rm per}} \\
& \hspace{0.2cm} = 
\| \pa{H^0_k+\mu}^{-1}  \pa{H^0_k+\mu} \pa{H^0_k-z}_{\perp}^{-1}  V_1\|_{H^3_{\rm per} \rightarrow H^{2}_{\rm per}} \\
 & \hspace{0.2cm}= \|\pa{H^0_k+\mu}^{-1}  \pa{H^0_k-z}_{\perp}^{-1} \pa{H^0_k+\mu}  V_1\|_{H^3_{\rm per} \rightarrow H^{2}_{\rm per}} \\
&\hspace{0.2cm}\le  \|\pa{H^0_k+\mu}^{-1}\|_{L^2_{\rm per} \rightarrow H^{2}_{\rm per}} \| \pa{H^0_k-z}_{\perp}^{-1} \pa{H^0_k+\mu}\|_{L^2_{\rm per} \rightarrow L^{2}_{\rm per}} 
\|V_1\|_{H^3_{\rm per} \rightarrow L^{2}_{\rm per}} \\
&\hspace{0.2cm} \le  \|\pa{H^0_k+\mu}^{-1}\|_{L^2_{\rm per} \rightarrow H^{2}_{\rm per}} \| \pa{H^0_k-z}_{\perp}^{-1} \pa{H^0_k+\mu}\|_{L^2_{\rm per} \rightarrow L^{2}_{\rm per}} 
\|V_1\|_{\cC^0_\per(\cC^0_\per)} \\
& \hspace{0.2cm} \le C (1+\dist (z,\sigma(H^0_k)\setminus\{z\})^{-1}),
\end{align*}
for a constant $C \in \R_+$ independent of $k \in \Omega^*$ and $z$ ($\mu$ is here a fixed real number chosen as in Lemma~\ref{basiclemma}).
\end{proof}

To transform \eqref{eq:op_approx}-\eqref{eq:op_approx3} into resolvent estimates, we need uniform bounds on $(H^\eps_k-z)^{-1}$ provided by the following lemma.

\begin{lemma}\label{basiclemma2}
Let $v$ and $W$ satisfying \eqref{hypo:on_v_W}.  
\begin{enumerate}
\item Let $\mu \in \R$ be such that $H^\eps+\mu \ge 1$ for all $\eps \in (0,1]$ (see Lemma~\ref{basiclemma}, assertion 3). Then, there exists $C \in \R_+$ such that for all $k \in \Omega^*$ and $\eps \in \N^{-1}$,
\begin{equation} \label{eq:1st_bound_resolvent}
 \|(H^\eps_k+\mu)^{-1}-(H^0_k+\mu)^{-1}\|_{L^2_{\rm per} \to L^2_{\rm per}} \le C \eps.
\end{equation}
\item Let $\omega^*$ be a compact subset of $\overline{\Omega^*}$, and $\sC$ a compact subset of $\C$ such that
$$
    \dist \left( \sC, \overline{\bigcup_{k \in \omega^*}\sigma(H^0_k)}\right) > 0.
$$
Then, it holds
\begin{equation}\label{eq:CV_spectrum}
\dist\left( \sC, \overline{\bigcup_{k \in \omega^*}\sigma(H^\eps_k)}\right) \mathop{\longrightarrow}_{\substack{\ep \in \N^{-1} \\ \eps \to 0}} 
\dist \left( \sC, \overline{\bigcup_{k \in \omega^*}\sigma(H^0_k)}\right).
\end{equation}
In addition, there exists $\eps_0 > 0$ and $C \in \R_+$ such that
 \begin{equation}\label{eq:res_Hm1_H1}
  \forall k \in \omega^*, \quad \forall z \in \sC, \quad
 \forall \eps \in  (0, \eps_0] \cap \N^{-1}, \quad    \left\| (H_k^{\eps} -z)^{-1} \right\|_{H^{-1}_\per \to H^1_\per} \le C.
    \end{equation}
 \end{enumerate}
\end{lemma}

\begin{proof} Taking $z=\mu$ in \eqref{eq:op_approx}, we obtain that there exists $C \in \R_+$ such that 
\begin{align*}
\forall k \in \Omega^*, \quad \forall \eps \in \N^{-1}, \quad  \nor{\pa{H_k^\eps +\mu}  \Mzu - \pa{H^0_k +\mu}}{H^2_{\rm per} \rightarrow H^{-1}_{\rm per}} \le C  \, \ep.
\end{align*}
Since
\begin{align*}
    & \Mzu \pa{H^0_k +\mu}^{-1} -  \pa{H_k^\eps +\mu} ^{-1} \\
    & \quad = \pa{H_k^\eps +\mu}^{-1}  \left(\pa{H_k^\eps +\mu}  \Mzu - \pa{H^0_k +\mu}\right)  \pa{H_k^0 +\mu}^{-1},
\end{align*}
we obtain that
\begin{multline*}
  \nor{ \Mzu \pa{H^0_k +\mu}^{-1} -  \pa{H_k^\eps +\mu} ^{-1}}{L^2_{\rm per} \rightarrow L^{2}_{\rm per}} \\
 \le \nor{\pa{H_k^0 +\mu}^{-1}}{L^2_{\rm per} \rightarrow H^{2}_{\rm per}} \nor{\pa{H_k^\eps +\mu}^{-1}}{H^{-1}_{\rm per} \rightarrow L^{2}_{\rm per}} C \eps.
 \end{multline*}
 Let us bound these two norms.
 From \eqref{eq:zdep}-\eqref{eq:def_dk}, the first norm is bounded uniformly in $k \in \Omega^*$. The second norm is also bounded uniformly in $k \in \Omega^*$ and $\eps \in \N^{-1}$ in view of  \eqref{tss}:
\begin{align*}
&\nor{\pa{H_k^\eps +\mu}^{-1}}{H^{-1}_{\rm per} \rightarrow L^{2}_{\rm per}} \\
&\qquad \le \nor{\pa{H_k^\eps +\mu}^{-1/2}}{L^{2}_{\rm per} \rightarrow L^{2}_{\rm per}}  
\nor{\pa{H_k^\eps +\mu}^{-1/2}(1-\nabla_k^2)^{1/2}}{L^{2}_{\rm per} \rightarrow L^{2}_{\rm per}} \\ 
& \qquad\qquad\qquad\qquad\qquad\qquad\qquad\qquad\qquad\qquad\quad \times\nor{(1-\nabla_k^2)^{-1/2}}{H^{-1}_{\rm per} \rightarrow L^{2}_{\rm per}} \\
&\qquad \le M_\mu \nor{(1-\nabla_k^2)^{-1/2}}{H^{-1}_{\rm per} \rightarrow L^{2}_{\rm per}}.
\end{align*}
Therefore, we have 
 \begin{align*}
 \nor{ \Mzu \pa{H^0_k +\mu}^{-1} -  \pa{H_k^\eps +\mu} ^{-1}}{L^2_{\rm per} \rightarrow L^{2}_{\rm per}} & \le C \eps,
\end{align*}
for a constant $C \in \R_+$ independent of $k \in \Omega^*$ and $\eps \in \N^{-1}$. Using \eqref{eq:osc_L2} and the fact that $\nor{\pa{H^0_k +\mu}^{-1}}{L^2_{\rm per}} \le 1$, we also get
\begin{align*}
 & \nor{ \pa{H^0_k +\mu}^{-1} -  \pa{H_k^\eps +\mu} ^{-1}}{L^2_{\rm per} \rightarrow L^{2}_{\rm per}} \\
 & \qquad\quad \le 
  \nor{ \Mzu \pa{H^0_k +\mu}^{-1} -  \pa{H_k^\eps +\mu} ^{-1}}{L^2_{\rm per} \rightarrow L^{2}_{\rm per}} \\
 &\qquad\qquad\qquad\qquad\qquad\qquad\qquad\qquad+\eps \nor{(\chi_1)_\eps \pa{H^0_k +\mu}^{-1}}{L^2_{\rm per} \rightarrow L^{2}_{\rm per}} \\
  & \qquad\quad \le C\eps + \eps \|\chi_1\|_{\cC^0_\per(\cC^0_\per)}.
 \end{align*}
This proves \eqref{eq:1st_bound_resolvent}.

\medskip

By the Courant-Fisher min-max principle, the bound \eqref{eq:1st_bound_resolvent} implies that for all $j \in \N^*$, the distance between the $j^{\rm th}$-th eigenvalues of the positive compact self-adjoint operators $(H_k^\eps +\mu)^{-1}$ and $(H_k^0 +\mu)^{-1}$ (the eigenvalues being ranked here in non-increasing order and counting multiplicities) is not larger than $C\eps$. Since the spectra of $H^\eps_k$ and $H^0_k$ are obtained from the spectra of $(H_k^\eps +\mu)^{-1}$ and $(H_k^0 +\mu)^{-1}$ by the transform $(0,1] \ni \lambda \mapsto \lambda^{-1}-\mu \in \R$, we obtain~\eqref{eq:CV_spectrum}.

\medskip

	Using the resolvent formula, we have for all $z$ in the resolvent set of $H^\eps_k$,
	\begin{align*}
	(H^\eps_k-z)^{-1} &= (H^\eps_k+\mu)^{-1} + (z+\mu) (H^\eps_k+\mu)^{-1} (H^\eps_k-z)^{-1} \\
	&= (H^\eps_k+\mu)^{-1} + (z+\mu)  (H^\eps_k+\mu)^{-1/2}  (H^\eps_k-z)^{-1} (H^\eps_k+\mu)^{-1/2}.
	\end{align*}
	We therefore deduce from the uniform bounds~\eqref{tss} that
\begin{align}
	& \| (H^\eps_k-z)^{-1}  \|_{H^{-1}_\per \to H^1_\per} \nonumber \\
	& \qquad \le  \| (H^\eps_k+\mu)^{-1}  \|_{H^{-1}_\per \to H^1_\per}  \nonumber  \\ 
    & \qquad\qquad  +(|z|+|\mu|)  \|(H^\eps_k+\mu)^{-1/2}\|_{L^2_\per \to H^1_\per} \|(H^\eps_k-z)^{-1}\|_{L^2_\per \to L^2_\per} \nonumber\\
    & \qquad \qquad\qquad \qquad\qquad \qquad\qquad \qquad\qquad \qquad \times \|(H^\eps_k+\mu)^{-1/2}\|_{H^{-1}_\per \to L^2_\per} \nonumber \\
	& \qquad \le C (1+(1+|z|) \dist(z,\sigma(H^\eps_k))^{-1}), \label{eq:pre_DOS}
\end{align}
	which, together with~\eqref{eq:CV_spectrum}, gives \eqref{eq:res_Hm1_H1}.
\end{proof}

We are now in position to complete the proof of Theorem~\ref{thm:res}.

\begin{proof}[Proof of Theorem~\ref{thm:res}] The fact that there exists $\eps_0 > 0$ such that, for all $\eps \in (0, \eps_0] \cap \N^{-1}$, $k \in \omega^*$, and  $z \in \sC$, $z \notin \sigma(H^\eps_k) \cup \sigma(H^0_k+\eps V_1)$, results from \eqref{eq:1st_bound_resolvent} and Kato's perturbation theory (in its simplest form since the function $V_1$ is in $\cC^0_\per(\Omega)$ and therefore gives rise to a bounded multiplication operator). We therefore have for all $\eps \in (0, \eps_0] \cap \N^{-1}$ and  $z \in \sC$, 
\begin{align*}
    & \left\| \cM^{\eps, (1)} (H^0_k - z)^{-1} - (H^\eps_k - z)^{-1} \right\|_{L^2_\per \to H^1_\per} \\
    & \qquad\quad \le
    \left\|  (H^\eps_k - z)^{-1} \right\|_{H^{-1}_\per \to H^1_\per} 
    \left\| ( H_k^\eps - z) \cM^{\eps, (1)} - (H^0_k - z)  \right\|_{H^2_\per \to H^{-1}_\per} \\
    &\qquad\qquad\qquad\qquad\qquad\qquad\qquad\qquad\qquad\qquad\times \left\|  (H^0_k - z)^{-1} \right\|_{L^2_\per \to H^2_\per}.
\end{align*}
As $\omega^*$ and $\sC$ are compact, $\dist(\sC,\overline{\bigcup_{k \in \omega^*} \sigma(H_k^0)}) > 0$, and the domain of $H^0$ is equal to $H^2(\R^d)$, the term $ \left\|  (H^0_k - z)^{-1} \right\|_{L^2_\per \to H^2_\per}$ is bounded uniformly in $k \in \omega^*$ and $z \in \sC$. The resolvent estimate \eqref{eq:resest1} then immediately follows from \eqref{eq:op_approx} and \eqref{eq:res_Hm1_H1}.  

\medskip

We deduce \eqref{eq:res_L2} from \eqref{eq:res_Hm1_H1} by observing that for all $k \in \omega^*$, $z \in \sC$ and $\eps \in (0,\eps_0]\cap \N^{-1}$,
\begin{align*}
 & \left\| (H^0_k - z)^{-1} - (H^\eps_k - z)^{-1} \right\|_{L^2_\per \to L^2_\per} \\
 & \qquad \le  \left\| \cM^{\eps, (1)} (H^0_k - z)^{-1} - (H^\eps_k - z)^{-1} \right\|_{L^2_\per \to L^2_\per} \\
 &\qquad\qquad\qquad\qquad\qquad\qquad\qquad+ \eps \left\| (\chi_1)_\eps (H^0_k - z)^{-1}  \right\|_{L^2_\per \to L^2_\per} \\
 & \qquad \le \left\| \cM^{\eps, (1)} (H^0_k - z)^{-1} - (H^\eps_k - z)^{-1} \right\|_{L^2_\per \to H^1_\per} \\
 &\qquad\qquad\qquad\qquad\qquad\qquad\qquad+  \eps \left\| \chi_1\right\|_{\cC^0_\per(\cC^0_\per)} \left\| (H^0_k - z)^{-1}  \right\|_{L^2_\per \to L^2_\per} \\ 
  & \qquad \le C\eps,
\end{align*}
for a constant $C \in \R_+$ independent of $k$, $z$ and $\eps$. 

\medskip

Let us now prove \eqref{eq:resest3}. We have
\begin{align*}
    & \left\| \widetilde\cM^{\eps, (2)} (H^0_k - z)^{-1} - (H^\eps_k - z)^{-1} \right\|_{H^1_\per \to H^1_\per} 
     \le
    \left\|  (H^\eps_k - z)^{-1} \right\|_{H^{-1}_\per \to H^1_\per} \\
    & \qquad \times    \left\| ( H_k^\eps - z) \widetilde \cM^{\eps, (2)} - (H^0_k - z)  \right\|_{H^3_\per \to H^{-1}_\per} 
    \left\|  (H^0_k - z)^{-1} \right\|_{H^1_\per \to H^3_\per}.
\end{align*}
As $V_0 \in \cC^2_\per(\Omega)$ and $W \in W^{1,\infty}_\per(\Omega)$, the boundedness of 
\begin{align*}
\left\|  (H^0_k - z)^{-1} \right\|_{H^1_\per \to H^3_\per},
\end{align*}
uniformly in $k \in \omega^*$ and $z \in \sC$, follows from elementary elliptic regularity results. Together with \eqref{eq:op_approx3} and  \eqref{eq:res_Hm1_H1}, the above inequality therefore leads to~\eqref{eq:resest3}.

\medskip

Lastly, we have 
\begin{align*}
    & \left\| \cM^{\eps, (2)} ((H^0_k+\eps V_1) - z)^{-1} - (H^\eps_k - z)^{-1} \right\|_{H^1_\per \to H^1_\per} \\
    &\qquad \le \left\| ( H_k^\eps - z) \cM^{\eps, (2)} - ((H^0_k +\eps V_1) - z)  \right\|_{H^3_\per \to H^{-1}_\per}\\
    & \quad\qquad\qquad\qquad \times   \left\|  (H^\eps_k - z)^{-1} \right\|_{H^{-1}_\per \to H^1_\per}    \left\|  ((H^0_k+\eps V_1) - z)^{-1} \right\|_{H^1_\per \to H^3_\per},
\end{align*}
and
\begin{align*}
 &\left\|  ((H^0_k+\eps V_1) - z)^{-1} \right\|_{H^1_\per \to H^3_\per} \\
 &\qquad\qquad\le  \left\|  (H^0_k- z)^{-1} \right\|_{H^1_\per \to H^3_\per}  \left\|  (1+\eps V_1(H^0_k- z)^{-1})^{-1} \right\|_{H^1_\per \to H^1_\per} \\
 &\qquad\qquad\le \frac{\dps \left\|  (H^0_k- z)^{-1} \right\|_{H^1_\per \to H^3_\per}}{\dps 1-\eps  \left\|   V_1(H^0_k- z)^{-1} \right\|_{H^1_\per \to H^1_\per}}.
\end{align*}
Using \eqref{eq:op_approx2} and  \eqref{eq:res_Hm1_H1}, we finally obtain \eqref{eq:resest2}.
\end{proof}

%


\subsection{Proofs of Theorems~\ref{thm:unif_CV_eigenvalues}, \ref{thm:eigenmode_nondeg} and ~\ref{thm:eigenmode_deg} (eigenmodes)} 
\label{sec:proof_eig}

We first prove Theorem~\ref{thm:unif_CV_eigenvalues} as a corollary of the previous results. We then focus on Theorem~\ref{thm:eigenmode_deg} (degenerate eigenvalues). 

\subsubsection{Proof of Theorem~\ref{thm:unif_CV_eigenvalues}}
\label{sec:proof_unif_CV_eigenvalues}

From \eqref{deus}, there exists two real constants $c_-$ and $C_+$ such that 
$$
\forall \ell \in \N^*, \quad \forall k \in \Omega^*, \quad \frac 12 \lambda_\ell(-\nabla_k^2) + c_- \le E^\eps_{\ell,k} \le \frac 32 \lambda_\ell(-\nabla_k^2) + C_+,
$$
where $\lambda_\ell(-\nabla_k^2)$ is the $\ell^{\rm th}$ lowest eigenvalue of $-\nabla_k^2$. Since 
\begin{align*}
\min_{k \in \Omega^*} \lambda_\ell(-\nabla_k^2) \underset{\ell \to +\infty}{\longrightarrow} +\infty \qquad \mbox{(as $\ell^{2/d}$)},
\end{align*}
we have \eqref{eq:asympt_eigenvalues}.

\medskip

From \eqref{eq:1st_bound_resolvent} and the Courant-Fischer min-max theorem, there exists a constant $C \in \R_+$ such that for all $\ell \in \N^*$, $\forall k \in \Omega^*$ and  $\eps \in \N^{-1}$,
$$
|(E^\eps_{\ell,k}+\mu)^{-1}-(E^0_{\ell,k}+\mu)^{-1}| \le C \eps,
$$
and thus
$$
|E^\eps_{\ell,k} - E^0_{\ell,k}| \le  C\,  |E^0_{\ell,k}+\mu| \, |E^\eps_{\ell,k}+\mu| \; \eps \le C_\ell \eps,
$$
with $C_\ell = C \max_{k \in \Omega^*} \left( |E^0_{\ell,k}+\mu| \left( \frac 32 \lambda_\ell(-\nabla_k^2) + C_+ \right) \right) < +\infty$.

\subsubsection{Proof of Theorem~\ref{thm:eigenmode_deg}}
\label{sec:proof_eigemode_deg}
Let $E_{\ell,k}^0=\cdots=E_{\ell+m-1,k}^0$ be an eigenvalue of $H^0_k$ of multiplicity $m$. Let $\gamma := \min(E_{\ell,k}^0-E_{\ell-1,k}^0,E_{\ell+m,k}^0-E_{\ell,k}^0) > 0$ be the spectral gap around $E_{\ell,k}^0$ and $\sC$ be the circle in the complex plane with center $E_{\ell,k}^0$ and radius $\delta:=\gamma/2$. Thanks to \eqref{eq:unif_eigenvalues}, we can find $\eps_0 > 0$ such that for all $\eps \in (0,\eps_0] \cap \N^{-1}$, $H^\eps_k$ has exactly $m$ eigenvalues (counting multiplicities) in the  interval $I_\delta := [E_{\ell,k}^0-\delta,E_{\ell,k}^0+\delta]$ and neither $(E_{\ell,k}^0-\delta)$ nor $(E_{\ell,k}^0+\delta)$ are in $\sigma(H^\eps_k)$. The same holds for $(H^0_k+\eps V_1)$ by virtue of Kato's perturbation theory. We therefore have
\begin{align}
& P^0 := \1_{\{E_{\ell,k}^0\}}(H^0_k) = \1_{I_\delta}(H^0_k) = \frac{1}{2i\pi}  \oint_\sC (z-H_k^0)^{-1} \, \d z, \label{eq:contour_int_1}  \\
& P^\eps := \1_{I_\delta}(H^\eps_k) =\frac{1}{2i\pi}   \oint_\sC (z-H_k^\eps)^{-1} \, \d z, \label{eq:contour_int_2}   \\ 
& P^{0, \eps} := \1_{I_\delta}(H^0_k+\eps V_1) = \frac{1}{2i\pi}  \oint_\sC (z-(H_k^0+\eps V_1))^{-1} \, \d z. \label{eq:contour_int_3}
\end{align}
Using \eqref{eq:res_L2}-\eqref{eq:resest2}, there exist $\eps_0 > 0$ and $C \in \R_+$ such that for all $\eps \in (0,\eps_0] \cap \N^{-1}$,
\begin{align*}
& \left\| P^\eps -  P^0  \right\|_{L^2_\per \to L^2_\per} \nonumber \\
&\qquad \le   \left\| \frac{1}{2i\pi}  \oint_\sC (z-H_k^\eps)^{-1} \, \d z -  \frac{1}{2i\pi}  \oint_\sC (z-H_k^0)^{-1} \, \d z  \right\|_{L^2_\per \to L^2_\per} \nonumber \\
& \qquad \le \frac{1}{2\pi}   \oint_\sC \left\|  (z-H_k^\eps)^{-1} -  (z-H_k^0)^{-1}   \right\|_{L^2_\per \to L^2_\per}  \, \d z 
\le C \eps.
\end{align*}
Similarly,
\begin{align*}
& \left\| P^\eps - \cM^{\eps,(1)} P^0 \right\|_{L^2_\per \to H^1_\per} \nonumber \\
& \qquad \le \frac{1}{2\pi}   \oint_\sC \left\|  (z-H_k^\eps)^{-1} - \cM^{\eps,(1)} (z-H_k^0)^{-1}   \right\|_{L^2_\per \to H^1_\per}  \, \d z
 \le C \eps, 
\end{align*}
proving \eqref{eq:deg_evec00}, and
\begin{align*}
& \left\| P^\eps - \cM_k^{\eps,(2)} P^{0, \eps}  \right\|_{H^1_\per \to H^1_\per} \nonumber\\
& \quad \le \frac{1}{2\pi}   \oint_\sC \left\|  (z-H_k^\eps)^{-1} - \cM_k^{\eps,(2)} (z-(H_k^0+\eps V_1))^{-1}   \right\|_{H^1_\per \to H^1_\per}  \, \d z 
\le C \eps^2,
\end{align*}
proving \eqref{eq:deg_evec}. Let us now prove \eqref{eq:deg_evec2}. Using \eqref{eq:def_M2z}, \eqref{eq:contour_int_1}, \eqref{eq:contour_int_2}, we first have
\begin{align*}
& \qquad P^\eps - \widetilde \cM_{k,E^0_{\ell,k}}^{\eps,(2)} P^0   =   \frac{1}{2i\pi}   \oint_\sC \left(  (z-H_k^\eps)^{-1} -  \widetilde \cM_{k,z}^{\eps,(2)}  (z-H_k^0)^{-1} \right) \, \d z \\
&  + \eps (1+\eps(\chi_1)_\eps) \hspace{-0.05cm}\left( \hspace{-0.05cm} (H^0_k-E^0_{\ell,k})^{-1}_\perp V_1 P^0 -  \frac{1}{2i\pi}   \oint_\sC (H_k^0-z)^{-1}V_1 (z-H_k^0)^{-1} \, \d z \hspace{-0.05cm}\right)\hspace{-0.05cm}.
\end{align*}
Since
\begin{multline*}
 \frac{1}{2i\pi}  \oint_\sC (z-H_k^0)^{-1} V_1 (z-H_k^0)^{-1} \, \d z \\
 =   (H^0_k-E^0_{\ell,k})^{-1}_\perp V_1 P^0 + P^0 V_1  (H^0_k-E^0_{\ell,k})^{-1}_\perp,
\end{multline*}
we get
\begin{align*}
P^\eps - \widetilde \cM_{k,E^0_{\ell,k}}^{\eps,(2)}P^0  
& =   \frac{1}{2i\pi}   \oint_\sC \left(  (z-H_k^\eps)^{-1} -  \widetilde \cM_{k,z}^{\eps,(2)}  (z-H_k^0)^{-1} \right) \, \d z \\
& \qquad - \eps (1+\eps(\chi_1)_\eps) P^0 V_1  (H^0_k-E^0_{\ell,k})^{-1}_\perp.
\end{align*}
Multiplying both sides by $P^0$ on the right and using $(H^0_k-E^0_{\ell,k})^{-1}_\perp P^0=0$, we obtain
\begin{multline*}
 \left(P^\eps - \widetilde \cM_{k,E^0_{\ell,k}}^{\eps,(2)}P^0 \right)  P^0   \\
 =  \left(  \frac{1}{2i\pi}   \oint_\sC \left(  (z-H_k^\eps)^{-1} -  \widetilde \cM_{k,z}^{\eps,(2)}  (z-H_k^0)^{-1} \right) \, \d z \right) P^0,
\end{multline*}
and from \eqref{eq:resest3} we obtain
\begin{align}\label{eq:bound_projs_2}
\nor{\left(P^\eps - \widetilde \cM_{k,E^0_{\ell,k}}^{\eps,(2)}P^0 \right)  P^0}{H^1_\per \rightarrow H^1_\per} \le C \ep^2.
\end{align}
Finally, we use that $\| P^0 \|_{L^2_\per \rightarrow H^1_\per}$ is finite, and write $P^0 = P^0 P^0$, to show \eqref{eq:deg_evec2}.

\medskip

Let us finally prove \eqref{eq:deg_eval}. First, by a direct application of Kato's perturbation theory for degenerate eigenvalues, we have
\begin{equation} \label{eq:Kato_deg}
E^{0,\eps}_{\ell+j,k} = E^0_{\ell,k} + \eps \eta_j + O(\eps^2),
\end{equation}
where we recall that $\eta_0 \le \eta_1 \le \cdots \le \eta_{m-1}$ are the eigenvalues of the (diagonal) $m \times m$ matrix with entries $\langle \psi^{0}_{\ell+j,k}, V_1 \psi^{0}_{\ell+j',k}\rangle$. The eigenvalues $E^\eps_{\ell, k} \le \cdots \le E^\eps_{\ell + m - 1, k}$ are also the eigenvalues of the $m \times m$ matrix $\bH^\eps$ with coefficients
\[
    \bH^\eps_{j,j'} := \langle \psi^\eps_{\ell+j, k}, H^\eps_k \psi^\eps_{\ell+j', k} \rangle
     = \langle \psi^\eps_{\ell+j, k}, P^\eps H^\eps_k P^\eps \psi^\eps_{\ell+j', k} \rangle.
\]
Since $H^\eps_k$ defines a bounded quadratic form on $H^1_\per \times H^1_\per$, see~\eqref{deus}, together with~\eqref{eq:deg_evec}, we obtain
\[
    \bH^\eps_{j,j'} =  \langle \psi^\eps_{\ell+j, k}, P^{0, \eps} \cM^{\eps, (2)}_k H^\eps_k \cM^{\eps, (2)}_k P^{0, \eps} \psi^\eps_{\ell+j', k} \rangle + O(\eps^2).
\]
Introducing the (positive) $m \times m$ Gram matrix $\bS$ with coefficients
\[
    \bS_{j,j'} := \langle \psi_{\ell + j, k}^{0, \eps}, \psi_{\ell + j', k}^{\eps} \rangle,
\]
the previous equalities also read
\[
    \bH^\eps = \bS^* \bH^{0, \eps} \bS + O(\eps^2),
\]
with $\bH^{0, \eps}$ the $m \times m$ matrix with coefficients
\[
    \bH^{0, \eps}_{j,j'} := \langle \psi_{\ell + j, k}^{0, \eps}, \cM^{\eps, (2)}_k H^\eps_k \cM^{\eps, (2)}_k \psi_{\ell + j', k}^{0, \eps} \rangle.
\]
We now expand the coefficients of $\bH^{0, \eps}$ using \eqref{eq:def_M2} and \eqref{eq:expansion_md} and the fact that the $\psi^{0, \eps}_{\ell+j,k}$'s are uniformly bounded in $H^3_{\rm per}(\Omega)$, we get
\begin{align*}
\bH_{j,j'}^{0, \eps} & = \left\langle \cM^{\eps,(2)}_k \psi^{0, \eps}_{\ell+j,k} \big| H^\eps_k  \cM^{\eps,(2)}_k \big| \psi^{0, \eps}_{\ell+j',k}\right\rangle \\
&= \left\langle \left( 1 + \eps (\chi_1)_\eps \right)
\psi^{0, \eps}_{\ell+j,k} \big| \left( H^0_k+\eps V_1 +  \eps \cJ^{\eps, (1)}_{k, z} \right) \psi^{0, \eps}_{\ell+j',k}\right\rangle + O(\eps^2)\\
&=E^{0, \eps}_{\ell+j,k}\delta_{jj'}+ \eps \left( E^{0, \eps}_{\ell+j,k} \left\langle (\chi_1)_\eps \psi^{0, \eps}_{\ell+j,k} \big|  \psi^{0, \eps}_{\ell+j',k}\right\rangle
+ \left\langle  \psi^{0, \eps}_{\ell+j,k} \big| \cJ^{\eps,(1)}_{k,z} \psi^{0, \eps}_{\ell+j',k}\right\rangle \right) \\
& \qquad\qquad\qquad + O(\eps^2).
\end{align*}
Using \eqref{eq:def_Ieps}, \eqref{eq:osc_H1}, and again the fact that the $\psi^{0, \eps}_{\ell+j,k}$'s are bounded in $H^3_{\rm per}(\Omega)$, we obtain
\begin{align}
&\left| \left\langle (\chi_1)_\eps \psi^{0, \eps}_{\ell+j,k} \big|  \psi^{0, \eps}_{\ell+j',k}\right\rangle \right| \le 
\left\|  (\chi_1)_\eps \psi^{0, \eps}_{\ell+j,k} \right\|_{H^{-1}_{\rm per}} \left\| \psi^{0, \eps}_{\ell+j',k} \right\|_{H^1_{\rm per}} \le C\eps, \label{eq:chi1_psi0eps_psi0eps} \\
&\left| \left\langle  \psi^{0, \eps}_{\ell+j,k} \big| \cJ^{\eps,(1)}_{k,z} \psi^{0, \eps}_{\ell+j',k}\right\rangle  \right| \le C \eps.  \nonumber
\end{align}
So, introducing the diagonal matrix $\bH_1 = {\rm diag}( \eta_0, \eta_1, \cdots, \eta_{m-1})$, we have
\[
    \bH^{0, \eps} = E^0_{\ell,k} \mathbb{I}_m + \eps \bH_1 + O(\eps^2).
\]
On the other hand, from~\eqref{eq:deg_evec00}, we see that $\bS = \mathbb{I}_m + O(\eps)$. In addition, from~\eqref{eq:deg_evec}, we obtain
\[
    \left\| P^{0, \eps} P^\eps P^{0, \eps} -  P^{0, \eps} \cM^{\eps, (2)}_k P^{0, \eps} \right\| \le C \eps^2,
\]
which gives, using again~\eqref{eq:chi1_psi0eps_psi0eps},
\begin{align}
    \left( \bS \bS^* \right)_{j,j'} & = \langle \psi_{\ell + j, k}^{0, \eps}, \cM^{\eps, (2)}_k \psi_{\ell + j', k}^{0, \eps} \rangle + O(\eps^2) \nonumber \\
    & =  \langle \psi_{\ell + j, k}^{0, \eps},  \left( 1 + \eps (\chi_1)_\eps \right)    \psi^{0, \eps}_{\ell+j',k} \rangle + O(\eps^2) \nonumber \\
    & = \delta_{jj'} + O(\eps^2), \label{eq:SS*}
\end{align}
So $\bS \bS^* = \mathbb{I}_m + O(\eps^2)$ and $\bS^* \bS = \mathbb{I}_m + O(\eps^2)$ as well. This proves that
\begin{align*}
    \bH^\eps & = \bS^* \left(E^0_{\ell,k} \mathbb{I}_m + \eps \bH_1 \right) \bS + O (\eps^2)  \\
	     &= E^0_{\ell,k}\bS^* \bS + \eps \left( \mathbb{I}_m + O(\eps) \right) \bH_1   \left( \mathbb{I}_m + O(\eps) \right) + O(\eps^2) \\
    & =  E^0_{\ell,k}\mathbb{I}_m + \eps \bH_1 + O(\eps^2),
\end{align*}
and the result follows.

 \subsubsection{Proof of Theorem~\ref{thm:eigenmode_nondeg}}
 
It is easily seen that under the assumptions of Theorem~\ref{thm:eigenmode_nondeg}, the bounds \eqref{eq:deg_evec00}-\eqref{eq:deg_evec2} established in the previous section hold uniformly in $k \in \omega^*$. Applying these results in the special case when $E^0_{\ell,k}$ is non-degenerate ($m=1$), we obtain that for $\eps \in \N^{-1}$ small enough, $E^\eps_{\ell,k}$ and $E^{0, \eps}_{\ell,k}$ are non-degenerate for all $k \in \omega^*$, and both converge to $E^0_{\ell,k}$ when $\eps \to 0$, uniformly in $k \in \omega^*$.  We deduce from \eqref{eq:deg_evec00}-\eqref{eq:deg_evec2} that for $k \in \omega^*$ and $\eps \in (0,\eps_0] \cap \N^{-1}$, 
 \begin{align}
 \left\| \langle \psi_{\ell,k}^\eps, \psi^0_{\ell,k} \rangle_{L^2_{\rm per}} \psi_{\ell,k}^\eps - \psi^0_{\ell,k} \right\|_{L^2_{\rm per}} \le C \eps, \label{eq:estim10}  \\
 \left\| \langle \psi_{\ell,k}^\eps, \psi^0_{\ell,k} \rangle_{L^2_{\rm per}} \psi_{\ell,k}^\eps - \cM^{\eps,(1)} \psi^0_{\ell,k} \right\|_{H^1_{\rm per}} \le C \eps,  \label{eq:estim11}  \\
 \left\| \langle \psi_{\ell,k}^\eps, \psi^{0, \eps}_{\ell,k} \rangle_{L^2_{\rm per}} \psi_{\ell,k}^\eps -  \cM_k^{\eps,(2)} \psi^{0, \eps}_{\ell,k}  \right\|_{H^1_{\rm per}} \le C \eps^2,  \label{eq:estim12} \\
 \left\| \langle \psi_{\ell,k}^\eps, \psi^0_{\ell,k} \rangle_{L^2_{\rm per}} \psi_{\ell,k}^\eps -\widetilde \cM_{E^0_{\ell,k},k}^{\eps,(2)} \psi^0_{\ell,k} \right\|_{H^1_{\rm per}} \le C \eps^2. \label{eq:estim13}
  \end{align}
  From \eqref{eq:estim10} and the orientation condition $\langle \psi_{\ell,k}^\eps, \psi^0_{\ell,k} \rangle_{L^2_{\rm per}} \in \R_+$, we get
  \begin{align*}
 \frac 12 \|\psi_{\ell,k}^\eps -  \psi^0_{\ell,k}\|_{L^2_{\rm per}}^2 & = 1- \langle \psi_{\ell,k}^\eps, \psi^0_{\ell,k} \rangle_{L^2_{\rm per}}   \le  1- \langle \psi_{\ell,k}^\eps, \psi^0_{\ell,k} \rangle_{L^2_{\rm per}}^2 \\
 & =\left\| \langle \psi_{\ell,k}^\eps, \psi^0_{\ell,k} \rangle_{L^2_{\rm per}} \psi_{\ell,k}^\eps - \psi^0_{\ell,k} \right\|_{L^2_{\rm per}}^2 \le C^2\eps^2.
  \end{align*}
  Combining with \eqref{eq:estim11} and  \eqref{eq:estim13}, we get
  \begin{equation}\label{eq:bound_unnorm}
  \left\|  \psi_{\ell,k}^\eps - \cM^{\eps,(1)} \psi^0_{\ell,k} \right\|_{H^1_{\rm per}} \le C \eps \quad \mbox{and} \quad 
   \left\|  \psi_{\ell,k}^\eps - \widetilde \cM_{E^0_{\ell,k},k}^{\eps,(2)} \psi^0_{\ell,k}  \right\|_{H^1_{\rm per}} \le C \eps^2.
  \end{equation}
  In particular, we have
  \[
    \left\| \cM^{\eps,(1)} \psi^0_{\ell,k} \right\|_{L^2_\per} = 1 + O(\eps), \quad \text{and} \quad
    \left\|  \widetilde \cM_{E^0_{\ell,k},k}^{\eps,(2)} \psi^0_{\ell,k}  \right\|_{L^2_\per} = 1 + O(\eps^2),
  \]
  from which we get
  $$
  \left\|  \psi_{\ell,k}^\eps - \psi^{\eps,(1)}_{\ell,k} \right\|_{H^1_{\rm per}} \le C \eps \quad \mbox{and} \quad 
   \left\|  \psi_{\ell,k}^\eps -\widetilde \psi^{\eps,(2)}_{\ell,k}   \right\|_{H^1_{\rm per}} \le C \eps^2,
  $$
  where $\psi^{\eps,(1)}_{\ell,k} $ and $\widetilde \psi^{\eps,(2)}_{\ell,k}$ are the $L^2_{\rm per}(\Omega)$-normalizations of $\cM^{\eps,(1)} \psi^0_{\ell,k}$ and $\widetilde \cM_{E^0_{\ell,k},k}^{\eps,(2)} \psi^0_{\ell,k}$ respectively introduced in \eqref{eq:first_order_app_eig} and \eqref{eq:second_order_app_eig2}.

\medskip

From~\eqref{eq:SS*} in the case where the eigenvalue is non-degenerate, we get
\[
    \ab{ \langle \psi_{\ell,k}^\eps, \psi^{0, \eps}_{\ell,k} \rangle_{L^2_{\rm per}} }^2 = 1 + O(\eps^2).
\]
Taking square-roots and using the orientation condition $\langle \psi_{\ell,k}^\eps, \psi^{0, \eps}_{\ell,k} \rangle_{L^2_{\rm per}} \in \R_+$, we get
\[
     \langle \psi_{\ell,k}^\eps, \psi^{0, \eps}_{\ell,k} \rangle_{L^2_{\rm per}}  = 1 + O(\eps^2).
\]
Together with \eqref{eq:estim12}, this yields $\left\|  \psi_{\ell,k}^\eps  - \psi^{\eps,(2)}_{\ell,k} \right\|_{H^1_{\rm per}} \le C \eps^2$, where $\psi^{\eps,(2)}_{\ell,k}$ is defined in \eqref{eq:second_order_app_eig}.

\medskip

The bounds on the eigenvalues in \eqref{eq:bounds_order_1}-\eqref{eq:bounds_order_2} are straightforward consequences of the following well-known result.

\begin{lemma}[Approximations of eigenvalues]\label{lem:approx_eigenvalues}
	Let $H$ be an operator on $\ld(\Omega)$ with domain $H^2_{\rm per}(\Omega)$ and form domain $H^1_\per(\Omega)$ such that 
	\begin{equation}\label{eq:2side_bound}
	c_-  \le H  \le C_+ (1-\Delta)
	\end{equation}
	in the sense of quadratic forms, for some constants $c_- \in \R$ and $C_+ \in \R_+$. Let $E \in \R$ be an eigenvalue of $H$ and $\p \in H^2_\per$ an associated normalized eigenvector, i.e.  $\nor{\p}{L^2_\per} = 1$ and $H \p = E \p$. Then, for all $\vp \in H^1_\per$ such that $\nor{\vp}{L^2_\per} = 1$,
 \begin{align} \label{eq:bound_eig_ell}
	 \ab{E - \langle \vp | H| \vp \rangle} \le \max \big( | E-c_-|,C_++\max(-E,0) \big) \nor{\p - \vp}{H^1_\per}^2.
 \end{align}
\end{lemma}
This shows that if we have an $H^1_\per$-approximation $\vp$ of $\p$ of order $\ep^{\alpha}$, then $\ps{\vp| H | \vp}$ (i.e. the Rayleigh quotient since $\vp$ is normalized) provides an approximation of $E$ of order $\ep^{2\alpha}$.
 \begin{proof}
	 We have
 \begin{align*}
	 2 \re \ps{\p, \p - \vp}_{L^2_{\rm per}} = 2 - 2 \re \ps{\p,\vp}_{L^2_{\rm per}} = \nor{\p - \vp}{\ld}^2,
 \end{align*}
 hence
 \begin{align*}
	  \langle \vp | H| \vp \rangle - E &=  \langle \p -\vp | H| \p- \vp \rangle  - 2 \re  \langle \p  | H| \p- \vp \rangle   \\
	 & = \langle \p -\vp | H| \p- \vp \rangle - 2 E \re \ps{\p, \pa{\p - \vp}}_{L^2_{\rm per}} \\
	 & = \langle \p -\vp | H| \p- \vp \rangle - E \nor{\p - \vp}{\ld}^2 \\
	 & = \langle \p -\vp | H-c_- | \p- \vp \rangle - \pa{E -c_-}\nor{\p - \vp}{\ld}^2 .
 \end{align*}
 We therefore have 
 $$
  - (E-c_-) \nor{\p - \vp}{\ld}^2   \le  \langle \vp | H| \vp \rangle - E \le C_+ \nor{\p - \vp}{H^1_\per}^2 - E \nor{\p - \vp}{\ld}^2
 $$
 hence \eqref{eq:bound_eig_ell}.
  \end{proof}
  
Let us finally prove \eqref{eq:second_order_app}. For this purpose, we need the following lemma on oscillatory integrals.

\begin{lemma} For $f \in L^2_{\rm per}(\Omega)$, $g \in \cC^0_{\rm per}(\Omega;\cC^0_\per(\Omega))$,  and $\eps \in \N^{-1}$, we set
\begin{align*}
I_\eps(f,g) &:= \int_\Omega g_\eps f = \int_\Omega g(x,\eps^{-1}x)  f(x) \, \d x, \\
I_0(f,g) &:= \int_\Omega \left( \fint_\Omega g(x,y) \, \d y \right) f(x) \, \d x.
\end{align*}
Then, for all $m \in \N$, there exists $C_m \in \R_+$ such that for all $f \in H^{m}_{\rm per}(\Omega)$ and $g \in C^{m}_{\rm per}(\Omega; \cC^0_{\rm per}(\Omega))$,
\begin{equation}\label{eq:osc_int}
\left| I_\eps(f,g) - I_0(f,g) \right| \le C_m \; \eps^{m} \, \|f\|_{H^{m}_{\rm per}} 
\|g\|_{C^{m}_{\rm per}(\cC^0_{\rm per})}.
\end{equation}
\end{lemma}
  
\begin{proof}
For $m=0$, the inequality \eqref{eq:osc_int} (with $C_0=2|\Omega|^{1/2}$) is a straightforward consequence of the Cauchy-Schwarz inequality. Let us first prove \eqref{eq:osc_int} for $m=1$. Setting
$G(x):=\fint_\Omega g(x,y) \, \d y$, and $h_g$ the unique solution to $\Delta_yh(x,y)=g(x,y)-G(x)$ in $C^{1}_{\per}(\Omega;\cC^1_{\per,0}(\Omega))$, we have 
$$
 \nabla \cdot \left[  (\nabla_y h_g)_\eps \right] = \left( (\nabla_x \cdot \nabla_y) h_g \right)_\eps + \eps^{-1}(\Delta_y h_g)_\eps 
    =  \left( (\nabla_x \cdot \nabla_y) h_g \right)_\eps + \eps^{-1} (g_\eps-G) 
$$
and therefore, as 
\begin{align*}
\int_\Omega \nabla \cdot \left[  (\nabla_y h_g)_\eps \right]  f=- \int_\Omega (\nabla_yh_g)_\eps \cdot \nabla f=- I_\eps (\nabla f, \nabla_yh_g),
\end{align*}
then
\begin{align}
I_\eps(f,g)&=  I_0(f,g) -\eps I_\eps(\nabla f, \nabla_yh_g) -\eps I_\eps(f,(\nabla_x \cdot \nabla_y)h_g).  
\end{align}
Thus, using the inequality \eqref{eq:osc_int} for $m=0$, we get using \eqref{eq:elliptic_reg},
\begin{align*}
& \left| I_\eps(f,g) - I_0(f,g) \right|  \le \eps \left| I_\eps(\nabla f, \nabla_yh_g) + I_\eps(f,(\nabla_x \cdot \nabla_y)h_g) \right| \\
&\qquad  \le C_0 \eps \left( \|\nabla f\|_{L^2_\per} \|\nabla_yh_g\|_{\cC^0_\per(\cC^0_\per)} +  \| f \|_{L^2_\per} \|(\nabla_x \cdot \nabla_y)h_g\|_{\cC^0_\per(\cC^0_\per)} \right) \\
&\qquad \le C_1 \eps \|f\|_{H^1_\per} \|g\|_{\cC^1_\per(\cC^0_\per)}.
\end{align*}
This proves the result for $m=1$. We conclude using an elementary recursion argument.
\end{proof}
  
We now apply the above lemma to expand in powers of $\eps$ the quantity 
$$
\widetilde E^{\eps,(2)}_{\ell,k}:=\langle \widetilde \psi^{\eps,(2)}_{\ell,k} | H^\eps_k | \widetilde \psi^{\eps,(2)}_{\ell,k}\rangle
= \frac{\langle \widetilde\cM^{\eps,(2)}_{E^0_{\ell,k},k}  \psi^{0}_{\ell,k} | H^\eps_k | \widetilde \cM^{\eps,(2)}_{E^0_{\ell,k},k}  \psi^{0}_{\ell,k}\rangle}{\|\widetilde\cM^{\eps,(2)}_{E^0_{\ell,k},k}  \psi^{0}_{\ell,k} \|_{L^2_{\rm per}}^2} , 
$$
which, according to \eqref{eq:bounds_order_2}, is equal to $E^\eps_{\ell,k} + O(\eps^4)$. We have 
\begin{multline}\label{eq:tildeMpsi}
\widetilde\cM^{\eps,(2)}_{E^0_{\ell,k},k}  \psi^{0}_{\ell,k}  \\
= \psi^{0}_{\ell,k} + \eps \phi^{1}_{\ell,k} + \eps (\chi_1)_\eps (\psi^{0}_{\ell,k} + \eps \phi^{1}_{\ell,k})
+ \eps^2 \left( (\chi_2)_\eps \psi^0_{\ell,k} + (A_2)_\eps \cdot \nabla_k \psi^0_{\ell,k} \right),
\end{multline}
where 
$$
\phi^1_{\ell,k} := - \left( H^0_k-E^0_{\ell,k} \right)_\perp^{-1} (V_1\psi^0_{\ell,k}),
$$
and, using \eqref{eq:expansion_mz} with $z=E^0_{\ell,k}$,
\begin{align*}
& H^\eps_k \widetilde \cM^{\eps,(2)}_{E^0_{\ell,k},k}  \psi^{0}_{\ell,k} =
E^0_{\ell,k} \widetilde \cM^{\eps,(2)}_{E^0_{\ell,k},k}  \psi^{0}_{\ell,k}  + \eps \langle \psi^0_{\ell,k} | V_1 |  \psi^0_{\ell,k} \rangle  \psi^0_{\ell,k} \\
& \qquad\qquad\quad + \eps \left( \cJ^{\eps,(1)}_{k,E^0_{\ell,k}} \psi^0_{\ell,k} + \cI_k^{\eps,(0)} \phi^1_{\ell,k} \right)+ \eps^2 \left( \cJ^{\eps,(2)}_{k,E^0_{\ell,k}} \psi^0_{\ell,k} + \cI_{k,E^0_{\ell,k}}^{\eps,(1)} \phi^1_{\ell,k} \right).
\end{align*}
We therefore obtain, using the fact that $\phi^1_{\ell,k} \in (\psi^0_{\ell,k})^\perp$, 
\begin{multline*}
 \|\widetilde\cM^{\eps,(2)}_{E^0_{\ell,k},k}  \psi^{0}_{\ell,k} \|_{L^2_{\rm per}}^2 
 = 
1 + \eps^2 \|\phi^1_{\ell,k}\|_{L^2_{\rm per}}^2 + \eps^2  \|(\chi_1)_\eps \psi^0_{\ell,k}\|_{L^2_{\rm per}}^2  \\
 + 2 \eps^3 {\rm Re} \left( \int_\Omega (\chi_1^2)_\eps \overline{\psi^0_{\ell,k}} \phi^1_{\ell,k} + \int_\Omega (\chi_1\chi_2)_\eps |\psi^0_{\ell,k}|^2 
+ \int_\Omega \overline{\psi^0_{\ell,k}} (\chi_1A_2)_\eps \cdot\nabla \psi^0_{\ell,k} \right) \\
 + O(\eps^4),
\end{multline*}
and 
  \begin{align*}
& \langle \widetilde\cM^{\eps,(2)}_{E^0_{\ell,k},k}  \psi^{0}_{\ell,k} | H^\eps_k |  \widetilde\cM^{\eps,(2)}_{E^0_{\ell,k},k}  \psi^{0}_{\ell,k}\rangle \\
& \qquad = E^0_{\ell,k}
\|\widetilde\cM^{\eps,(2)}_{E^0_{\ell,k},k}  \psi^{0}_{\ell,k} \|_{L^2_{\rm per}}^2 + \eps \langle \psi^0_{\ell,k} | V_1 |  \psi^0_{\ell,k} \rangle \\
& \qquad  + \eps^2 \left\langle (\chi_1)_\eps \psi^0_{\ell,k} \bigg| \cJ^{\eps,(1)}_{k,E^0_{\ell,k}} \psi^0_{\ell,k} + \cI_k^{\eps,(0)} \phi^1_{\ell,k}  \right\rangle  \\
& \qquad + \eps^3 \left\langle (\chi_1)_\eps \psi^0_{\ell,k} \bigg|  \cJ^{\eps,(2)}_{k,E^0_{\ell,k}} \psi^0_{\ell,k} + \cI_{k,E^0_{\ell,k}}^{\eps,(1)} \phi^1_{\ell,k}  \right\rangle \\
&\qquad + \ep^3 \left\langle (\chi_1)_\eps \phi^1_{\ell,k} + (\chi_2)_\eps \psi^0_{\ell,k} + (A_2)_\eps \cdot \nabla_k \psi^0_{\ell,k} \bigg| \cJ^{\eps,(1)}_{k,E^0_{\ell,k}} \psi^0_{\ell,k} + \cI_k^{\eps,(0)} \phi^1_{\ell,k}  \right\rangle  \\
& \qquad + O(\eps^4).
\end{align*}
The terms in the above expansions of 
\begin{align*}
\|\widetilde\cM^{\eps,(2)}_{E^0_{\ell,k},k}  \psi^{0}_{\ell,k} \|_{L^2_{\rm per}}^2 \qquad \text{and} \qquad \langle \widetilde\cM^{\eps,(2)}_{E^0_{\ell,k},k}  \psi^{0}_{\ell,k} | H^\eps_k |  \widetilde\cM^{\eps,(2)}_{E^0_{\ell,k},k}  \psi^{0}_{\ell,k}\rangle
\end{align*}
are sums of elementary terms of the form $I_\eps(f,g)=\int_\Omega g_\eps f$. Using \eqref{eq:osc_int} with $m=4$, we obtain that each of these elementary terms are of the form $C+O(\eps^4)$ where $C \in \C$ is a constant. We conclude that $\widetilde E^{\eps,(2)}_{\ell,k}$ can be expanded in powers of $\eps$:
$$
\widetilde E^{\eps,(2)}_{\ell,k} = E^0_{\ell,k} + \ep \langle \psi^0_{\ell,k} | V_1 |  \psi^0_{\ell,k} \rangle + \eps^2 \eta_{\ell,k}^{(2)} + \eps ^3 \eta_{\ell,k}^{(2)}  + O(\ep^4).
$$
Hence \eqref{eq:second_order_app} since we already know that $E^\eps_{\ell,k} =  \widetilde E^{\eps,(2)}_{\ell,k} + O(\eps^4)$.

\subsection{Proof of Theorems~\ref{thm:phys_quantities1} and~\ref{thm:phys_quantities2} (quantities of interest)}

\label{sec:QoI}

\subsubsection{Proof of Theorem~\ref{thm:phys_quantities1} (kinetic and potential energies)} 
\label{sec:proof_phys_quantities1}

In this section we use the notation $\p^\eps := \pE$, $\p^0 := \pz$, and $\phi^1:=\phi^{1}_{\ell,k}$  to alleviate the reading. From \eqref{eq:bound_unnorm}, and setting
\begin{align*}
\vp^\ep := \psi^0+\eps\phi^1+\eps(\chi_1)_\eps  (\psi^0+\eps\phi^1) + \eps^2 \left( (\chi_2)_\eps \psi^0+(A_2)_\eps \cdot\nabla_k\psi^0 \right)
\end{align*}
we get $\psi^\eps= \vp^\ep+ O_{H^1_{\rm per}}(\eps^2)$, so that, 
\begin{align*}
 T^\eps_{\ell,k} =  \int_\Omega \ab{ \na_k \vp^\ep }^2 + O(\eps^2)\qquad
V^\eps_{\ell,k} = \int_\Omega (\eps^{-1} v_\eps+W) \ab{\vp^\ep}^2 + O(\eps^2).
\end{align*}
Using \eqref{eq:osc_int} and the equalities
\begin{align*}
&\na_k \left( \psi^0+\eps\phi^1+\eps(\chi_1)_\eps  (\psi^0+\eps\phi^1) + \eps^2 \left( (\chi_2)_\eps \psi^0+(A_2)_\eps \cdot\nabla_k\psi^0 \right) \right) \\
&\qquad = \na_k\psi^0+\eps\na_k\phi^1+\eps(\chi_1)_\eps  \na_k\psi^0+ \eps (\nabla_x\chi_1)_\eps  \psi^0  + (\nabla_y\chi_1)_\eps (\psi^0+\eps\phi^1)   \\
& \qquad \qquad +\eps \left( (\nabla_y\chi_2)_\eps \psi^0+(\d_yA_2)_\eps \cdot\nabla_k\psi^0 \right) +O_{L^2_{\rm per}}(\eps^2),
\end{align*}
\begin{align*}
{\rm Re} \langle \na_k\psi^0|\na_k\phi^1 \rangle &= {\rm Re}  \langle -\na_k^2\psi^0|\phi^1 \rangle = - \int_\Omega (V+W-E^0_{\ell,k})  {\rm Re}(\overline{\psi^0}\phi^1) \\
&=  - \int_\Omega (V+W)  {\rm Re}(\overline{\psi^0}\phi^1) ,
\end{align*}
and
\begin{align*}
& \int_\Omega |\nabla_y\chi_1(x,y)|^2 \, \d y = - V(x), \quad \int_\Omega \chi_1(x,y) \nabla_y\chi_1(x,y) \, \d y = 0, \\
& \int_\Omega \nabla_x \chi_1(x,y) \nabla_y\chi_1(x,y) \, \d y = -F(x), \\
&\int_\Omega \nabla_x \chi_1(x,y) \nabla_y\chi_2(x,y) \, \d y = 2F(x)-G(x), \\
& \int_\Omega (\nabla_x \chi_1(x,y))^T \d_yA_2(x,y) \, \d y = 0, \\
&\int_\Omega v(x,y) \chi_1(x,y)^2 \, \d y = G(x), \quad \int_\Omega v(x,y) \chi_2(x,y) \, \d y = G(x) - 2F(x), \\
&  \int_\Omega v(x,y) A_2(x,y) \, \d y = 0,
\end{align*}
we obtain
\begin{align*}
T^\eps_{\ell,k}& =  \| \na_k\psi^0+\eps\na_k\phi^1  \|_{L^2_{\rm per}}^2+   \| (\nabla_y\chi_1)_\eps(\psi^0+\eps\phi^1)  \|_{L^2_{\rm per}}^2 \\
	       & \qquad + 2 \eps \, {\rm Re} \left( \langle (\nabla_y\chi_1)_\eps \psi^0 | (\chi_1)_\eps  \na_k\psi^0+  (\nabla_x\chi_1)_\eps  \psi^0  + (\nabla_y\chi_2)_\eps \psi^0\rangle \right)  \\
	       & \qquad + 2 \eps \, {\rm Re} \left( \langle (\nabla_y\chi_1)_\eps \psi^0 | (\d_yA_2)_\eps \cdot\nabla_k\psi^0 \rangle \right) + O(\eps^2) \\
& = T^0_{\ell,k} + 2\eps  \, {\rm Re}(\langle \na_k\psi^0|\na_k\phi^1 \rangle) + \int_\Omega (|\nabla_y\chi_1|^2)_\eps \left(|\psi^0|^2+ 2\eps \, {\rm Re}(\overline{\psi^0}\phi^1)\right)  \\ 
& \qquad + 2\eps  \, {\rm Re} \left( \int_\Omega  \overline{\psi^0} (\chi_1 \nabla_y\chi_1)_\eps \cdot  \nabla_k \psi^0 \right) + 2\eps  \int_\Omega (\nabla_x\chi_1 \cdot \nabla_y\chi_1)_\eps  |\psi^0|^2 \\
&\qquad+  2 \eps \int_\Omega (\nabla_y\chi_1 \cdot \nabla_y\chi_2)_\eps  |\psi^0|^2 + 2\eps  \, {\rm Re} \left( \int_\Omega  \overline{\psi^0} ((\nabla_y\chi_1)^T \d_yA_2)_\eps \nabla_k \psi^0 \right) \\
& \qquad\qquad\qquad\qquad+ O(\eps^2) \\
&= T^0_{\ell,k} - V_{\ell,k} \\
& + \eps \left(2 {\rm Re}(\langle \na_k\psi^0|\na_k\phi^1 \rangle) - 2 \int_\Omega V {\rm Re}(\overline{\psi^0}\phi^1) + 2 \int_\Omega (F-G) |\psi^0|^2  \right)  + O(\eps^2)\\
&= T^0_{\ell,k} - V_{\ell,k} + \eps \left(2 \int_\Omega (F-G) |\psi^0|^2 - 2 \int_\Omega (2V+W)   {\rm Re}(\overline{\psi^0}\phi^1)   \right) + O(\eps^2)
\end{align*}
and
\begin{align*}
V^\eps_{\ell,k}&= \int_\Omega W \left|\psi^0+\eps\phi^1 \right|^2 +  \eps \int_\Omega (v\chi_1^2)_\eps |\psi^0|^2 + 2 \int_\Omega (v\chi_1)_\eps \left|\psi^0+\eps\phi^1 \right|^2 
\\ & \qquad + 2 \eps \int_\Omega (v\chi_2)_\eps |\psi^0|^2 + 2 \eps \int_\Omega {\rm Re}(\overline{\psi^0} (vA_2)_\eps \cdot \nabla_k\psi^0)  + O(\eps^2) \\
&= \int_\Omega W \left|\psi^0 \right|^2 + 2 \eps \int_\Omega W \, {\rm Re}(\overline{\psi^0}\phi^1) + 2 \int_\Omega V  \left|\psi^0 \right|^2 + 4 \eps \int_\Omega V \, {\rm Re}(\overline{\psi^0}\phi^1) \\ 
& \qquad + \eps \int_\Omega G|\psi^0|^2 + \eps \int_\Omega (2G-4F)|\psi^0|^2 + O(\eps^2) \\
&= V^0_{\ell,k} + V_{\ell,k} + \eps \left( \int_\Omega (3G-4F) |\psi^0|^2 + 2 \int_\Omega (2V+W) \, {\rm Re}(\overline{\psi^0}\phi^1)  \right) \\
   &\qquad\qquad\qquad\qquad+ O(\eps^2).
\end{align*}

To show \eqref{eq:H2ep_approx}, we can write the Schrödinger equation satisfied by $\p^\ep_{\ell,k}$ and multiply it by $\ep$, yielding
\begin{align*}
\ep \Delta \p^\ep_{\ell,k} = v_\ep \p^\ep_{\ell,k} -\ep2ik \cdot \na \p^\ep_{\ell,k} + \ep \pa{\ab{k}^2 + W - E^\ep_{\ell,k}}  \p^\ep_{\ell,k}.
\end{align*}
Taking the square and integrating, we obtain
\begin{align*}
\ep^2 \int_\Omega \ab{\Delta \p^\ep_{\ell,k} }^2 = \int_\Omega v_\ep^2 \ab{\p^\ep_{\ell,k} }^2 + O(\ep),
\end{align*}
and finally we use the approximation $\p_{\ell,k}^\ep = (1+\eps(\chi_1)_\eps)\p^0_{\ell,k} + O_{H^1_{\per}}(\ep)$.

\subsubsection{Proof of Theorem~\ref{thm:phys_quantities2} (integrated density of states)} 
\label{sec:QoI2}

If $E \notin \sigma(H^0)$, we deduce from Theorem~\ref{thm:unif_CV_eigenvalues} that for $\eps \in \N^{-1}$ small enough, $\cN^\eps(E)=\cN^0(E)=\ell_*$, where $\ell_*=0$ if $E < \min \sigma(H^0)$ and $\ell_*$ is the unique integer such that 
\begin{align*}
\max_{k \in \Omega^*}E^0_{\ell_*,k} < E < \min_{k \in \Omega^*}E^0_{\ell_*+1,k}
\end{align*}
if $E > \min \sigma(H^0)$.

\medskip

Let us now assume that $E$ is a non-degenerate point of the spectrum of $H^0$. Define
\begin{align*}
&S^{0,E}_{\ell}:=\{k \in \Omega^* \, | \, E^0_{\ell,k} = E\}, \quad S^{0,E}:=\bigcup_{\ell \in \N^*} S^{0,E}_\ell, \\
&S^{\eps,E}_{\ell}:=\{k \in \Omega^* \, | \, E^\eps_{\ell,k} = E\}, \quad S^{\eps,E}:=\bigcup_{\ell \in \N^*} S^{\eps,E}_\ell.
\end{align*}
Since $\min_{k \in \Omega^*} E^0_{\ell,k} \to +\infty$ when $\ell \to \infty$, only a finite number of sets $S^{0,E}_\ell$ are non-empty. Let $1 \le \ell_- \le \ell_+ < \infty$ be such that 
$S^{0,E}_\ell \neq \emptyset$ if and only if $\ell_- \le \ell \le \ell_+$. Relying again on Theorem~\ref{thm:unif_CV_eigenvalues}, we also have for $\eps > 0$ small enough, $S^{\eps,E}_\ell \neq \emptyset$ if and only if $\ell_- \le \ell \le \ell_+$.

\medskip

Recall that for each $\ell \in \N^*$, the function $\R^d \ni k \mapsto E^0_{\ell,k} \in \R$ is $\L^*$-periodic, Lipschitz continuous, and real-analytic at each point $k$ such that $E^0_{\ell-1,k} < E^0_{\ell,k} < E^0_{\ell+1,k}$.
Let $\ell_- \le \ell \le \ell_+$. In view of the non-degeneracy conditions (C2) and (C3), the set $S^{0,E}_\ell$, seen as a subset of the torus $\R^d/\L^*$, is a compact ${(d-1)}$-dimensional submanifold of $\R^d/\L^*$, and there exists $\delta_0 > 0$, such that for all $0 < \delta \le \delta_0$ and all $k \in \Omega^*$, $\sigma(H^0_k) \cap [E-\delta,E+\delta]$ either is empty or consists of one non-degenerate eigenvalue of $H^0_k$. 
We define 
\begin{align*}
\omega_\delta^*:=\{k \in \overline{\Omega^*} \, | \, \sigma(H^0_k) \cap [E-\delta,E+\delta] \neq \emptyset\}.
\end{align*}
Note that in view of (C3), $\omega_\delta^*$ is a compact neighborhood of $S^{0,E}$ for each $\delta > 0$. Besides, from (C3) and the smoothness of $k \mapsto E_{\ell,k}$ away from eigenvalue crossings, there exists a positive constant $c_\nabla > 0$ such that
$$
\forall 0 < \delta \le \delta_0, \quad  \forall k \in \omega_\delta^*, \quad \forall \ell \mbox{ s.t. } E^0_{\ell,k} \in [E-\delta,E+\delta], \quad |\partial_k E^0_{\ell,k}| \ge c_\nabla.
$$
This implies that there exists a constant $M \in \R_+$ depending on $c_\nabla$, the Haussdorff measure of $S^{0,E}$, and its curvature, such that, up to reducing the value of $\delta_0$,
$$
\forall \; 0 < \delta \le \delta_0, \quad |\omega_\delta^*| \le M \delta.
$$
In view of Theorem~\ref{thm:unif_CV_eigenvalues}, there exists $\eps_0 > 0$ such that for all $0 < \eps \le \eps_0$, 
$$
\max_{\ell_- \le \ell \le \ell_+} \max_{k \in \Omega^*} |E^\eps_{\ell,k}-E^0_{\ell,k}| \le C_* \eps \le \delta_0 ,
$$
with $C_* := \max_{\ell_- \le \ell \le \ell_+} C_\ell$. It follows that for all $0 < \eps \le \eps_0$, 
\begin{align*}
\cN^\eps(E) &= \cN^0(E) + \sum_{\ell=\ell_-}^{\ell_+} \fint_{\Omega^*} \left(\1(E^\eps_{\ell,k} \le E)- \1(E^0_{\ell,k} \le E) \right) \, \d k \\
 &= \cN^0(E) + \sum_{\ell=\ell_-}^{\ell_+} \fint_{\Omega^*} \1(E^\eps_{\ell,k} \le E < E^0_{\ell,k}) \, \d k \\
 & \qquad \qquad\qquad\qquad\qquad\qquad - \sum_{\ell=\ell_-}^{\ell_+} \fint_{\Omega^*} \1(E^0_{\ell,k} \le E < E^\eps_{\ell,k}) \, \d k,
\end{align*}
with 
\begin{align*}
	\bigcup_{\ell=\ell_-}^{\ell_+} \{ k \in \Omega^* \, | \ E^\eps_{\ell,k} \le E < E^0_{\ell,k}\} &\subset \omega_{C_*\eps}^* \\
	\bigcup_{\ell=\ell_-}^{\ell_+} \{ k \in \Omega^* \, | \, E^0_{\ell,k} \le E < E^\eps_{\ell,k}\} &\subset \omega_{C_*\eps}^*,
\end{align*}
so that
\begin{align*}
\forall 0 < \eps \le \eps_0, \quad | \cN^\eps(E) - \cN^0(E) | \le \frac{|\omega_{C_*\eps}^* |}{|\Omega^*|} \le \frac{MC_*}{|\Omega^*|} \eps.
\end{align*}
Hence, $\cN^\eps(E)=\cN^0(E)+O(\eps)$. Likewise, using the results of Theorem~\ref{thm:eigenmode_nondeg} on degenerate eigenvalues with $\omega^*=\omega^*_{\delta_0}$, we have that there exists $\eps_0 > 0$ and $C \in \R_+$ such that for all $k \in \omega^*_{\delta_*}$, $0 < \eps \le \eps_0$, and $\ell_- \le \ell \le \ell_+$ such that $E^0_{\ell,k} \in [E-\delta_0,E+\delta_0]$, 
\begin{align*}
&\left| E^\eps_k - \left( E^0_{\ell,k} + \eps \int_\Omega V_1 |\psi^0_{\ell,k}|^2 \right) \right| \le C\eps^2, \\
& \qquad \left| E^\eps_k - E^{\eps,(2)}_{\ell,k} \right| \le C\eps^4, \quad 
\left| E^\eps_k - \widetilde E^{\eps,(2)}_{\ell,k} \right| \le C\eps^4.
\end{align*}
Observing that
\begin{multline*}
\cN^\eps(E) = \sum_{\ell \in \N^*} \fint_{\Omega^*} \1(E^{\eps,(2)}_{\ell,k} \le E) \, \d k + \sum_{\ell=\ell_-}^{\ell_+} \fint_{\Omega^*} \1(E^\eps_{\ell,k} \le E < E^{\eps,(2)}_{\ell,k}) \, \d k \\
- \sum_{\ell=\ell_-}^{\ell_+} \fint_{\Omega^*} \1(E^{\eps,(2)}_{\ell,k} \le E < E^\eps_{\ell,k}) \, \d k
\end{multline*}
and reasoning as above, we obtain 
$$
\cN^\eps(E) = \sum_{\ell \in \N^*} \fint_{\Omega^*} \1(E^{\eps,(2)}_{\ell,k} \le E) \, \d k + O(\eps^4),
$$
and similarly the remaining two estimates of $\cN^\eps(E)$.

\appendix

 \section{Other scalings} 
 \label{sec:appendix_other_scalings}

In this Appendix, we study operators of the form 
\begin{align*}
	H^{\ep,\zeta}_k := -\nabla_k^2 + \ep^{-\zeta} v\pa{x,\eps^{-1} x} + W(x) \quad \mbox{ for } \ep \in \N^{-1},
\end{align*}
and prove that only the scaling $\zeta = 1$ provides interesting features. 
We denote by $E^{\ep,\zeta}_{\ell,k}$ the $\ell^{\textup{th}}$ eigenvalue of $H^{\ep,\zeta}_k$.

\begin{proposition}
    \label{prop:other_scalings} Let $k \in \Omega^*$ and $\ell \in \N^*$.
    \begin{description}
\item[Case $\zeta < 1$] Assume that the $\ell^{\textup{th}}$ eigenmode $\big( \lambda^{0}_{\ell,k}, \vp^{0}_{\ell,k} \big)$ of $-\nabla_k^2 + W$, is non-degenerate. Then, for $\eps \in \N^{-1}$ small enough, $E^{\ep,\zeta}_{\ell,k}$ is non-degenerate, and if $\p^{\ep,\zeta}_{\ell,k}$ is an associated $\ell^{\textup{th}}$ $L^2_{\rm per}$-normalized eigenfunction of $H^{\ep,\zeta}_k$ such that $\big\langle \p^{\ep,\zeta}_{\ell,k},\vp^{0}_{\ell,k}\big\rangle_{L^2_\per} \ge 0$, then 
    \begin{align}\label{eq:relxi}
        \nor{\p^{\ep,\zeta}_{\ell,k} - \vp^{0}_{\ell,k}}{H^1_\per} \le c \ep^{1 - \zeta}, \qquad \ab{E^{\ep,\zeta}_{\ell,k} - \lambda^0_{\ell,k}} \le c \ep^{1-\zeta}.
    \end{align}
\item[Case $\zeta > 1$] Assume that $\na_y v \neq 0$. Then, for all $\ell \in \N$ and $k \in \Omega^*$, we have $E^{\ep,\zeta}_{\ell,k} \rightarrow -\infty$ when $\eps \rightarrow 0$.
\end{description}
\end{proposition}

 \begin{proof} $\,$
 
 \noindent
 {\em Case $\zeta < 1$.} Let $h^0_{k} := -\na_k^2 + W$. It holds
\begin{align*}
    \pa{H^{\ep,\zeta}_k -z} \pa{1 + \ep^{2-\zeta} (\chi_1)_\eps} = \pa{h^0_k -z} + \ep^{1-\zeta} \cS_\eps + \ep^{2\pa{1-\zeta}} V_0
\end{align*}
where
\begin{multline*}
    \cS_\eps := -2 \pa{(\na_x \cdot \na_y) \cu}_\eps  + \ep^{1-\zeta} \pa{\pa{\Delta_y \eta}_\eps -2 \pa{\na_y \cu}_\eps \cdot \na_k} \\
    + \eps \pa{-\pa{\Delta_x \cu}_\eps + (\chi_1)_\eps h^0_k} + \ep^{2-\zeta} \pa{-\pa{\Delta_x \cu}_\eps - 2 \pa{\na_x \cu}_\eps \cdot \na_k}
\end{multline*}
is a highly-oscillating operator. We deduce that
\begin{multline*}
    \pa{H^{\ep,\zeta}_k -z}^{-1} - \pa{1 + \ep^{2-\zeta} (\chi_1)_\eps} \pa{h^0_k -z}^{-1} \\
    = - \ep^{1-\zeta} \pa{H^{\ep,\zeta}_k -z}^{-1} \pa{\cS_\eps + \ep^{1-\zeta}V_0} \pa{h^0_k -z}^{-1}.
\end{multline*}
With a similar proof as in the $\zeta = 1$ case,
\begin{align*}
    \nor{\pa{H^{\ep,\zeta}_k -z}^{-1} - \pa{1 + \ep^{2-\zeta} (\chi_1)_\eps} \pa{h^0_k -z}^{-1}}{L^2_\per \rightarrow H^1_\per} \le c \ep^{1-\zeta},
\end{align*}
hence we also have
\begin{align*}
    \nor{\pa{H^{\ep,\zeta}_k -z}^{-1} - \pa{h^0_k -z}^{-1}}{L^2_\per \rightarrow H^1_\per} \le c \ep^{1-\zeta},
\end{align*}
where $c$ does not depend on $\eps$ or $k$. In particular, by using the Cauchy formula, we obtain \eqref{eq:relxi}.

\medskip

 \noindent
 {\em Case $\zeta > 1$.} We recall that $\p^\ep_{\ell,k}$ denotes a normalized eigenvector of $H^\ep_k$ (that is for $\zeta = 1$) associated with the $\ell^{\rm th}$ eigenvalue $E^\eps_{\ell,k}$. We have
\begin{align*}
    \ps{\p^\ep_{\ell,k}, H^{\ep,\zeta}_k \p^\ep_{\ell,k}} &= E^\ep_{\ell,k} + \f{1-\ep^{\zeta -1}}{\ep^{\zeta -1}} \int_\Omega \f{v_\ep}\eps \ab{\p^\ep_{\ell,k}}^2  \\
							  &= E^\ep_{\ell,k} + \f{1-\ep^{\zeta -1}}{\ep^{\zeta -1}} \pa{2 \int_\Omega V_0 \ab{\p^0_{\ell,k}}^2 + O(\ep)},
\end{align*}
where we used \eqref{eq:osc_int} to obtain the second equality. We have $V_0 \le 0$, and $V_0 \neq 0$ since $\na_y v \neq 0$, so by unique continuation \cite{JerKen85}, we know that $\ab{\acs{x \in \Omega \; \bigr\vert \; \p^0_{\ell,k}(x) = 0}} = 0$ hence $\int_\Omega V_0 |\p^0_{\ell,k}|^2 < 0$ and $\big\langle\p^\ep_{\ell,k}, H^{\ep,\zeta}_k \p^\ep_{\ell,k}\big\rangle_{L^2_\per} \rightarrow -\infty$. By the min-max characterization of eigenvalues, the $\ell^{\textup{th}}$ eigenvalue of $H^{\ep,\zeta}_k$ goes to $-\infty$ when $\eps \rightarrow 0$.
	  \end{proof}

\section{Higher-order formal two-scale expansion}
\label{sec:detail_ansatz}

In this Appendix, we formally derive the next order terms in the expansion. We go back to~\eqref{eq:eps1}, and continue the computations. After some manipulations, using in particular the identity
\begin{align*}
    & \left(- \na_{x,k}^2 +W(x) - E^0_{\ell, k} \right) \left[\chi_1 \p_0(x) +U_1(x) \right]
    \\
    & \quad 
    =  \left( H^0_k - E^0_{\ell, k} \right)  \left[\chi_1 \p_0(x) +U_1(x) \right] - V(x) \chi_1 \psi_0(x) - V(x) U_1(x) \\
    & \quad = (\nabla_{x}^2 \chi_1) \psi_0(x) + 2 \nabla_x \chi_1 \cdot \nabla_{k} \psi_0(x) + (E_1 - V_1) \psi_{\ell, k}^0 - V(x) \chi_1 \psi_0(x)\\
    & \qquad\qquad\qquad\qquad\qquad\qquad\qquad\qquad\qquad\qquad\qquad\qquad\qquad - V(x) U_1(x),
\end{align*}
we get
\begin{align*}
    \Delta_y \p_3 = &\pa{-\na_{x,k}^2+W(x)- E^0_{\ell,k}} \pa{\chi_1 \p^0_{\ell,k}(x)} - V_1(x) \p^0_{\ell,k}(x) - V(x)U_1(x) \\
    & + \pa{v-2 \na_{x,k} \cdot  \na_y } \p_2,
\end{align*}
and thus, using the expression \eqref{eq:psi2} of $\psi_2$, and introducing the functions 
$\beta,\delta,\xi \in L^\infty_{\rm per,0}(\Omega \times \Omega)$ 
and the field $B_1 \in L^\infty_{\rm per,0}(\Omega \times \Omega)^d$, uniquely defined by 
\begin{align*}
    &   \Delta_y \beta = 4\au, \qquad  \Delta_y \delta = \eta , \qquad  \Delta_y \xi= v\chi_2 - V_1, \qquad  \Delta_y B_1 = vA_2 ,
\end{align*}
we obtain
\begin{multline*}
	\Delta_y \p_3 = \pa{-\na_{x,k}^2 +W- E^0_{\ell,k}} \pa{\p^0_{\ell,k} \Delta_y \au} + \p^0_{\ell,k} \Delta_y \xi + U_1 \Delta_y \eta  \\
		      + \pa{\Delta_y B_1}\cdot \na_{k} \p^0_{\ell,k} + U_2 \Delta_y \chi_1  - 2 (\na_{x,k} \cdot  \na_y)  \pa{\p^0_{\ell,k} \Delta_y \delta} \\
		      +  (\na_{x,k} \cdot  \na_y) (\na_{x,k} \cdot  \na_y) \pa{\p^0_{\ell,k} \Delta_y \beta} -2 (\na_{x,k} \cdot  \na_y) \pa{U_1 \Delta_y \au}.
\end{multline*}
Therefore,
\begin{multline*}
    \p_3 = \pa{-\na_{x,k}^2 +W- E^0_{\ell,k}} \pa{\p^0_{\ell,k}  \au} + \xi \p^0_{\ell,k}  +B_1 \cdot \na_{k} \p^0_{\ell,k}   - 2 (\na_{x,k} \cdot  \na_y)  \pa{\p^0_{\ell,k} \delta}\\
    +  (\na_{x,k} \cdot  \na_y)(\na_{x,k} \cdot  \na_y) \pa{\p^0_{\ell,k} \beta} -2 (\na_{x,k} \cdot  \na_y) \pa{U_1  \au}  + U_1  \eta+ U_2 \chi_1+ U_3,
\end{multline*}
where the function $U_3$ only depends on the macroscopic variable $x$.
Rearranging the term
\begin{align*}
    &\pa{-\na_{x,k}^2+W - E^0_{\ell,k}} \pa{\au \p^0_{\ell,k}} \\
    & \qquad = \au \pa{-\na_{k}^2 +W- E^0_{\ell,k}} \p^0_{\ell,k} - 2 \na_x \au\cdot \na_{k} \p^0_{\ell,k}  - \p^0_{\ell,k} \Delta_x \au \\
    & \qquad = - \pa{ V \au + \Delta_x \au} \p^0_{\ell,k}- 2  \na_x \au\cdot \na_{k} \p^0_{\ell,k},
\end{align*}
and observing that for all $f:\R^d \times \R^d \to \C$ and $\phi:\R^d \to \C$ regular enough,
\begin{align*}
    (\na_{x,k} \cdot  \na_y) \pa{ f\phi} &= \na_y f \cdot \na_{k} \phi +( (\nabla_x \cdot  \na_y) f)  \phi, \\
    \pa{\na_{x,k} \cdot \na_y}^2 \pa{f\phi} &=  \pa{\na_y^{\otimes 2} f} : \pa{\na_k^{\otimes 2}\phi}  + 2 \pa{\na_y(\na_x\cdot\na_y)f} \cdot \na_k\phi\\
					    &\qquad\qquad\qquad\qquad\qquad\qquad\qquad\qquad + \pa{\pa{\na_{x} \cdot \na_y}^2 f}  \phi,
\end{align*}
we finally get
\begin{multline*}
    \p_3 =  \chi_3 \p^0_{\ell,k}(x)  + A_3 \cdot \nabla_k\p^0_{\ell,k}(x) + M_3 : \pa{\na_k^{\otimes2} \p^0_{\ell,k}}(x) \\
     + \chi_2 U_1(x)  +A_2  \cdot \na_{k} U_1(x) + \chi_1 U_2(x)  + U_3(x),
\end{multline*}
with 
\begin{align*}
    \chi_3&:=-V \au - \Delta_x \au + \xi  - 2 (\nabla_x \cdot  \na_y)  \delta  +  \pa{\na_{x} \cdot  \na_y}^2 \beta , \\
    A_3&:= -2\na_x \au +B_1 - 2\na_y \delta +2 \na_y \pa{\na_{x} \cdot  \na_y} \beta, \\
    M_3&:= \na_y^{\otimes 2} \beta.
\end{align*}

\medskip

\noindent
\textbf{Terms in $\ep^{2}$.} Using \eqref{eq:E0_psi0}, the equation satisfied at order $\ep^2$ is
\begin{multline*}
    -\Delta_y \p_4 + \pa{v-2 \na_{x,k} \cdot  \na_y} \p_3 +\pa{- \na_{x,k}^2 +W - E^0_{\ell,k}} \p_2 - E_2 \p^0_{\ell,k} - E_1 \p_1 \\
    = 0,
\end{multline*}
and averaging in $y$ over $\Omega$ yields
\begin{align}
    & \qquad E_2 \p^0_{\ell,k}(x) = -E_1 U_1(x) + \pa{-\na_{k}^2 - E^0_{\ell,k}} U_2(x)  + \fint_\Omega v(x,y) \p_3(x,y) \, \d y  \nonumber \\
    & = \hspace{-0.1cm} \left[ \pa{V_1 -E_1} U_1 + \pa{H^0_k - E^0_{\ell,k}} U_2  + K \p^0_{\ell,k} +  L \cdot \na_{k} \p^0_{\ell,k} +  \na^{\otimes 2} \p^0_{\ell,k} :  M\right](x) \label{eq:eq_on_U2}
\end{align}
where the macroscopic scalar field $K$, vector field $L$, and matrix field $M$, are defined by
\begin{align}
    K(x) &:= \fint_\Omega v(x,y) \chi_3(x,y) \, \d y ,  \label{eq:def_K}\\
    L(x) & :=  \fint_\Omega v(x,y)A_3(x,y) \, \d y,  \label{eq:def_L}\\
    M(x) & :=  \fint_\Omega v(x,y) \na_x^{\otimes 2} \beta(x,y) \, \d y \nonumber \\
	 &= 4 \pa{\fint_\Omega \pa{\partial_{y_j} \na \gamma}(x,y) \cdot (\partial_{y_m} \na \gamma)(x,y) \, \d y }_{1 \le j,m \le d}. \label{eq:def_M}
\end{align}
Note that $M(x)$ is symmetric and semidefinite positive for a.e. $x$. Using \eqref{cro}, we have 
\begin{align*}
    \ps{\p^0_{\ell,k},\pa{V_1 -E_1} U_1} = - \ps{U_1,\pa{H^0_k-E^0_{\ell,k}}U_1},
\end{align*}
and integrating \eqref{eq:eq_on_U2} against $\overline{\p^0_{\ell,k}}$ on $\Omega$,
\begin{align}\label{eq:E2}
    E_2 \hspace{-0.05cm}= \hspace{-0.05cm} - \ps{U_1,\pa{H^0_k -E^0_{\ell,k}}U_1}\hspace{-0.05cm} +\hspace{-0.05cm} \int_\Omega \ab{\p^0_{\ell,k}}^2 \pa{K - \nabla_k \cdot L} +  \int_\Omega \overline{\p^0_{\ell,k}} M \hspace{-0.05cm}:\hspace{-0.05cm} \na^{\otimes 2} \p^0_{\ell,k}.
\end{align}
The last term is real because $M$ is a symmetric matrix and hence
\begin{align*}
    &\int_\Omega \overline{\p^0_{\ell,k}} M : \na^{\otimes 2} \p^0_{\ell,k} \\
    & \qquad = - \sum_{j,m=1}^d \left( \int_\Omega  M_{jm} \pa{\partial_{x_j} \overline{\p^0_{\ell,k}}} \partial_{x_m} \p^0_{\ell,k} +  \int_\Omega  \overline{\p^0_{\ell,k}} \pa{\partial_{x_m} \p^0_{\ell,k}} \pa{\partial_{x_j} M_{jm}} \right) \\
    & \qquad = - \re \left( \sum_{j,m=1}^d \int_\Omega M_{jm} \pa{\partial_{x_j} \overline{\p^0_{\ell,k}}} \partial_{x_m} \p^0_{\ell,k} + \f{1}{2} \int_\Omega \ab{\p^0_{\ell,k}}^2 \partial_{x_j,x_m}^2 M_{jm} \right).
\end{align*}
We deduce from \eqref{eq:eq_on_U2} that 
\begin{align*}
& (H^0_k - E^0_{\ell,k}) U_2 \\
& \qquad =  \Big( \pa{E_2 -K} \p^0_{\ell,k} - L\cdot \na_{k} \p^0_{\ell,k} 
- M: \na^{\otimes 2} \p^0_{\ell,k} - \pa{V_1 -E_1} U_1\Big),
\end{align*}
which is orthogonal to $\p^0_{\ell,k}$ in $L^2_{\rm per}(\Omega)$, and using \eqref{eq:intermediate_norm}, we get $\big\langle \psi^0_{\ell,k},U_2 \big\rangle_{L^2_{\rm per}}=0$. Thus, 
$$
U_2 = - (H^0_k - E^0_{\ell,k})_\perp^{-1} \left(K \p^0_{\ell,k} + L\cdot \na_{k} \p^0_{\ell,k} 
+ M: \na^{\otimes 2} \p^0_{\ell,k} + \pa{V_1 -E_1} U_1\ \right).
$$

\begin{remark}
    Using the hypotheses and notation of Theorem~\ref{thm:unif_CV_eigenvalues}, the previous formal computations lead us to conjecture that
    \begin{align*}
\nor{\p^\ep_{\ell,k} -  \widetilde{\cM}^{\eps, (3)}_{E_{\ell, k}^0, k} \p^0_{\ell,k}}{L^2_{\rm per}} \le c \ep^3,
    \end{align*}
    with
    \begin{multline*}
       \widetilde{\cM}^{\eps, (3)}_{z, k}  := \Mz  \\
	- \ep^2 \pa{H^0_k - z}^{-1}_{\perp} \Big( K + L\cdot \na_{k}  
        + M:\na^{\otimes 2} - \pa{V_1 -E_1}\pa{H^0_k - z}^{-1}_{\perp}\pa{V_1 -E_1} \Big),
    \end{multline*}
    and $E_1$, $E_2$, $K$, $L$, $M$ respectively defined by \eqref{eq:def_E1}, \eqref{eq:E2}, \eqref{eq:def_K}, \eqref{eq:def_L}, and \eqref{eq:def_M}. Since we are mainly interested in $H^1_{\rm per}$ (energy) estimates, we will not further investigate this conjecture for the sake of brevity.
\end{remark}

\bibliographystyle{my-alpha}
\bibliography{cances_garrigue_gontier}
\end{document}